\documentclass[a4paper,12pt]{article}

\usepackage{bm}
\usepackage{graphicx}
\usepackage{amsbsy}
\usepackage{amsmath}
\usepackage{amsfonts}
\usepackage{amsthm}
\usepackage{listings}
\usepackage{hyperref}  
\usepackage{blindtext}
\usepackage{titlesec}
\usepackage{epigraph}

\usepackage{CJKutf8}

\usepackage{geometry}
\usepackage[T1]{fontenc}      

\graphicspath{ {./images/} }

\numberwithin{equation}{section}
\numberwithin{figure}{section}

\usepackage{cite}

\usepackage{amssymb}
\usepackage{tikz-cd}
\usepackage{tikz}
\usepackage{mathrsfs}
\usepackage{mathtools}

\usepackage{pgfplots}
\pgfplotsset{compat=1.18}

\def\mg{\mathfrak{g}}

\def\tr{\text{tr}}

\def\Ad{\text{Ad}}

\def\tx{\text}
\def\bdy{\scalebox{0.5}{$\bigcirc$}}
\def\cft{\tx{CFT}}
\def\ads{\tx{AdS}}
\def\PSL{\tx{PSL}(2, \mathbb{R})}
\def\psl{\mathfrak{sl}(2, \mathbb{R})}
\def\G{G_{\tx{N}}}

\theoremstyle{definition}

\newtheorem{proposition}{Proposition}[section]

\newcommand{\RNum}[1]{\uppercase\expandafter{\romannumeral #1\relax}}

\usepackage{authblk}
\title{Filtering out Erratic Observables:
\\
Wormholes from Gauging Nonlocal Symmetries}

\usepackage{hyperref}

\renewcommand{\thefootnote}{\fnsymbol{footnote}}

\begin{document}
\begin{CJK*}{UTF8}{gkai}

\author{Qi-Feng Wu (吴奇峰)\thanks{qifeng.wu@ugent.be}}

\affil{Department of Physics and Astronomy

Ghent University, Krijgslaan, 281-S9, 9000 Gent, Belgium}

\maketitle


{\small  \noindent 
\begin{center} 
\textbf{Abstract}
\end{center} 
The wormhole contribution to the gravitational path integral may be interpreted as smooth
remnant of correlations among the erratic large-$N$ behaviors of dual CFTs. In this work, we investigate this idea in (2+1)-dimensional gravity. We show that one-sided boundary gravitons are intrinsically incomplete in the sense that the associated observable algebra has a nontrivial center regardless of choices of boundary conditions.  Based on asymptotic symmetries, we bootstrap a general Poisson bracket to construct completions of the boundary gravitons. In the simplest completion, the commutant of the boundary graviton observable algebra is given by an observable algebra of monodromy data which we interpret as an effective description of one-sided black holes. We show that, to describe Lorentzian multi-boundary wormholes, only the monodromy data with a positivity restriction is needed. The positivity restriction results in emergent erratic large-$N$ behaviors for some observables. We filter out the erratic observables by restricting to a subspace on which they act trivially. The monodromy observables generate nonlocal symmetries lack of corresponding local currents. We show that gauging the nonlocal symmetries is equivalent to filtering out the erratic observables. For one CFT, gauging the nonlocal symmetries at the quantum level removes all black hole states. Filtering the partition function of CFTs leads to an apparent ensemble averaging. For two CFTs, a Hilbert subspace describing wormholes survives after gauging global part of the nonlocal symmetries. The filtered partition function of the two CFTs is an ensemble average over quantum gates entangling the monodromy degrees of freedom of the two CFTs. The correlation between the erratic observables of the two CFTs is preserved, which contributes to the filtered partition function as a wormhole term. 
}

\newpage

\tableofcontents
\newpage

\setcounter{footnote}{0}
\renewcommand{\thefootnote}{\arabic{footnote}}

\section{Introduction}
\label{sec: Introduction}

A quantum mechanical description of horizons is not yet fully established. It leads to various puzzles, e.g.~\cite{Hawking:1976ra,  Mathur:2009hf, Almheiri:2012rt, Marolf:2013dba}. As a novel interpretation of black hole horizons, an operational limit on physical observers was introduced in~\cite{Harlow:2013tf} and was further developed, e.g., in~\cite{Susskind:2015toa, Aaronson:2016vto, Bouland:2019pvu, Kim:2020cds}. A central idea is that, for observers outside the horizon, the operational complexity required to access information behind the black hole horizon may be extremely high such that the interior is effectively inaccessible, and the horizon may be interpreted as a geometric manifestation of the operational limit. Similar ideas also appear in~\cite{Papadodimas:2012aq, Papadodimas:2013jku, Engelhardt:2017aux,Engelhardt:2018kcs, Bousso:2019dxk,Brown:2019rox,Engelhardt:2021qjs,Akers:2022qdl}. In the context of AdS/CFT correspondence~\cite{Maldacena:1997re, Gubser:1998bc, Witten:1998qj}, this operational limit may be interpreted as a form of the information loss~\cite{Maldacena:2001kr} of boundary theory in the large $N$ limit~\footnote{$N$ is a schematic parameter characterizing the effective degrees of freedom and the coupling constant of the boundary theory in AdS/CFT. The Newton's constant $\G$ scales as $N^\gamma$ for some negative exponent $\gamma$.}. It can be reformulated as the phenomenon that the ``single trace'' operators or observables of low complexity do not generate the full boundary observable algebra in the large $N$ limit~\cite{Leutheusser:2021frk,Leutheusser:2021qhd  , Furuya:2023fei, Engelhardt:2023xer,Gesteau:2024rpt,Liu:2025cml, Liu:2025krl}.

Based on the operational limit, the algebra $\mathcal{A}_\cft$ of boundary observables is divided into the set $\mathcal{A}_\cft^{\tx{simple}}$ of simple observables (often called ``single trace'' operators) and the set $\mathcal{S}_\cft^{\tx{complex}}$ of complex observables. $\mathcal{A}^{\tx{simple}}_\cft$ becomes an algebra in the large $N$ limit because of the information loss. In other words, in the large $N$ limit, (finite) compositions of simple operations do not produce intrinsically complex operations. The operational limit in the boundary theory is expected to be dual to the presence of the horizon in the bulk, so the algebra $\mathcal{A}_\ads^{\tx{exterior}}$ of observables outside the horizon may be dual to $\mathcal{A}_\cft^{\tx{simple}}$, i.e.
\begin{equation}
\label{eq: exterior algebra  = simple algebra}
    \mathcal{A}_\ads^{\tx{exterior}} \cong \mathcal{A}_\cft^{\tx{simple}} .
\end{equation}
One may expect that the algebra $\mathcal{A}_\ads^{\tx{interior}}$ of observables behind the horizon is dual to $\mathcal{S}_\cft^{\tx{complex}}$, but $\mathcal{S}_\cft^{\tx{complex}}$ is not an algebra. This is because the composition of a simple operation and a complex operation is generically complex, which implies that the composition of two complex operations can be simple. So $\mathcal{S}_\cft^{\tx{complex}}$ essentially generates the full algebra $\mathcal{A}_\cft$. We hence expect that the boundary observable subalgebra $\mathcal{A}_\cft^{\tx{complex}}$ dual to $\mathcal{A}_\ads^{\tx{interior}}$ is composed of part of $\mathcal{S}_\cft^{\tx{complex}}$ and part of $\mathcal{A}_\cft^{\tx{simple}}$. 

Besides the operational limit, part of the information loss can also be understood as arising from the fact that some observables do not survive the large $N$ limit~\cite{Witten:2021jzq, Schlenker:2022dyo, Kudler-Flam:2025cki,Kudler-Flam:2026nzz}. The interior observable algebra $\mathcal{A}_\cft^{\tx{complex}}$ are accordingly divided into two parts: the set~$\mathcal{A}_\cft^{\tx{smooth}}$ of observables that survives the large $N$ limit, and the set~$\mathcal{S}_\cft^{\tx{erratic}}$ of observables that does not. The algebra of a set of observables does not (resp. does) survive the large $N$ limit if it depends on $N$ erratically (resp. smoothly). A (finite) sum of (finite) products of smooth observables is smooth, so $\mathcal{A}_\cft^{\tx{smooth}}$ is an algebra. $\mathcal{A}_{\cft}^{\tx{smooth}}$ may govern the semiclassical physics behind the horizon since $\mathcal{S}_\cft^{\tx{erratic}}$ has no large $N$ limit. While $\mathcal{S}_\cft^{\tx{erratic}}$ may capture complementary aspects. It was shown that wormhole contributions to the gravitational path integral admit an ensemble averaging interpretation~\cite{Saad:2019lba}. See \cite{Coleman:1988cy, Giddings:1988cx, Cotler:2016fpe} for earlier developments, \cite{Stanford:2019vob, Saad:2019pqd,Marolf:2020xie, Blommaert:2020seb, Post:2022dfi} for further discussions on 2d gravity, and \cite{Cotler:2020ugk,Belin:2020hea, Afkhami-Jeddi:2020ezh,Chandra:2022bqq, Chandra:2022fwi,DiUbaldo:2023qli, Belin:2023efa, Jafferis:2025vyp, Geng:2025efs, Hartman:2025ula, Yu:2026gdf} for higher dimensional generalizations. See also~\cite{Chandra:2023rhx,Klinger:2025tvg,Klinger:2026kqj} for related discussions.

In~\cite{Schlenker:2022dyo}, it was proposed that the erratic $N$-dependence of observables is essential for the ensemble-averaging behavior of black holes: if an observable $\mathcal{O}_N$ depends erratically on $N$ in the large $N$ limit, then $\mathcal{O}_N$ for nearing values of $N$ can be effectively regarded as independent samples in an ensemble. See~\cite{Kudler-Flam:2025cki, Kudler-Flam:2026nzz } for further developments. In \cite{Liu:2025ikq}, it was systematically formulated and argued that the smooth remnants of correlations among the erratic large-$N$ behaviors of boundary CFTs are manifested as wormholes on the gravity side. The map extracting the smooth remnants in the large $N$ limit is called a large-$N$ filter. More precisely, suppose that in the large $N$ limit the partition function $Z_\cft[M]$ of the boundary CFT on a Euclidean manifold $M$ decomposes into two parts 
\begin{equation}
    Z_\cft[M] = Z_\cft^{\tx{smooth}}[M] + Z_\cft^{\tx{erratic}} [M]
\end{equation}
with $Z_\cft^{\tx{smooth}}[M]$ (resp. $Z_\cft^{\tx{erratic}} [M]$) depending on $N$ smoothly (resp. erratically). The large-$N$ filter $\mathbb{F}$ projects $Z_\cft[M]$ onto the smooth part, i.e.
\begin{equation}
\label{eq: filtered one CFT partition function}
    \mathbb{F}\{ Z_\cft[M] \} = Z_\cft^{\tx{smooth}} [M].
\end{equation}
It was argued that the large-$N$ filter extracts the wormhole contribution from the erratic correlations, i.e.
\begin{equation}
    \mathbb{F} \left\{ \, Z_\cft^{\tx{erratic}}[M_1] Z_\cft^{\tx{erratic}}[M_2]\, \right\} = Z_{\tx{gravity}}^{\tx{wormhole}} [M_1, M_2],
\end{equation}
where $Z_{\tx{gravity}}^{\tx{wormhole}} [M_1, M_2]$ is the gravitational path integral over the wormholes with boundary $M_1 \cup M_2$. Then the filtered partition function of two disjoint CFTs decomposes into a disconnected part and a connected wormhole part:
\begin{equation}
\label{eq: F( Z Z )}
    \mathbb{F} \left\{\, Z_\cft[M_1] Z_\cft [M_2]  \, \right\} = \mathbb{F} \{ Z_\cft[M_1] \} \,\mathbb{F}\{ Z_\cft [M_2]   \} +  Z_{\tx{gravity}}^{\tx{wormhole}} [M_1, M_2].
\end{equation}
See~\cite{Liu:2026fnd} for further developments.

In this work, we investigate the aforementioned ideas in (2+1)-dimensional gravity. We show that the CFT dual of a one-sided black hole possesses a subalgebra of observables with erratic large-$N$ behaviors. We filter out the erratic observables by projecting onto a subspace on which they act trivially. This generalizes the notion of the large-$N$ filter to the level of quantum states. The CFT partition function with respect to the filtered quantum states is equal to an ensemble-averaged partition function with respect to all states. The ensemble average is over nonlocal symmetry transformations. For two CFTs, Eq.~\eqref{eq: F( Z Z )} is realized as Eq.~\eqref{eq: realizing large N filter}. 
\\

\textbf{Summary and outline:}

With a negative cosmological constant, the (2+1)-d pure gravity admits the BTZ black hole solutions~\cite{Banados:1992wn}. In (2+1)-d pure gravity, there are no local propagating degrees of freedom except at the asymptotic boundary. They are called boundary gravitons. In Section~\ref{sec: Intrinsic Incompleteness of One-Sided Boundary Gravitons}, we show that one-sided boundary gravitons describing a set of general BTZ black holes do not form a complete system. They must be a subsystem of a complete system. In other words, the classical observable algebra $C^\infty (\mathcal{P}_{\bdy})$ of the boundary gravitons has a nontrivial center.  From the boundary perspective, this means some CFT degrees of freedom are missing. The incompleteness of boundary gravitons is intrinsic in the sense that it is irrelevant of boundary conditions. In Section~\ref{sec: Bootstrapping Completions from Asymptotic Symmetries}, we illustrate a method to bootstrap a general Poisson bracket of Wilson lines to construct potential completions of the boundary gravitons. The Poisson bracket is parametrized by a function $r$. In Section~\ref{sec: An Effective Observable Algebra of One-Sided Black Holes}, we show that the simplest completion (i.e. $r$ being a constant) reproduces the observable algebra $C^\infty (\mathcal{P}_\odot)$ of the chiral WZNW model~\cite{Gawedzki:1990jc, Alekseev:1990vr, Alekseev:1991wq} and the observable algebra of the Chern-Simons theory defined on a punctured disc~\cite{Mertens:2025ydx}. The chiral WZNW filed hence can be identified with the Chern-Simons Wilson line connecting the puncture to the boundary of the disc. The observable algebra $C^\infty (\mathcal{P}_{\bdy})$ of boundary gravitons has a nontrivial commutant in the simplest completion $C^\infty (\mathcal{P}_\odot)$. The commutant is the algebra $C^\infty (\mathcal{P}_\bullet)$ of the puncture monodromy. In our construction, the boundary graviton observable algebra $C^\infty (\mathcal{P}_{\bdy})$ is identified with the algebra $C^\infty (\mathcal{P}_\ads^{\tx{exterior}})$ of classical observable outside the horizon. By the commutant property, the puncture monodromy may be interpreted as an effective description of the interior degrees of freedom. The commutant property will be proved in Section~\ref{sec: Gauging Nonlocal Symmetries}.

In Section~\ref{sec: Interior Phase Spaces from Gauging Asymptotic Symmetries}, we show that the asymptotic symmetry is effectively a gauge symmetry for the effective interior degrees of freedom since it acts trivially. The associated physical phase subspaces are identified as a coset space. These physical phase subspaces are classified in Section~\ref{sec: Classification of Interior Phase Spaces}. In the spirit of the black hole complementarity~\cite{Susskind:1993if,Lowe:1995ac}, the puncture monodromy admits another interpretation complementary to the infalling observer perspective. From the perspective of an observer outside the horizon, the effective interior degrees of freedom can be alternatively interpreted as horizon edge modes proposed in~\cite{Mertens:2022ujr}. For general discussions of edge modes associated with local subsystems in gauge theories and gravity, see \cite{Donnelly:2016auv}. In Section~\ref{sec: Positivity Restrictions Lorentzian Wormholes}, We show that only the monodromy data with a positivity restriction is needed to describe Lorentzian multi-boundary wormholes~\cite{Aminneborg:1997pz, Aminneborg:1998si, Brill:1998pr, barbot2008causal, Barbot:2005qk}. In Section~\ref{ref: Loss of Erratic Information in the Large $N$ Limit}, we show that this positivity restriction automatically extends the quantized observable algebra of the puncture monodromy such that a set of observables with erratic large-$N$ behaviors emerges.

In Section~\ref{sec: Nonlocal Symmetries of One-Sided Black Holes}, we show that the effective interior degrees of freedom possess a novel nonlocal symmetry. It is nonlocal in the sense that it does not correspond to a local current but still corresponds to a global charge. To explain the notion of the nonlocal symmetry, we reformulate the notion of ordinary Lie group symmetries in Section~\ref{sec: Reformulation of Lie Group Symmetries}. This nonlocal symmetry acts trivially on the boundary gravitons. In Section~\ref{sec: Gauging Nonlocal Symmetries}, we show that the commutant of $C^\infty (\mathcal{P}_{\bdy})$ is the observable algebra $C^\infty (\mathcal{P}_\bullet)$ of the monodromy associated with the puncture. 

In Section~\ref{sec: Emergent Wormholes from the Large-$N$ Filter}, we show that filtering out the erratic observables emergent from the positivity restrictions is equivalent to gauging the nonlocal symmetries. In Section~\ref{sec: Filtering One CFT}, we show that filtering out erratic observables of one CFT leads to a filtered Hilbert space describing the global AdS$_3$ vacuum with boundary gravitons. In Section~\ref{sec: Apparent Ensemble Averaging over Entanglers}, we show that filtering out global part of the erratic observables of the two CFTs leads to a filtered Hilbert space describing both of two disconnected AdS$_3$ vacuum and wormhole spacetimes. We show that gauging the nonlocal symmetries gives rise to an ensemble averaged partition function for CFTs. After filtering out the global erratic behaviors of two CFTs, the smooth remnant of correlations of erratic behaviors between
the two CFTs survive as the wormhole contribution to the filtered partition function. So gauging the nonlocal symmetry can be interpreted as the large-$N$ filter.

\section{Completions of Boundary Gravitons}
\label{sec: Completion of Boundary Gravitons}

\subsection{Intrinsic Incompleteness of One-Sided Boundary Gravitons}
\label{sec: Intrinsic Incompleteness of One-Sided Boundary Gravitons}

The general one-sided solution to the (2+1)-dimensional vacuum Einstein equation with cosmological constant $-1$
\begin{equation}
\label{eq: Vacuum Einstein}
    R_{\mu \nu}- \frac{1}{2}R \,g_{\mu \nu} = -g_{\mu \nu}
\end{equation}
with asymptotically AdS$_3$ boundary conditions is given by the
Ba\~{n}ados metric~\cite{Navarro-Salas:1998fgp, Banados:1998gg} 
\begin{align}
    ds^2 = r^{-2}
    dr^2 - ( r dx^+ - 4 G_{\tx{N}} r^{-1} \overline{T} dx^{-} ) ( r dx^- - 4 G_{\tx{N}} r^{-1} T dx^{+} ) .
    \label{eq: Banados metric}
\end{align}
The lightcone coordinates $x^\pm$ are periodically identified as $(x^+, x^-) \cong (x^+ +  2\pi, x^- -  2\pi )$. $T$ (resp. $\overline{T}$) is an arbitrary periodic function of $x^+$ (resp. $x^-$) with period $2\pi$. The phase space $\mathcal{P}_{\scalebox{0.5}{$\bigcirc$}}$ of boundary gravitons is then parametrized by $T$ and $\overline{T}$.

Heuristically, the boundary gravitons cannot be all degrees of freedom of a one-sided black hole spacetime. In this subsection, we rephrase and justify this statement from the perspective of the observable algebra of the boundary gravitons. Then we provide a straightforward argument based on the first order formulation of (2+1)-d gravity. This subsection is a preparation for the discussion on completions of the boundary graviton observable algebra in Section~\ref{sec: Bootstrapping Completions from Asymptotic Symmetries}.

The diffeomorphisms acting nontrivially on the asymptotic boundary and preserving the asymptotically AdS$_3$ boundary conditions are of the chirally factorized form
\begin{equation}
\label{eq: asymptotic symmetry transformation}
    x^+ \to F(x^+), \quad x^- \to \bar{F}(x^-)
\end{equation}
near the asymptotic boundary. So the asymptotic symmetry group is $\tx{Diff}(S^1) \times \tx{Diff}(S^1)$. When realized by Poisson brackets, the asymptotic symmetry algebra $\mathfrak{diff}(S^1) \times \mathfrak{diff}(S^1)$ is centrally extended~\cite{Brown:1986nw}. In terms of Fourier modes of $T$ and $\overline{T}$ in Eq.~\eqref{eq: Banados metric},
\begin{align}
    L_n &\equiv  \frac{1}{2 \pi} \int_0^{2 \pi}  e^{in x^+} T dx^+,
    \label{eq: L_n def}
    \\
    \bar{L}_n &\equiv \frac{1}{2 \pi}  \int_0^{-2 \pi}  e^{in x^-} \overline{T} dx^-,
\end{align}
the Poisson bracket is given by
\begin{align}
    \{ L_m, L_n \} &= i(m-n) L_{m+n} + \frac{ic_{\tx{BH}}}{12} m^3 \delta_{m+n,0},
    \label{eq: L, L Poisson}
    \\
    \{ \overline{L}_m, \overline{L}_n \} &= i(m-n) \overline{L}_{m+n} + \frac{ic_{\tx{BH}}}{12} m^3 \delta_{m+n,0},
    \label{eq: Lbar, Lbar Poisson}
    \\
    \{ L_m, \overline{L}_n \} &= 0, \quad \forall m, n \in \mathbb{Z},
    \label{eq: L, Lbar Poisson}
\end{align}
where $c_{\tx{BH}}$ is called the Brown-Henneaux central charge
\begin{equation}
\label{eq: BH central charge}
    c_{\tx{BH}} = \frac{3}{2G_{\tx{N}}}.
\end{equation}
The classical observable algebra $C^\infty (\mathcal{P}_{\bdy})$ of the boundary gravitons is isomorphic to two decoupled copies of the Virasoro algebra of central charge $c_{\tx{BH}}$. $T$ and $\overline{T}$ are hence identified with chiral and antichiral components of the stress tensor of a 2-dimensional conformal field theory. In particular, the Hamiltonian is given by
\begin{equation}
\label{eq: bdy graviton Hamiltonian}
    H = L_0 + \overline{L}_0
\end{equation}
as the Noether charge associated with the time translation. 

The asymptotic symmetry transformation~\eqref{eq: asymptotic symmetry transformation} acts on the chiral stress tensor $T$ via the Virasoro coadjoint action
\begin{equation}
\label{eq: global Virasoro coadjoint action}
    T (x^+) \to T (F(x^+)) (F'(x^+))^2 - \frac{c_{\tx{BH}}}{12} \tx{Sch}(F, x^+),
\end{equation}
where $'$ denotes $\partial/\partial x^+$ and $\tx{Sch}$ denotes the Schwarzian derivative
\begin{equation}
    \tx{Sch}(F, x^+) \equiv \frac{F'''}{F'} - \frac{3}{2} \left(\frac{F''}{F'}\right)^2.
\end{equation}
The antichiral stress tensor transforms similarly. Not every pair of the Ba\~nados metrics is related by a Virasoro coadjoint action, so the boundary graviton phase space $\mathcal{P}_{\bdy}$ is foliated into orbits of the Virasoro coadjoint action~\cite{Navarro-Salas:1999ejl, Nakatsu:1999wt, Compere:2015knw}. If $T$ and $\overline{T}$ are positive constants, then the Ba\~{n}ados metric \eqref{eq: Banados metric} describes a BTZ black hole of energy $(T+ \overline{T})$ and angular momentum $(T - \overline{T})$~\cite{Banados:1992wn}. Since a BTZ black hole solution is the global minimum of the Hamiltonian~\eqref{eq: bdy graviton Hamiltonian} in its Virasoro coadjoint orbit~\cite{Witten:1987ty,  Balog:1997zz}, BTZ black holes of different energies and angular momenta are in different Virasoro coadjoint orbits~\cite{Garbarz:2014kaa}. In order to describe generic BTZ black holes, the corresponding Virasoro coadjoint orbits must be included in the boundary graviton phase space $\mathcal{P}_{\bdy}$.

If an observable commutes with all the observables of a complete system, then it must be proportional to the identity. In other words, the observable algebra of a complete system has a trivial center. By a complete system, we refer to a physical system without any correlation with other systems. If there exists a nontrivial observable $\mathcal{O}$ commuting with all the observables including the Hamiltonian, then then phase space can be foliated into leaves corresponding to fixed values of $\mathcal{O}$ such that points in one leaf never evolve to another leaf under the Hamiltonian evolution. For a classical observable algebra, the commutator is replaced by the corresponding Poisson bracket. We will show that the boundary graviton observable algebra $C^\infty (\mathcal{P}_{\bdy})$ will have a nontrivial center if the boundary graviton phase space $\mathcal{P}_{\bdy}$ includes generic black hole solutions.

The phase space of a complete system is a manifold equipped with a nondegenerate and closed 2-form called a symplectic form. Such a manifold is called a symplectic manifold. Inverting the symplectic form gives a Poisson bracket. A manifold equipped with a Poisson bracket is called a Poisson manifold. A symplectic manifold is a Poisson manifold but not vice versa. In general, a Poisson manifold is foliated into a set of symplectic manifolds called symplectic leaves~\cite{kirillov1976local,weinstein1983local}. The algebra $C^\infty (\mathcal{P})$ of classical observables on a Poisson manifold $\mathcal{P}$ has a trivial center if and only if $\mathcal{P}$ does not decompose into symplectic leaves. The dual $V^*$ of a linear space $V$ is defined as the linear space of linear functions on $V$. Given a Lie algebra~$\mg$, the Lie bracket $[,]_\mg$ of $\mg$ induces a Poisson bracket for the dual $\mg^*$ of $\mg$ by conversely viewing elements in $\mg$ as linear functions on $\mg^*$. The Poisson bracket $\{,\}_{\mg^*}$ of $\mg^*$ is then given by
\begin{equation}
\label{eq: gdual Poisson bracket}
    \{ X, Y \}_{\mg^*} \equiv [ X, Y ]_{\mg}\,, 
\end{equation}
$\forall X, Y \in \mg \subsetneq C^\infty (\mg^*)$. Let $G$ be the Lie group corresponding to $\mg$. The coadjoint action $\Ad^*_g$ of $g \in G$ on $u \in \mg^*$ is defined by
\begin{equation}
   \langle \Ad^*_g \,u, X \rangle_\mg \equiv \langle u , g^{-1} X g \rangle_\mg, \quad \forall X \in \mg,
\end{equation}
where $\langle, \rangle_\mg$ is the canonical pairing between $\mg$ and $\mg^*$. Being equipped with the Poisson bracket~\eqref{eq: gdual Poisson bracket}, $\mg^*$ decomposes into $G$-coadjoint orbits as symplectic leaves~\cite{kirillov2025lectures}.

The full stress tensor associated with the Ba\~nados metric~\eqref{eq: Banados metric} is the sum of the chiral quadratic differential $T dx^+ \otimes dx^+$ and the antichiral quadratic differential $\overline{T} dx^- \otimes dx^-$. Since the chiral Virasoro algebra~\eqref{eq: L, L Poisson} and the antichiral Virasoro algebra~\eqref{eq: Lbar, Lbar Poisson} are decoupled, we focus on the chiral part for simplicity. The chiral stress tensor $T dx^+ \otimes dx^+$ can be viewed as an element in the dual $\mathfrak{diff}(S^1)^*$ of (the chiral part of) the asymptotic symmetry algebra $\mathfrak{diff}(S^1)$. The canonical pairing $\langle, \rangle_{\mathfrak{diff}(S^1)}$ is given by Eq.~\eqref{eq: L_n def}, i.e. for $e^{inx^+} \partial_{x^+} \in \mathfrak{diff}(S^1)$,
\begin{equation}
    \langle T dx^+ \otimes dx^+, e^{inx^+} \partial_{x^+} \rangle_{\mathfrak{diff}(S^1)} \equiv L_n.
\end{equation}
The chiral stress tensor $T$ can also be viewed as an element $(T dx^+ \otimes dx^+, c_{\tx{BH}})$ in the dual $\widehat{\mathfrak{diff}}(S^1)^*$ of the centrally extended chiral asymptotic symmetry algebra $\widehat{\mathfrak{diff}}(S^1)$ which is the Virasoro algebra, i.e.
\begin{equation}
    \langle (Tdx^+ \otimes dx^+, c_{\tx{BH}}), (e^{inx^+} \partial_{x^+}, t) \rangle_{\widehat{\mathfrak{diff}}(S^1)} = L_n + t c_{\tx{BH}},
\end{equation}
where $t \in \widehat{\mathfrak{diff}}(S^1)$ is a central element. Being equipped with the Poisson bracket~\eqref{eq: L, L Poisson}, \eqref{eq: Lbar, Lbar Poisson}, and \eqref{eq: L, Lbar Poisson}, the space of the Ba\~nados metrics~\eqref{eq: Banados metric} is then foliated into Virasoro coadjoint orbits as symplectic leaves. Recall that the boundary graviton phase space $\mathcal{P}_{\bdy}$ need to include Virasoro coadjoint orbits to describe the corresponding generic BTZ black holes, so the boundary graviton observable algebra $C^\infty (\mathcal{P}_{\bdy})$ has a nontrivial center. Thus, the one-sided boundary gravitons can only be a subsystem of a complete system and the Hamiltonian defined in terms of boundary graviton data can never describe dynamical process in which the black hole mass and spin change.

The above proof relies on the asymptotically AdS$_3$ boundary conditions, but the incompleteness of the boundary gravitons is actually intrinsic in the sense that it is true regardless of boundary conditions. To see this, we use the first order formulation of 3d gravity~\cite{Achucarro:1986uwr, Witten:1988hc} in which the Einstein equation~\eqref{eq: Vacuum Einstein} is written as the flatness conditions of $\PSL \times \PSL$ gauge potential $(A , \overline{A})$,
\begin{align}
    dA +A^2 &= 0,
    \label{eq: flatness 1}
    \\
    d\overline{A} + \overline{A}^2 &= 0.
    \label{eq: flatness 2}
\end{align}
In terms of the gauge potential $(A, \overline{A})$, the asymptotically AdS$_3$ boundary condition is given by~\cite{Coussaert:1995zp}
\begin{align}
    A &= -\left(
    \begin{array}{cc}
        \frac{dr }{2r}  & O(\frac{1}{r}) dx^+ \\
        r dx^+ & -\frac{dr }{2r} 
    \end{array}
    \right),
    \label{eq: AdS boundary condition 1}
    \\
    \overline{A} &= -\left(
    \begin{array}{cc}
        -\frac{dr }{2r}  & r dx^- \\
         O(\frac{1}{r}) dx^- & \frac{dr }{2r} 
    \end{array}
    \right),
    \label{eq: AdS boundary condition 2}
\end{align}
where $O(1/r)$ means that it is a term of order $1/r$. The Ba\~nados metric~\eqref{eq: Banados metric} is encoded in the flat gauge potential via~\cite{Banados:1998gg}
\begin{align}
    A &= -\left(
    \begin{array}{cc}
        \frac{dr }{2r}  &\frac{4 G_{\tx{N}}}{r} T dx^+ \\
        r dx^+ & -\frac{dr }{2r} 
    \end{array}
    \right),
    \label{eq: Banados potential}
    \\
    \label{eq: Banados potential bar}
    \overline{A} &= -\left(
    \begin{array}{cc}
        -\frac{dr }{2r}  & r dx^- \\
         \frac{4 G_{\tx{N}}}{r} \overline{T} dx^- & \frac{dr }{2r} 
    \end{array}
    \right).
\end{align}
For a BTZ black hole of energy $(T + \overline{T})$ and angular momentum $(T - \overline{T})$, the Wilson loops winding around the asymptotic boundary once are given by
\begin{equation}
\label{eq: BTZ wilson loop}
    \tr \, M = 2\cosh \left(4\pi  \sqrt{G_{\tx{N}} T } \right), \quad \tr\, \overline{M} = 2\cosh \left(4\pi  \sqrt{G_{\tx{N}} \overline{T} } \right) ,
\end{equation}
where the trace $\tr$ is normalized as
\begin{equation}
\label{eq: trace normalization}
    \tr \begin{pmatrix}
        1 & 0 
        \\
        0 & 1
    \end{pmatrix} = 2.
\end{equation}
In a one-sided spacetime, the presence of the nontrivial Wilson loops~\eqref{eq: BTZ wilson loop} implies that the flatness conditions \eqref{eq: flatness 1} and \eqref{eq: flatness 2} cannot hold everywhere. Equivalently, the vacuum Einstein equation~\eqref{eq: Vacuum Einstein} cannot hold everywhere. So the boundary gravitons can only be a subsystem of the one-sided black hole spacetime.

We can make this argument clearer by utilizing the fact that the gauge transformation of the gauge potential $(A, \overline{A})$ is the $\PSL \times \PSL$ Kac-Moody coadjoint action. See e.g. \cite{pressley1987loop} for a review of the Kac-Moody coadjoint action. The Poisson bracket of the gauge potentials gives rise to the $\PSL \times \PSL$ Kac-Moody algebra, so the symplectic leaves of the boundary graviton phase space $\mathcal{P}_{\bdy}$ are $\PSL \times \PSL$ Kac-Moody coadjoint orbits and hence are orbits of the large gauge transformations. The Wilson loops~\eqref{eq: BTZ wilson loop} are gauge invariant, so they are constant on each Kac-Moody coadjoint orbit. Different BTZ black hole solutions hence lie in different $\PSL \times \PSL$ Kac-Moody coadjoint orbits. A coordinate transformation of the metric corresponds to a gauge transformation of the gauge potential in the first order formulation. After imposing some boundary conditions, less large gauge transformations are allowed. This leads to a finer foliation within each Kac-Moody coadjoint orbits. But different BTZ black hole solutions still belong to different leaves. To describe generic black holes, the one-sided boundary graviton observable algebra hence must have a nontrivial center, regardless of what boundary conditions one imposes. As we will see in Section~\ref{sec: Gauging Nonlocal Symmetries}, the center is in fact generated by the Wilson loops. Thus, the incompleteness of the one-sided boundary gravitons is an intrinsic feature independent of boundary conditions.

Although the vacuum Einstein equation~\eqref{eq: Vacuum Einstein} cannot hold everywhere for one-sided black holes, one can at least impose the vacuum Einstein equation~\eqref{eq: Vacuum Einstein} outside the horizon. In this case, boundary gravitons are the only degrees of freedom outside the horizon since there is no bulk graviton. In other words, the boundary graviton phase space $\mathcal{P}_{\bdy}$ is identified with the phase space $\mathcal{P}_\ads^{\tx{exterior}}$ of degrees of freedom outside the horizon, i.e.
\begin{equation}
    \mathcal{P}_\ads^{\tx{exterior}} = \mathcal{P}_{\bdy}.
\end{equation}
The boundary graviton observable algebra $C^\infty (\mathcal{P}_{\bdy})$ is hence the algebra $C^\infty(\mathcal{P}_\ads^{\tx{exterior}})$ of observables outside the horizon
\begin{equation}
\label{eq: identifying exterior algebra}
    C^\infty(\mathcal{P}_\ads^{\tx{exterior}}) = C^\infty (\mathcal{P}_{\bdy}).
\end{equation}

\subsection{Bootstrapping Completions from Asymptotic Symmetries}
\label{sec: Bootstrapping Completions from Asymptotic Symmetries}

In Section~\ref{sec: Intrinsic Incompleteness of One-Sided Boundary Gravitons}, we showed that the boundary graviton observable algebra $C^\infty (\mathcal{P}_{\bdy})$ needs to have a nontrivial center to describe generic BTZ black holes. In other words, the one-sided boundary gravitons needs to be a subsystem. If one imposes the vacuum Einstein equation~\eqref{eq: Vacuum Einstein} outside the horizon, then $C^\infty (\mathcal{P}_{\bdy})$ is identified with the algebra $C^\infty(\mathcal{P}_\ads^{\tx{exterior}})$ of observable outside the horizon. In this subsection, under these conditions, we show that the asymptotic symmetry strongly constrains the observable algebra such that we can bootstrap a general Poisson bracket~\eqref{eq: universal WW Poisson} of Wilson lines. We will show that this Poisson bracket leads to a complete observable algebra (i.e. an algebra without center) of one-sided black holes in Section~\ref{sec: An Effective Observable Algebra of One-Sided Black Holes}.

The flatness condition~\eqref{eq: flatness 1} of the gauge potential is the equation of motion of the Chern-Simons theory~\cite{Witten:1988hf, Elitzur:1989nr}. The Chern-Simons theory is a gauge theory defined by the Lagrangian
\begin{equation}
\label{eq: chern-simons lagrangian}
    L_{\tx{CS}} = \frac{k}{4 \pi} \tr (A dA + \frac{2}{3} A^3).
\end{equation}
The gauge potential $A$ is a $\mg$-valued 1-form with $\mg$ the Lie algebra of the gauge group $G$. The coupling constant $k$ is called level. For the $\PSL \times \PSL$ Chern-Simons theory describing 3d gravity, the level $k$ is related to the Newton constant $G_{\tx{N}}$ by
\begin{equation}
\label{eq: k = 1/4G_N}
    k = \frac{1}{4 G_{\tx{N}}}.
\end{equation} 
Suppose the theory is defined on $\Sigma \times \mathbb{R}$ where $\Sigma$ is a 2-dimensional manifold with a given asymptotic boundary. For the infinitesimal gauge transformation
\begin{equation}
\label{eq: gauge transformation}
    \delta_\phi A=d\phi - [A ,\phi]
\end{equation}
nontrivially acting on the given asymptotic boundary, the corresponding Noether charge is given by~\cite{Banados:1994tn}
\begin{equation}
\label{eq: KAc-Moody charge}
    Q[\phi] = \frac{k}{2\pi} \oint_{\bdy } \tr ( A \phi ).
\end{equation}
$\bigcirc$ is the given asymptotic boundary of Cauchy slice on which the large gauge transformation acts. In terms of Poisson bracket, a Noether charge is defined via~\footnote{After quantization, we have the Ward identity \begin{equation*}
    \delta_\phi \mathcal{O} = i[ Q[\phi], \mathcal{O} ].
\end{equation*}}
\begin{equation}
\label{eq: general Noether charge}
    \delta_\phi \mathcal{O} = \{ Q[\phi], \mathcal{O} \}
\end{equation}
with $\mathcal{O}$ an arbitrary function on the phase space and $\phi$ parameterizing the corresponding symmetry transformation.

In general, the asymptotic symmetry is constrained to preserve asymptotic boundary conditions. With different boundary conditions, the asymptotic asymmetry algebra reduces to different algebras. For example, it will reduce to the Virasoro algebra~\eqref{eq: L, L Poisson} if the asymptotically AdS$_3$ boundary conditions~\eqref{eq: AdS boundary condition 1} and \eqref{eq: AdS boundary condition 2} are imposed. To see the full structure of the asymptotic symmetry algebra and to bootstrap a general Poisson bracket, we do not impose any asymptotic boundary condition in this subsection. One can impose boundary conditions after the Poisson bracket is derived, but the boundary conditions on boundary gravitons does not affect the commutant of the observable algebra of boundary gravitons.

Denote by $W|_{x,y}$ a Wilson line connecting a point $y \notin \bigcirc$ to a point $x \in \bigcirc$,
\begin{equation}
\label{eq: W def}
    W|_{x,y} \equiv \overleftarrow{P} \exp (  \int_y^x A )
\end{equation}
with $\overleftarrow{P}$ the path ordering. Under the infinitesimal gauge transformation \eqref{eq: gauge transformation}, its variation is
\begin{equation}
\label{eq: asymptotic symmetry}
    \delta_\phi W|_{x,y} = \phi|_x W|_{x,y} .
\end{equation}
$\phi|_x$ denotes the value of $\phi$ at $x$. Setting $\mathcal{O}=W|_{x,y}$ in Eq.~\eqref{eq: general Noether charge}, then Eq.~\eqref{eq: KAc-Moody charge} implies
\begin{equation}
\label{eq: A phi, W Poisson}
    \{ \oint_{\bdy } \tr ( A \phi ), W|_{x,y} \} = \frac{2\pi}{k} \phi|_x W|_{x,y}.
\end{equation}
$\phi$ is arbitrary, so Eq.~\eqref{eq: A phi, W Poisson} is equivalent to
\begin{equation}
\label{eq: A tensor 1 , 1 tensor W Poisson}
    \{ A|_{x_1} \otimes 1,  1 \otimes W|_{x_2,y} \} = \frac{2\pi}{k} \sum_{a=1}^{\tx{dim}\,\mg} t_a \otimes (t^a W|_{x_2,y})\, \delta (x_1 - x_2) dx_1
\end{equation}
for $x_1 \in \bigcirc$ and $|x_1-x_2|<2 \pi$. $\{t_a\}$ is a basis of $\mg$. $\{t^a\}$ is the basis dual to $\{t_a\}$, i.e.
\begin{equation}
\label{eq: dual basis}
    \tr ( t_a t^b ) = \delta^b_a.
\end{equation}
$A|_x$ is the value of $A$ at $x$. We introduce a convenient notation for elements of tensor products of the universal enveloping algebra $U(\mg)$ of $\mg$,
\begin{equation}
\label{eq: very convenient notation}
    \underset{1}{\mathcal{O}} \equiv \mathcal{O} \otimes 1, \quad \underset{2}{\mathcal{O}} \equiv 1 \otimes \mathcal{O},\quad \forall \mathcal{O} \in U(\mg).
\end{equation}
Eq.~\eqref{eq: A tensor 1 , 1 tensor W Poisson} is simplified with this notation,
\begin{equation}
\label{eq: A , W Poisson, better notation}
    \{ \underset{1}{A}|_{x_1} , \underset{2}{W}|_{x_2, y} \} = \frac{2\pi}{k} \sum_{a=1}^{\tx{dim}\,\mg} \underset{1}{t_a} \underset{2}{t^a} \underset{2}{W}|_{x,y} \delta (x_{12}) dx_1 
\end{equation}
with
\begin{equation}
\label{eq: x_12}
    x_{12} \equiv x_1 - x_2.
\end{equation}
Eq.~\eqref{eq: W def} implies
\begin{equation}
\label{eq: A = dWW-1}
    A|_x = d W|_{x,y} W^{-1}|_{x,y},
\end{equation}
where the exterior derivatives $d$ is restricted to the given asymptotic boundary. So Eq.~\eqref{eq: A , W Poisson, better notation} is equivalent to
\begin{equation}
\label{eq: d W W-1, W Poisson}
    \{ d \underset{1}{ W }|_{x_1, y_1} \underset{1}{W}^{-1}|_{x_1, y_1} , \underset{2}{ W}|_{x_2, y_2} \} = \frac{2 \pi}{k} \underset{12}{K} \underset{2}{W}|_{x_2, y_2} \delta (x_{12}) dx_1
\end{equation}
with the tensor quadratic Casimir $K$ defined as
\begin{equation}
\label{eq: K def}
    K \equiv \sum_{a=1}^{\tx{dim}\,\mg} t_a \otimes t^a.
\end{equation}
Using the Leibniz rule for the Poisson bracket 
\begin{equation}
\label{eq: Leibniz rule for the Poisson bracket}
    \{ \mathcal{O}_1 \mathcal{O}_2, \mathcal{O}_3 \} = \mathcal{O}_1\{ \mathcal{O}_2, \mathcal{O}_3  \} + \{ \mathcal{O}_1, \mathcal{O}_3 \} \mathcal{O}_2
\end{equation}
and the Leibniz rule for the exterior derivatives
\begin{equation}
    0 = d(WW^{-1}) = d W W^{-1} + W dW^{-1},
\end{equation}
it is straightforward to verify that Poisson bracket \eqref{eq: d W W-1, W Poisson} is equivalent to
\begin{equation}
\label{eq: ODE}
    d_{x_1} \left( \underset{1}{W}^{-1}|_{x_1, y_1} \underset{2}{W}^{-1}|_{x_2, y_2} \{ \underset{1}{W}|_{x_1, y_1}, \underset{2}{W}|_{x_2, y_2} \} \right) = \frac{\pi}{k} d_{x_1}  \left( \underset{1}{W}|_{y_1, y_2} \underset{12}{K} \underset{1}{W}|_{y_2, y_1}  \tx{sgn}(x_{12}) \right) 
\end{equation}
where 
\begin{equation}
\label{eq: sgn def}
    \tx{sgn}(x) \equiv \frac{x}{|x|}.
\end{equation}
The subscript ``$_{x_1}$'' of the exterior derivatives $d_{x_1}$ indicates that it only acts on $x_1$. Note that both sides of Eq.~\eqref{eq: ODE} are total derivatives (i.e. exact forms) of the form $d_{x_1} (...)$, so the Wilson line bracket is determined up to a $\mg \otimes \mg$-valued function $r$ independent of $x_1$, i.e.
\begin{equation}
\label{eq: universal WW Poisson}
    \{ \underset{1}{W}|_{x_1, y_1}, \underset{2}{W}|_{x_2, y_2} \} = \frac{2 \pi}{k} \underset{1}{W}|_{x_1, y_1} \underset{2}{W}|_{x_2, y_2} \left( \underset{12}{r} + \frac{1}{2} \underset{1}{W}|_{y_1, y_2} \underset{12}{K} \underset{1}{W}|_{y_2, y_1} \tx{sgn}(x_{12})  \right).
\end{equation}
Poisson bracket is antisymmetric, so $r$ is also independent of $x_2$. In other words, $r$ is independent of local data of the gauge potential on the given asymptotic boundary $\bigcirc$. 

\begin{figure}
    \centering
    \includegraphics[width=0.25\linewidth]{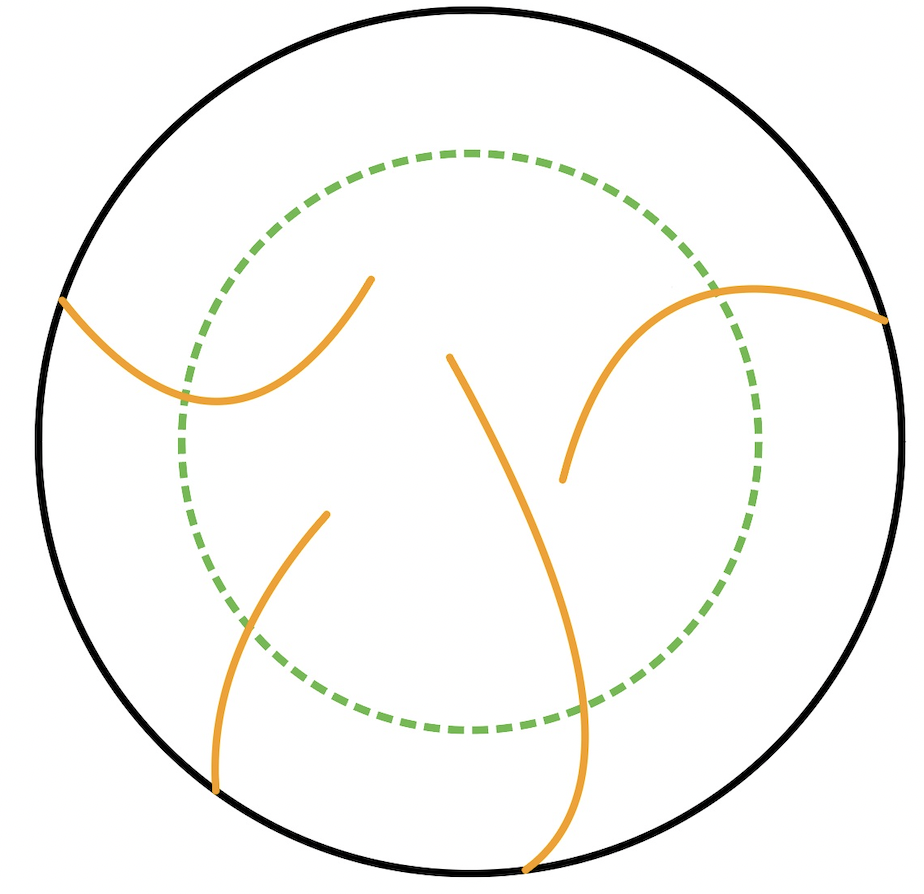}
    \caption{The outer solid circle represents the asymptotic boundary where Eq.~\eqref{eq: A phi, W Poisson} holds. The orange lines are Wilson lines connecting to the asymptotic boundary. We do not specify the topology of the region inside the green dashed circle. All the information inside the green dashed circle is encoded in the $r$ function in Eq.~\eqref{eq: universal WW Poisson}.}
    \label{fig: general WW Poisson}
\end{figure}

The topology of the Cauchy slice is not specified in the above derivation. See Figure~\ref{fig: general WW Poisson} for a illustration. So the Poisson bracket~\eqref{eq: universal WW Poisson} holds for any Cauchy slice topology with $x_1, x_2 \in \bigcirc$ and $y_1, y_2 \notin \bigcirc$. Different choices of $r$ in Eq.~\eqref{eq: universal WW Poisson} describe different systems containing the given asymptotic boundary. In Appendix~\ref{app: Two-Sided Wilson Line Poisson Bracket}, we will illustrate the choice of $r$ in Eq.~\eqref{eq: universal WW Poisson} that gives rise to the Poisson bracket of Chern-Simons theory defined on the cylinder Cauchy slice with two asymptotic boundaries as an example of a completion of the boundary graviton observable algebra $C^\infty (\mathcal{P}_{\bdy})$.

\subsection{An Effective Observable Algebra of One-Sided Black Holes}
\label{sec: An Effective Observable Algebra of One-Sided Black Holes}

As discussed in Section~\ref{sec: Intrinsic Incompleteness of One-Sided Boundary Gravitons}, the vacuum Einstein equation~\eqref{eq: Vacuum Einstein} must break down at least somewhere behind the horizon of one-sided black holes. In other words, there is matter confined in the black hole interior. 3d gravity coupled with matter in general is not equivalent to a topological theory, so the Chern-Simons formulation may not be applicable. However, one may still be able to derive an effective Chern-Simons description by integrating out the matter sector~\cite{Castro:2023bvo, Bourne:2024ded, Bourne:2025azc}. In this subsection, we assume that the one-sided black holes including the effective interior degrees of freedom are described by Wilson lines in $\PSL \times \PSL$ Chern-Simons theory. We apply the general Poisson bracket~\eqref{eq: universal WW Poisson} to construct a complete observable algebra containing the boundary graviton observable algebra $\mathcal{P}_{\bdy}$. This complete observable algebra will be interpreted as an effective observable algebra of one-sided black holes.

As shown in Eq.~\eqref{eq: BTZ wilson loop}, the energy $ (T + \overline{T})$ and angular momentum $ (T - \overline{T})$ of a BTZ black hole are encoded in Wilson loops around the black hole horizon. For the phase space of black holes to be complete, the Wilson loops need to be allowed vary freely. In the Chern-Simons theory, the simplest Cauchy slice topology that allows the Wilson loops to freely vary is a disc with a bulk puncture. We will focus on the chiral part $A$ to avoid redundant discussion for the anti-chiral part $\overline{A}$. Coordinates $y_{1}$ and $y_2$ in Eq.~\eqref{eq: universal WW Poisson} are not on the given asymptotic boundary, so they can only be at the puncture to ensure $W_{x_i, y_i}$ is a physical operator. Then $W_{y_1, y_2} = 1$ and Eq.~\eqref{eq: universal WW Poisson} reduces to
\begin{equation}
\label{eq: simplest WW Poisson}
    \{ \underset{1}{W}|_{x_1}, \underset{2}{W}|_{x_2} \} = \frac{2 \pi}{k} \underset{1}{W}|_{x_1} \underset{2}{W}|_{x_2} \left( \underset{12}{r} + \frac{1}{2} \underset{12}{K}  \tx{sgn}(x_{12})  \right).
\end{equation}
$W|_x$ is the Wilson line from the puncture to a point $x$ on the boundary. See Figure~\ref{fig: punctured disc} for a illustration. 

\begin{figure}
    \centering
    \includegraphics[width=0.25\linewidth]{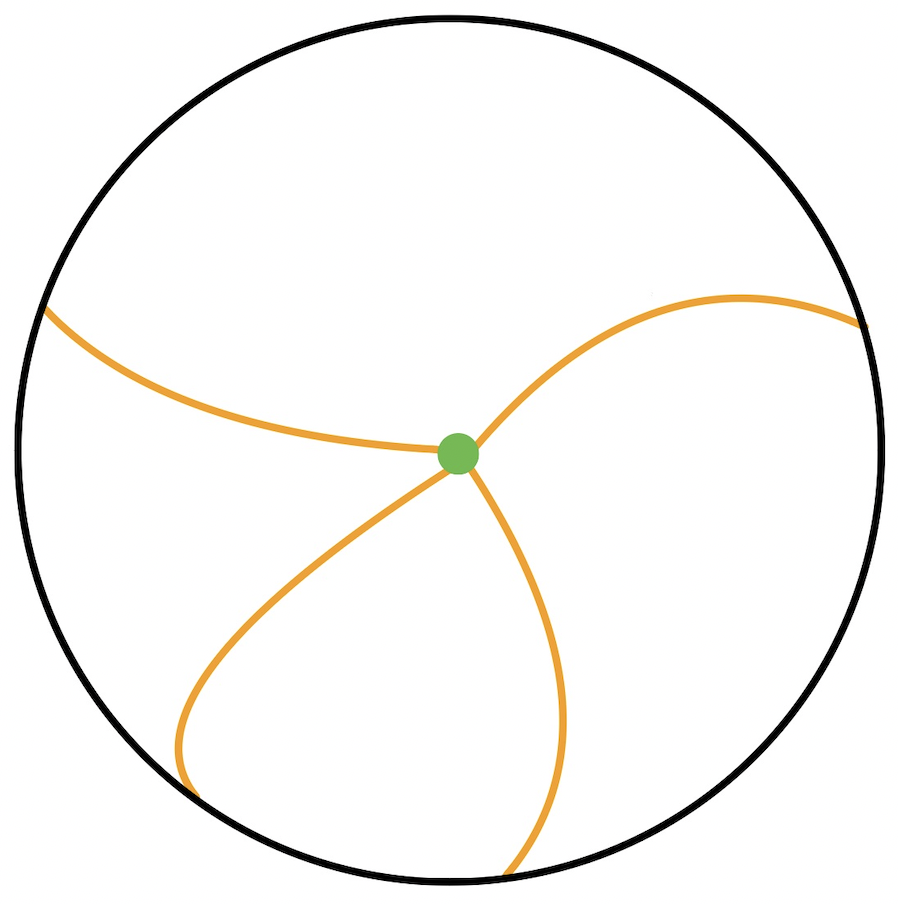}
    \caption{The outer black circle represents the asymptotic boundary. The green dot represents the puncture as an effective bulk illustration of the nonlocal boundary degrees of freedom holographically dual to the black hole. The orange lines represent Wilson lines.}
    \label{fig: punctured disc}
\end{figure}

The simplest choice of $r$ for Eq.~\eqref{eq: simplest WW Poisson} is a constant function antisymmetric with respect to the permutation $1 \leftrightarrow 2$. The Jacobi identity of Poisson bracket \eqref{eq: simplest WW Poisson} is equivalent to the modified classical Yang-Baxter equation (MCYBE)
\begin{equation}
\label{eq: MCYBE}
     [\underset{12}{r}, \underset{23}{r}] + [\underset{23}{r}, \underset{31}{r}] + [\underset{31}{r}, \underset{12}{r}] = -\frac{1}{4} \underset{123}{f},
\end{equation}
with $f \in \mg \otimes \mg \otimes \mg$ defined by
\begin{equation}
    \underset{123}{f} \equiv [ \underset{13}{K}, \underset{23}{K} ].
\end{equation}
We used the notation
\begin{equation}
    \underset{1}{\mathcal{O}} \equiv \mathcal{O} \otimes 1 \otimes 1, \quad \underset{2}{\mathcal{O}} \equiv 1 \otimes \mathcal{O} \otimes 1,\quad \underset{3}{\mathcal{O}} \equiv 1 \otimes 1 \otimes \mathcal{O}.
\end{equation}
Solutions to the MCYBE \eqref{eq: MCYBE} for simple Lie algebras were systematically classified in \cite{belavin1982solutions}. See e.g. \cite{chari1995guide} for a review. 

Eq.~\eqref{eq: simplest WW Poisson} with constant $r$-matrix is the simplest possible form of Poisson bracket of Wilson lines in Chern-Simons theory defined on a punctured disc. It was derived by symmetry bootstrap in chiral WZNW model on a cylinder \cite{Balog:1999pj} and in Chern-Simons theory on a punctured disc \cite{Mertens:2025ydx}. We generalized this bootstrap method to arbitrary topologies in Section \ref{sec: Bootstrapping Completions from Asymptotic Symmetries} and reproduce the special case of punctured disc in this section. The presence of MCYBE \eqref{eq: MCYBE} by the Poisson bracket bootstrap shows that the Yang-Baxter type structure is inevitable in Chern-Simons theory and 3d gravity.

The sign of the coefficient on the right hand side of MCYBE \eqref{eq: MCYBE} significantly affects the solutions. Eq.~\eqref{eq: MCYBE} has a negative coefficient and such an MCYBE is called the split type. It has no solution for compact groups \cite{cahen1994some}. Relevant to AdS$_3$ gravity is the Lie algebra $\mathfrak{sl}(2, \mathbb{R})$, for which there is a unique solution to Eq.~\eqref{eq: MCYBE} up to an isomorphism \cite{belavin1982solutions, chari1995guide}, i.e.
\begin{equation}
\label{eq: sl2 r-matrix}
    r = \frac{1}{2} \mathrm{E} \otimes \mathrm{F} - \frac{1}{2} \mathrm{F} \otimes \mathrm{E}
\end{equation}
with $\mathfrak{sl}(2, \mathbb{R})$ generators
\begin{equation}
\label{eq: sl2 generators}
\mathrm{H} = \left(\begin{array}{cc} 1 & 0 \\ 0  & -1\end{array}\right), \quad 
\mathrm{E} = \left(\begin{array}{cc} 0 & 1 \\ 0  & 0\end{array}\right), \quad  \mathrm{F} = \left(\begin{array}{cc} 0 & 0 \\ 1  & 0\end{array}\right),
\end{equation}
and tensor quadratic Casimir $K \in \mathfrak{sl}(2,\mathbb{R}) \otimes \mathfrak{sl}(2,\mathbb{R})$
\begin{equation}
\label{eq: sl2 K}
    K =  \frac{1}{2} \mathrm{H} \otimes \mathrm{H} +  \mathrm{E} \otimes \mathrm{F} +  \mathrm{F} \otimes \mathrm{E}.
\end{equation}

Poisson bracket \eqref{eq: simplest WW Poisson} was originally found as Poisson bracket of chiral WZNW fields $W$ and its corresponding symplectic form is given by \cite{Gawedzki:1990jc, Alekseev:1990vr, Alekseev:1991wq}
\begin{equation}
\label{eq: Omega odot}
    \Omega_\odot = \frac{k}{4 \pi} \left(  \tr ( W_0^{-1} \delta W_0 \delta m m^{-1})  + \oint_{\scalebox{0.5}{$\bigcirc$}} \tr ( \delta W  W^{-1} \delta (dW W^{-1}))  - \tr ( \delta m_- m_-^{-1} \delta m_+ m^{-1}_+ ) \right).  
\end{equation}
$\delta$ and $d$ are exterior derivatives of the phase space and the underlying spacetime respectively. We choose the convention $d \delta = \delta d$. $W_0$ is the Wilson line  connecting the puncture to the base point $x^+ = 0$ on the boundary circle $\bigcirc$. $x^+ = 0$ is the starting point of the contour integral along the boundary circle $\bigcirc$ \footnote{Choosing a base point for the contour integral matters since $W$ is not a single-value function on $\bigcirc$.}. The puncture monodromy $m$ is defined as
\begin{equation}
\label{eq: m def}
    m \equiv W^{-1}|_{0 } W|_{2 \pi}.
\end{equation}
See Figure~\ref{fig: puncture monodromy} for a illustration.
\begin{figure}
    \centering
    \includegraphics[width=0.25\linewidth]{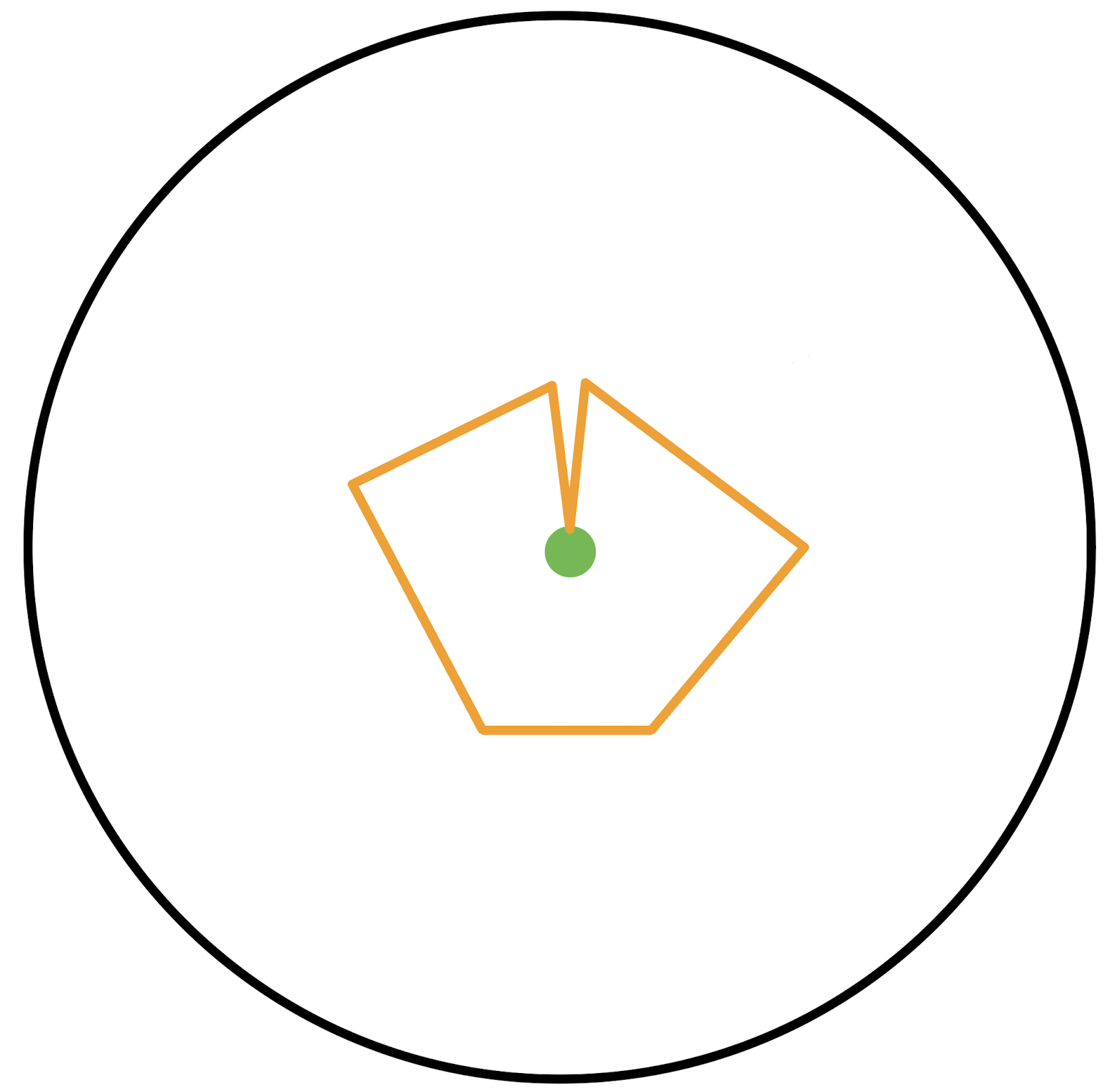}
    \caption{The outer black circle represents the asymptotic boundary. The green dot represents the puncture. The orange polygon represents a Wilson line starting and ending at the puncture and winding around it once. Such a Wilson line is defined as the puncture monodromy~\eqref{eq: m def}.}
    \label{fig: puncture monodromy}
\end{figure}
If we parametrize $m \in \tx{PSL}(2, \mathbb{R})$ by
\begin{equation}
\label{eq: m_e,f,k parameterization}
    m=\left(
\begin{array}{cc}
 -m_{\mathrm{K}}^2 &   -m_{\mathrm{K}} m_{\mathrm{f}} \\
 m_{\mathrm{e}} m_{\mathrm{K}}  &   m_{\mathrm{e}} m_{\mathrm{f}} -m_{\mathrm{K}}^{-2} \\
\end{array}
\right)
\end{equation}
with $m_{\mathrm{K}}^2,m_{\mathrm{e}}, m_{\mathrm{f}} \in \mathbb{R}$, then $m_\pm$ in the last term of \eqref{eq: Omega odot} are defined by
\begin{align}
    m_- &\equiv -\left(
\begin{array}{cc}
 m_{\mathrm{K}}^{-1} & 0 \\
 m_{\mathrm{e}} & m_{\mathrm{K}} \\
\end{array}
\right),
\label{eq: m-}
\\
m_+ &\equiv \left(
\begin{array}{cc}
 m_{\mathrm{K}}  & m_{\mathrm{f}} \\
 0 & m_{\mathrm{K}}^{-1} \\
\end{array}
\right),
\label{eq: m+}
\end{align}
such that
\begin{equation}
\label{eq: Borel decomposition}
    m = m_-^{-1} m_+.
\end{equation}
We call Eq.~\eqref{eq: Borel decomposition} the Semenov-Tian-Shansky decomposition of $m$ or STS decomposition for short~\cite{Semenov-Tian-Shansky:1985mgd}. With the symplectic form \eqref{eq: Omega odot} (and the anti-chiral counterpart) and imposing the asymptotically AdS$_3$ boundary conditions~\eqref{eq: AdS boundary condition 1} and \eqref{eq: AdS boundary condition 2}, the Wilson lines connecting the puncture to the asymptotic boundary form a complete phase space $\mathcal{P}_\odot$ of the one-sided black hole.

The boundary gauge potential $A$ is the charge generating the large gauge transformation~\eqref{eq: gauge transformation} on the asymptotic boundary, so
\begin{equation}
\label{eq: locality}
    \{ A, m \} = 0.
\end{equation}
This is a manifestation of the locality and can also be verified by a straightforward calculation using Eqs.~\eqref{eq: A = dWW-1}, \eqref{eq: Leibniz rule for the Poisson bracket}, \eqref{eq: simplest WW Poisson}, and \eqref{eq: m def}. We have identified the boundary graviton observable algebra $C^\infty (\mathcal{P}_{\bdy})$ with the algebra $C^\infty (\mathcal{P}_\ads^{\tx{exterior}})$ of observables outside the horizon in Eq.~\eqref{eq: identifying exterior algebra}. From the perspective of an infalling observer in a one-sided black hole spacetime, the commutant of $C^\infty (\mathcal{P}_\ads^{\tx{exterior}})$ is identified with the algebra $C^\infty (\mathcal{P}_\ads^{\tx{interior}})$ of observables behind the horizon. Eq.~\eqref{eq: locality} implies that $C^\infty (\mathcal{P}_\ads^{\tx{interior}})$ contains the Poisson algebra $C^\infty (\mathcal{P}_\bullet)$ generated by the puncture monodromy observables. In Section~\ref{sec: Gauging Nonlocal Symmetries}, we will conversely show that the monodromy observable algebra~$C^\infty (\mathcal{P}_\bullet)$ is actually the commutant of the boundary graviton algebra~$C^\infty (\mathcal{P}_{\bdy})$ and hence is identified with the interior algebra~$C^\infty (\mathcal{P}_\ads^{\tx{interior}})$, i.e.
\begin{equation}
\label{eq: interior algebra = monodromy algebra}
    C^\infty (\mathcal{P}_\ads^{\tx{interior}}) = C^\infty (\mathcal{P}_\bullet).
\end{equation}
Notice that Eq.~\eqref{eq: interior algebra = monodromy algebra} is a sensible interpretation only for an infalling observer who can access the degrees of freedom behind the horizon. We will provide a complementary interpretation in Section~\ref{sec: Positivity Restrictions Lorentzian Wormholes}.

The Wilson line $W$ in Poisson bracket~\eqref{eq: simplest WW Poisson} can be identified as the chiral WZNW filed $W$ in the symplectic form~\eqref{eq: Omega odot}. From the boundary observer perspective, the boundary gauge potential~\eqref{eq: A = dWW-1} can be constructed from local data, while the puncture monodromy~\eqref{eq: m def} is only accessible to nonlocal measurement, i.e. a boundary observer needs to travel around the universe once to collect the necessary data to measure the puncture monodromy. More precisely, Eqs.~\eqref{eq: A phi, W Poisson} and \eqref{eq: A = dWW-1} imply 
\begin{equation}
    \{ A|_{x_1}, A|_{x_2} \} = 0, \quad x_1 \ne x_2.
\end{equation}
So the boundary gauge potential are local degrees of freedom. Notice that the WZNW field $W$ are not local degrees of freedom since Eq.~\eqref{eq: simplest WW Poisson} implies
\begin{equation}
    \{ W|_{x_1}, W|_{x_2} \} \ne 0, \quad x_1 \ne x_2.
\end{equation}
Eq.~\eqref{eq: A = dWW-1} implies
\begin{equation}
    \overleftarrow{P} \exp (  \int_y^x A ) = W|_x W^{-1}|_y.
\end{equation}
So the puncture monodromy~\eqref{eq: m def} cannot be constructed from the boundary gauge potential~\eqref{eq: A = dWW-1}.
It leads us to identify the observable algebra $C^\infty (\mathcal{P}_{\bdy})$ of the boundary gauge potential~\eqref{eq: A = dWW-1} with the algebra $C^\infty (\mathcal{P}_\cft^{\tx{local}})$ of local classical observables in the boundary theory, i.e.
\begin{equation}
    C^\infty (\mathcal{P}_\cft^{\tx{local}}) = C^\infty (\mathcal{P}_{\bdy}), 
\end{equation}
Correspondingly, the observable algebra $C^\infty (\mathcal{P}_\bullet)$ of the puncture monodromy~\eqref{eq: m def} is hence identified with the algebra~$C^\infty (\mathcal{P}_\cft^{\tx{nonlocal}})$ of nonlocal classical
observables in the boundary theory, i.e.
\begin{equation}
\label{eq: identifying classical nonlocal algebra}
    C^\infty (\mathcal{P}_\cft^{\tx{nonlocal}}) = C^\infty (\mathcal{P}_\bullet).
\end{equation}
The puncture then can be interpreted as an effective bulk illustration of the nonlocal boundary degrees of freedom holographically dual to the black hole interior.

\section{Effective Phase Space of Black Hole Interior}

\subsection{Interior Phase Spaces from Gauging Asymptotic Symmetries}
\label{sec: Interior Phase Spaces from Gauging Asymptotic Symmetries}

In this subsection, we identify an effective phase spaces associated with the black hole interior by noting that the asymptotic symmetries are gauge symmetries of the effective interior degrees of freedom. We focus on the chiral sector to avoid redundant discussions on the anti-chiral sector.

As explained at the end of Section \ref{sec: An Effective Observable Algebra of One-Sided Black Holes}, the puncture monodromy $m$ should be interpreted as degrees of freedom associated to one-sided BTZ black hole interior. The space $\mathcal{P}_\bullet$ of puncture monodromy $m$ inherits a Poisson bracket from the Poisson bracket~\eqref{eq: simplest WW Poisson} on $\mathcal{P}_\odot$ via the puncture reduction map
\begin{align}
\label{eq: puncture reduction map}
    \mathcal{F}_{\bullet} : \mathcal{P}_\odot &\to \mathcal{P}_\bullet
    \\
    W &\to m = W^{-1}|_0 W|_{2 \pi}.
    \label{eq: W to m}
\end{align}
$\mathcal{P}_\odot$ is the phase space of Wilson lines connecting the puncture to the asymptotic boundary. Using Eqs.~\eqref{eq: Leibniz rule for the Poisson bracket}, \eqref{eq: simplest WW Poisson}, and \eqref{eq: W to m}, we have
\begin{equation}
\label{eq: STS Poisson}
    \{ \underset{1}{m}, \underset{2}{m} \} = \frac{2 \pi}{k} (- \underset{1}{m} \underset{12}{r_+} \underset{2}{m} + \underset{12}{r_-} \underset{1}{m} \underset{2}{m} + \underset{2}{m} \underset{1}{m} \underset{12}{r_+} - \underset{2}{m} \underset{12}{r_-} \underset{1}{m} ),
\end{equation}
with
\begin{equation}
\label{eq: r_pm def}
    r_\pm \equiv r \pm \frac{K}{2}. 
\end{equation}
$r$ is the $r$-matrix in Eq.~\eqref{eq: simplest WW Poisson} with the explicit form \eqref{eq: sl2 r-matrix}. $K$ is the tensor quadratic Casimir defined in Eq.~\eqref{eq: K def} with the explicit form \eqref{eq: sl2 K}. Poisson Bracket \eqref{eq: STS Poisson} is called the Semenov-Tian-Shansky bracket \cite{Semenov-Tian-Shansky:1985mgd}~\footnote{A closely related structure was also discovered in 3D gravity earlier. In \cite{Dupuis:2020ndx}, the authors discretize the Cauchy slice and obtain a set of links on the slice. In the Euclidean signature, the holonomy along every link is an element in $\tx{SL}(2, \mathbb{C})$. They showed that $G$ satisfies the Heisenberg double Poisson bracket\cite{Semenov-Tian-Shansky:1985mgd,Semenov-Tian-Shansky:1992xxt,Alekseev:1993qs}. In general, one can derive the Semenov-Tian-Shansky bracket from the Heisenberg double Poisson bracket by Poisson reduction~\cite{Semenov-Tian-Shansky:1992xxt}. However, in the current paper the Semenov-Tian-Shansky bracket is satisfied by the puncture monodromy instead of a Wilson line along an open link.}.

The Poisson bracket \eqref{eq: STS Poisson} of the puncture monodromy is independent of asymptotic boundary condition. Recall that we chose the free boundary condition in Section \ref{sec: An Effective Observable Algebra of One-Sided Black Holes}. If we impose further constraints on the gauge potential, the Wilson line Poisson bracket \eqref{eq: simplest WW Poisson} will be replaced by a Dirac bracket accordingly. But the corresponding Dirac bracket of the puncture monodromy variables will be the same as \eqref{eq: STS Poisson} due to the locality condition \eqref{eq: locality}. In this subsection, we leave the asymptotic boundary conditions and the corresponding asymptotic symmetry unspecified.

Poisson bracket \eqref{eq: STS Poisson} is degenerate on $\mathcal{P}_\bullet$ since $\tx{PSL}(2, \mathbb{R})$ is 3-dimensional and a symplectic manifold must be even-dimensional~\footnote{This is a quick way to see the degeneracy for the current case. The Semenov-Tian-Shansky bracket \eqref{eq: STS Poisson} on a Lie group is generically degenerate \cite{Semenov-Tian-Shansky:1985mgd}.}. So $\mathcal{P}_{\bullet}$ itself cannot be a physical phase space (i.e. a symplectic manifold) associated to the black hole interior. Recall that the space $\mathcal{P}_{\bdy}$ of Ba\~nados metrics is decomposed into the Virasoro coadjoint orbits as physical phase subspaces of boundary gravitons. Similarly, $\mathcal{P}_\bullet$ is also decomposed into physical phase subspaces which we refer to as interior phase spaces. They correspond to superselection sectors after quantization. In the following, we identify the interior phase spaces by interpreting the puncture reduction \eqref{eq: puncture reduction map} as gauging the asymptotic symmetry of the boundary gravitons. 

By definition \eqref{eq: m def}, the puncture monodromy $m$ is invariant under the asymptotic symmetry
\begin{align}
\label{eq: ASG action}
    \tx{ASG} 
    \times \mathcal{P}_\odot &\to \mathcal{P}_\odot
    \\
    (g, W) &\to g W.
    \label{eq: boundary large gauge transformation}
\end{align}
$\tx{ASG}$ denotes an unspecified asymptotic symmetry group. Eq.~\eqref{eq: ASG action} is the finite form of the infinitesimal large gauge transformation \eqref{eq: gauge transformation} preserving some unspecified boundary conditions. So the asymptotic symmetry \eqref{eq: ASG action} is a gauge symmetry of $\mathcal{P}_\bullet$. Conversely, we can show that an observable $\mathcal{O} \in C^{\infty}(\mathcal{P}_\odot)$ is invariant under the asymptotic symmetry \eqref{eq: ASG action} only if it is a function of the puncture monodromy $m$.  $\mathcal{O}$ as an observable must be a functional of Wilson lines ending at the puncture or the asymptotic boundary. $\mathcal{O}$ is invariant under the asymptotic symmetry only if it is not composed of Wilson lines ending at the asymptotic boundary. A Wilson line ending at the puncture must be a function of the puncture monodromy. Collecting the above, we see that an observable $\mathcal{O} \in C^\infty (\mathcal{P}_\odot)$ is invariant under the asymptotic symmetry \eqref{eq: ASG action} if and only if $\mathcal{O} \in C^{\infty} (\mathcal{P}_\bullet)$~\footnote{Strictly speaking, $\mathcal{O}$ is an element in the pullback $\mathcal{F}_{\bullet}^* (C^\infty (\mathcal{P}_\bullet))$ of $C^{\infty} (\mathcal{P}_\bullet)$ via the puncture reduction map \eqref{eq: puncture reduction map}. But we avoid introducing unnecessary notations here.}. Thus, the observable algebra $ C^\infty (\mathcal{P}_\bullet)$ of the black hole interior is identified as the gauge-invariant observable subalgebra of $C^{\infty} (\mathcal{P}_\odot)$ with the asymptotic symmetry $\eqref{eq: ASG action}$ as the gauge symmetry. 

The charge associated with a gauge symmetry is constant
on a physical phase subspace. The charge \eqref{eq: KAc-Moody charge} generating the asymptotic symmetry \eqref{eq: ASG action} is a functional of the boundary gauge potential $A$. So each physical phase subspace associated to the black hole interior can be identified as a subspace $\mathcal{L}_{\bullet} (A) \subset \mathcal{P}_\bullet$  corresponding to a fixed boundary gauge potential configuration $A$. More precisely,
\begin{equation}
\label{eq: L (A) leaf}
    \mathcal{L}_{\bullet} (A)  = \mathcal{F}_{\bullet} \left( \mathcal{F}^{-1}_{\bdy} (A) \right),
\end{equation}
where $\mathcal{F}_{\bullet}$ is the puncture reduction map \eqref{eq: puncture reduction map} and $\mathcal{F}_{\bdy}^{-1} (A)$ is the preimage of a given fixed configuration $A$ under the boundary reduction map
\begin{align}
    \mathcal{F}_{\bdy} : \mathcal{P}_\odot &\to \mathcal{P}_{\scalebox{0.5}{$\bigcirc$}}
    \label{eq: bdy reduction map}
    \\
    W &\to A = dW W^{-1}.
    \label{eq: W to A}
\end{align}

Now we describe the interior phase spaces $\mathcal{L}_\bullet (A)$ in detail. The preimage $\mathcal{F}_{\bdy}^{-1} (A) \subset \mathcal{P}_\odot$ is composed of all Wilson lines sharing the same boundary gauge potential configuration $A$, i.e.
\begin{equation}
\label{eq: A preimage}
    \mathcal{F}_{\bdy}^{-1} (A) = \{ W \in \mathcal{P}_\odot ;\, dW W^{-1}|_{\bdy} = A \}.
\end{equation}
$|_{\bdy}$ indicates restriction to the circle boundary $\bigcirc$. Any two elements $W$ and $W'$ in $\mathcal{F}_{\bdy}^{-1} (A)$ are related by a puncture large gauge transformation
\begin{align}
    \mathcal{P}_\odot \times \tx{PSL}(2, \mathbb{R}) &\to \mathcal{P}_\odot
    \label{eq: puncture gauge transformation}
    \\
    (W, h) &\to Wh.
\end{align}
So every $W \in \mathcal{F}_{\bdy}^{-1} (A) $ can be decomposed as
\begin{equation}
\label{eq: W = checkW h}
    W = \check{W} h
\end{equation}
with $h$ a puncture large gauge transformation \eqref{eq: puncture gauge transformation} and $\check{W}$ a fixed representative in $\mathcal{F}_{\bdy}^{-1} (A)$. The puncture large gauge transformation \eqref{eq: puncture gauge transformation} also keeps $\mathcal{F}_{\bdy}^{-1}(A)$ invariant, so $\mathcal{F}_{\bdy}^{-1} (A)$ is an orbit of the puncture large gauge transformations \eqref{eq: puncture gauge transformation} acting on the representative $\check{W}$. Using Eq.~\eqref{eq: W to m} and \eqref{eq: W = checkW h}, we have
\begin{equation}
    \mathcal{F}_{\bullet} (W) = h^{-1} \check{m} h
\end{equation}
with $\check{m}$ the puncture monodromy of the representative $\check{W}$, i.e.
\begin{equation}
    \check{m} = \check{W}^{-1}|_0 \check{W}|_{2 \pi}.
\end{equation}
The representative $\check{m} \in \tx{PSL}(2, \mathbb{R})$ can be any element conjugated to the boundary monodromy
\begin{equation}
\label{eq: M def}
    M \equiv \overleftarrow{P} \exp \oint_{\bdy} A  = W|_{2 \pi} W^{-1}|_0.
\end{equation}
$\oint_{\bdy}$ is an integral along the asymptotic boundary $\bigcirc$ from $x^+ = 0$ to $x^+= 2\pi$. See Figure~\ref{fig: boundary monodromy} for a illustration. Combining the above, we characterize $\mathcal{L}_\bullet (A)$ as follows. Given the degenerate Poisson bracket \eqref{eq: STS Poisson} on $\mathcal{P}_\bullet$, each interior phase space \eqref{eq: L (A) leaf} is an orbit of the puncture large gauge transformations \eqref{eq: puncture gauge transformation} acting on a representative puncture monodromy $\check{m}$ conjugated to the boundary monodromy $M$, i.e.
\begin{equation}
\label{eq: chracterizing L (A)}
    \mathcal{L}_\bullet (A) = \{ h^{-1}  \check{m} h\in \mathcal{P}_\bullet; \, h \in \tx{PSL}(2, \mathbb{R}), \check{m} \sim M \}.
\end{equation}
In other words, every interior phase space is labeled by a conjugacy class of the monodromy
and hence is identified as a right coset space 
\begin{equation}
\label{eq: leaf as coset}
     \mathcal{L}_\bullet \cong \tx{Stab}(M) \backslash \tx{PSL}(2, \mathbb{R}),
\end{equation}
where $\tx{Stab}(M)$ is the stabilizer group
\begin{equation}
\label{eq: Stab def}
    \tx{Stab} (M) \equiv \{ g \in \tx{PSL}(2, \mathbb{R});\, g M g^{-1} = M \}.
\end{equation}

\begin{figure}
    \centering
    \includegraphics[width=0.25\linewidth]{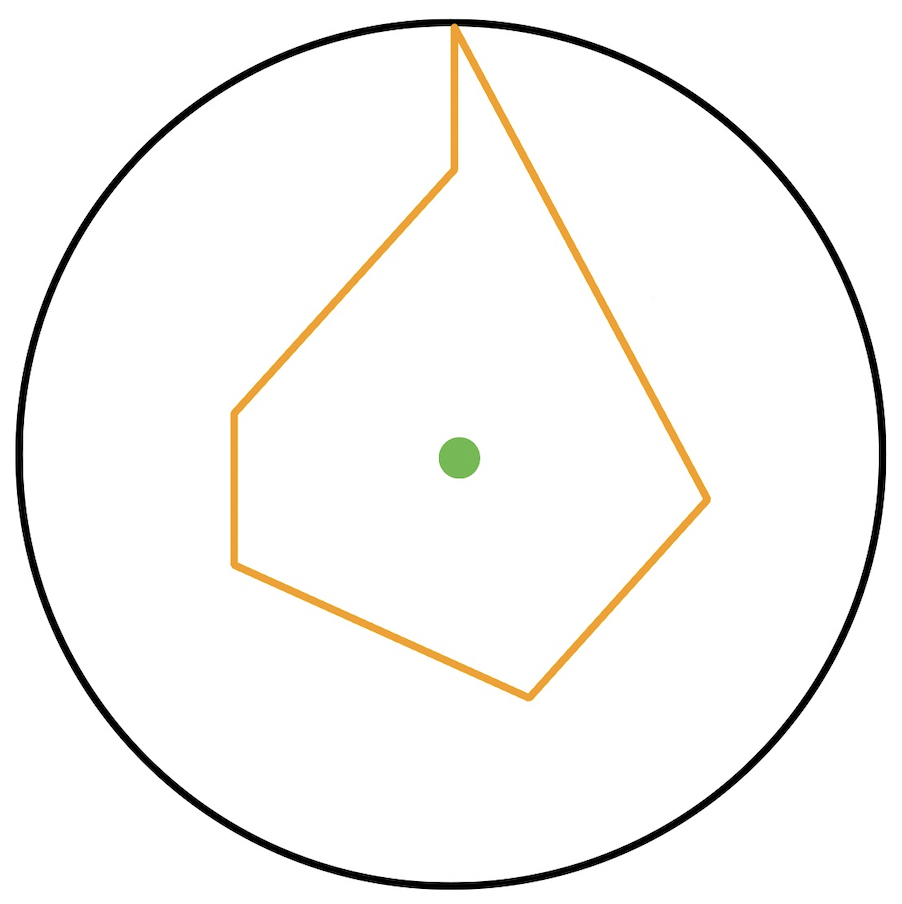}
    \caption{The outer black circle represents the asymptotic boundary. The green dot represents the puncture. The orange polygon represents a Wilson line starting and ending at the same point on the asymptotic boundary and winding around the puncture once. Such a Wilson line is defined as the boundary monodromy~\eqref{eq: M def}.}
    \label{fig: boundary monodromy}
\end{figure}

Besides the puncture large gauge transformation \eqref{eq: puncture gauge transformation}, the $\tx{PSL}(2, \mathbb{R})$ in \eqref{eq: leaf as coset} can also be interpreted as the Wilson line $W_0$ connecting the puncture to the base point $x^+ = 0$ on the boundary circle $\bigcirc$. Since the gauge potential configuration $A$ is fixed on the preimage \eqref{eq: A preimage} by definition \eqref{eq: W to A}, i.e.
\begin{equation}
\label{eq: delta A = 0}
    \delta A = 0,
\end{equation}
all the Wilson lines on $\mathcal{F}^{-1}_{\bdy} (A)$ are determined by $W_0$ via
\begin{equation}
    W|_x = \overleftarrow{P} ( e^{\int_0^x A} ) W_0.
\end{equation}
Hence,
\begin{equation}
\label{eq: F^-1 (A) cong PSL(2, R)}
    \mathcal{F}^{-1}_{\bdy} (A) \cong \tx{PSL}(2, \mathbb{R}).
\end{equation}
The boundary monodromy \eqref{eq: M def} is also fixed on $\mathcal{F}^{-1}_{\bdy} (A)$ due to Eq.~\eqref{eq: delta A = 0}, i.e.
\begin{equation}
\label{eq: delta M = 0}
    \delta M = 0.
\end{equation}
By definition \eqref{eq: m def} and \eqref{eq: M def}, the puncture monodromy $m$ is then parametrized by $W_0$ via
\begin{equation}
\label{eq: m = W_0^-1 M W_0}
    m = W^{-1}_0 M W_0,
\end{equation}
up to a left-$\tx{Stab}(M)$ action 
\begin{align}
\label{eq: left Stab(M) action}
    \tx{Stab}(M) \times \mathcal{F}^{-1}_{\bdy} (A) &\to \mathcal{F}^{-1}_{\bdy} (A)
    \\
    ( g, W_0 ) &\to g W_0.
\end{align}
So the interior phase space $\mathcal{L}_\bullet (A)$ can be identified with the quotient of $\mathcal{F}^{-1}_{\bdy} (A)$ by the left-$\tx{Stab}(M)$ action \eqref{eq: left Stab(M) action}, i.e.
\begin{equation}
\label{eq: left quotient stab (M)}
    \mathcal{L}_\bullet (A) \cong \tx{Stab}(M) \backslash \mathcal{F}^{-1}_{\bdy} (A).
\end{equation}
Combining Eq.~\eqref{eq: F^-1 (A) cong PSL(2, R)} and \eqref{eq: left quotient stab (M)} gives rise to the identification \eqref{eq: leaf as coset}.

So far, we only identified the interior phase space as a manifold \eqref{eq: leaf as coset} while a phase space is a symplectic manifold. Now we turn to derive the symplectic form on it. Denote by $\Omega_\bullet$ the restriction of the symplectic form \eqref{eq: Omega odot} of the full phase space $\mathcal{P}_\odot$ to the preimage \eqref{eq: A preimage}. Using Eq.~\eqref{eq: Omega odot}, \eqref{eq: delta M = 0}, and \eqref{eq: m = W_0^-1 M W_0}, we have
\begin{equation}
   \label{eq: Omega_bullet} \Omega_\bullet = \frac{k}{4 \pi}  \tr (  \delta W_0 W^{-1}_0 M \delta W_0 W^{-1}_0 M^{-1}  +  \delta m_+ m^{-1}_+ \delta m_- m_-^{-1} ).
\end{equation}
$\Omega_{\bullet}$ is invariant under the left-$\tx{Stab}(m)$ action \eqref{eq: left Stab(M) action}, so it is well-defined on the interior phase space $\mathcal{L}_\bullet (A)$ by the identification  \eqref{eq: left quotient stab (M)}. $\Omega_\bullet$ should be a symplectic form because it is reduced from the symplectic form $\Omega_\odot$ to a physical phase subspace. In fact, it was proved that $\Omega_\bullet$ is indeed a symplectic form on the $\tx{PSL}(2, \mathbb{R})$-adjoint orbit \eqref{eq: leaf as coset} \cite{Alekseev:1993qs}.

In general, a (co)adjoint orbit of a Lie group admits a canonical symplectic form called the Kirillov–Kostant–Souriau 2-form \cite{kirillov2025lectures, Kirillov1976ElementsOT}, which is different from the symplectic form $\Omega_\bullet$ given by Eq.~\eqref{eq: Omega_bullet}. However, the interior phase space $\mathcal{L}_\bullet$ admits two interpretations as a gauge orbit and $\Omega_\bullet$ is the canonical symplectic form for the other interpretation which identifies $\mathcal{L}_\bullet$ as an orbit of the so-called dressing transformations \cite{Semenov-Tian-Shansky:1985mgd, lu1990poisson,Alekseev:1993qs}. It reflects an essential difference between the boundary large gauge transformation \eqref{eq: ASG action} and the puncture large gauge transformation \eqref{eq: puncture gauge transformation}, which we will exploit in Section \ref{sec: Nonlocal Symmetries of One-Sided Black Holes}. For now, we only point out that the boundary large gauge transformation \eqref{eq: ASG action} as a symmetry corresponds to a local current while the puncture large gauge transformation \eqref{eq: puncture gauge transformation} does not. In this sense, the puncture large gauge transformation~\eqref{eq: puncture gauge transformation} is a nonlocal symmetry.

\subsection{Classification of Interior Phase Spaces}
\label{sec: Classification of Interior Phase Spaces}

As reviewed in Section \ref{sec: Intrinsic Incompleteness of One-Sided Boundary Gravitons}, the physical phase subspaces of boundary gravitons are the Virasoro coadjoint orbits. They were systematically classified in \cite{Witten:1988hc}. In this subsection, we classify physical phase subspaces associated to the effective interior degrees of freedom. We impose the asymptotically AdS$_3$ boundary condition \eqref{eq: AdS boundary condition 1} such that the gauge potential is characterized by the (chiral) stress tensor $T$. We restrict to orbits with a constant representative stress tensor $T$, which are general enough to include those corresponding to the BTZ black holes. We will see that the classification exhibits a similar pattern to the classification of Virasoro coadjoint orbits.

Given a constant representative boundary gauge potential ( i.e. with constant $T$ in \eqref{eq: Banados potential}), the boundary monodromy \eqref{eq: M def} associated to a contour of constant radial coordinate $r$ is
\begin{equation}
\label{eq: boundary monodromy at r}
    M =  \left(
\begin{array}{cc}
 \cosh \left(4\pi  \sqrt{G_{\tx{N}} T } \right) & \frac{2 \sqrt{G_{\tx{N}} T}}{r}  \sinh \left( 4\pi  \sqrt{G_{\tx{N}} T } \right) \\
  \frac{r}{2 \sqrt{ G_{\tx{N}} T }  } \sinh \left( 4\pi  \sqrt{G_{\tx{N}} T }\right) & \cosh \left( 4\pi  \sqrt{G_{\tx{N}} T } \right) \\
\end{array}
\right).
\end{equation}
The conjugacy classes of $\tx{PSL}(2, \mathbb{R})$ are classified into four types according to the stabilizer group \eqref{eq: Stab def}. $\tx{PSL}(2, \mathbb{R})$ admits the Iwasawa decomposition which is a diffeomorphism induced by group multiplication, i.e.
\begin{equation}
\label{eq: Iwasawa}
    \tx{PSL}(2, \mathbb{R}) = K \times A \times N
\end{equation}
with $K$ the maximal compact subgroup
\begin{equation}
\label{eq: K cong SO(2)}
    K \cong \tx{SO}(2),
\end{equation}
$A$ the split Cartan subgroup
\begin{equation}
\label{eq: A cong SO(1,1)}
    A \cong \tx{SO}(1,1),
\end{equation}
and $N$ the maximal nilpotent subgroup 
\begin{equation}
\label{eq: N cong R}
    N \equiv \{ \begin{pmatrix}
        1 & 0
        \\
        \gamma & 1
    \end{pmatrix} \in \tx{PSL}(2, \mathbb{R});\, \gamma \in \mathbb{R} \}.
\end{equation}
See e.g. \cite{knapp1996lie} for a review of the Iwasawa decomposition. The type of the conjugacy classes of $\tx{PSL}(2, \mathbb{R})$ can be labeled by subgroups $K$, $A$, and $N$. 

Besides the monodromy stabilizer group $\tx{Stab}(M)$, central elements can also classify physical phase subspaces (i.e. symplectic leaves) since they correspond to quantum numbers labeling superselection sectors \cite{Kirillov1976ElementsOT}. As we explain now, they are Wilson loops. The equation of motion~\eqref{eq: flatness 1} of the Chern-Simons theory
shows that the gauge potential $A$ is a flat connection. Then the trace of the puncture monodromy \eqref{eq: m def}, as a Wilson loop around the puncture, is equal to the Wilson loop located at the asymptotic boundary $\bigcirc$, i.e.
\begin{equation}
    \tr\, m = \tr \, \overleftarrow{P} \exp ( \oint_{\bdy} A ) = \tr\, M.
\end{equation}
The second equality holds by definition \eqref{eq: M def}. Hence the Wilson loop $\tr\, m$ is a central element due to the locality condition \eqref{eq: locality}, i.e.
\begin{equation}
\label{eq: Casimir m}
    \{ \tr\, m, m\} = 0.
\end{equation}
For the same reason, a Wilson loop of winding number $w$ is also central,
\begin{equation}
\label{eq: Casimir m^k}
    \{ \tr\, m^w, m\} = 0, \quad \forall w \in \mathbb{Z}.
\end{equation}
Eq.~\eqref{eq: Casimir m} and \eqref{eq: Casimir m^k} can also be verified by a straightforward calculation using Eq.~\eqref{eq: STS Poisson}. In Section \ref{sec: Gauging Nonlocal Symmetries}, we will show that the center of the observable algebra of the black hole interior is generated by Wilson loops $\tr\, m^w$. Using Eqs.~\eqref{eq: trace normalization}, \eqref{eq: m = W_0^-1 M W_0}, and \eqref{eq: boundary monodromy at r}, we have
\begin{equation}
\label{eq: tr m^k}
    \tr\, m^w = 2|\cosh \left(4 w \pi  \sqrt{G_{\tx{N}} T } \right) |.
\end{equation}
The right hand side of Eq.~\eqref{eq: tr m^k} is an absolute value because it is a trace of an element in $\tx{PSL}(2, \mathbb{R}) \cong \tx{SL}(2, \mathbb{R})/\mathbb{Z}_2$. Note that the Wilson loops \eqref{eq: tr m^k} of different winding numbers are functionally dependent, so we only need to consider
\begin{equation}
\label{eq: tr m = cosh}
        \tr\, m  = 2|\cosh \left(4\pi  \sqrt{G_{\tx{N}} T } \right) |.
\end{equation}

\subsubsection{Trivial Class}
\label{sec: Trivial class}
    
    If the chiral stress tensor satisfies
    \begin{equation}
    \label{eq: trivial T}
        T = - \frac{n^2}{16 G_{\tx{N}}}\ne 0, \quad n \in \mathbb{Z},
    \end{equation}
    then \begin{equation}
        \tx{Stab}(M) = \tx{PSL}(2, \mathbb{R}).
    \end{equation}
    Hence $\mathcal{L}_\bullet$ is a single point due to Eq.~\eqref{eq: leaf as coset}, i.e. no degree of freedom in the bulk. We refer to it as the trivial class. The bulk degrees of freedom disappear because the puncture monodromy is fixed to be the identity by the condition \eqref{eq: trivial T}. In other words, the condition \eqref{eq: trivial T} reduces the full bulk phase space $\mathcal{P}_\odot$ to be the moduli space of flat $\tx{PSL}(2, \mathbb{R})$-connections on a disc without the puncture. From this perspective, the nonzero integer $n$ in Eq.~\eqref{eq: trivial T} can be interpreted as the Euler class of the gauge potential $A$. The Euler class is essentially an element of the fundamental group of the gauge group \cite{milnor1958existence}. To have an intuitive understanding, we consider a Wilson line with one end moving along a circle $C$ of constant radial coordinate $r$,
    \begin{equation}
    \label{eq: W(theta)}
        W(\theta) = \overleftarrow{P} \exp ( \int_{0}^\theta A).
    \end{equation}
    This Wilson line can be viewed as a map from the circle $C$ to a path in the gauge group $\tx{PSL}(2, \mathbb{R})$, i.e. 
    \begin{align}
        W: [0,2\pi] &\to \tx{PSL}(2, \mathbb{R})
        \label{eq: W(theta) as path map}
        \\
        \theta &\to W(\theta).
    \end{align}
    See Figure~\ref{fig: Euler class} for a illustration. Using Eqs. \eqref{eq: Banados potential}, \eqref{eq: trivial T}, and \eqref{eq: W(theta)}, we have
    \begin{equation}
    \label{eq: W(theta) in a loop}
        W(\theta) = \left(
\begin{array}{cc}
 \sqrt{\frac{n}{2 r}} & 0 \\
 0 & \sqrt{\frac{2 r}{n}} \\
\end{array}
\right) \left(
\begin{array}{cc}
 \cos \left(\frac{n \theta}{2} \right) & -  \sin \left( \frac{n \theta}{2} \right) \\
  \sin \left( \frac{n \theta}{2}\right) & \cos \left( \frac{n \theta}{2} \right) \\
\end{array}
\right) \left(
\begin{array}{cc}
 \sqrt{\frac{2 r}{n}} & 0 \\
 0 & \sqrt{\frac{n}{2 r}} \\
\end{array}
\right) .
    \end{equation}
    So $W(\theta)$ lies in a maximal compact subgroup $\tx{SO}(2)$ which is a non-contractible loop in the gauge group $\tx{PSL}(2, \mathbb{R})$ \footnote{The subgroups $A$ \eqref{eq: A cong SO(1,1)} and $N$ \eqref{eq: N cong R} in the Iwasawa decomposition \eqref{eq: Iwasawa} are simply-connected, so $\tx{PSL}(2, \mathbb{R})$ is homotopy equivalent to $K \cong \tx{SO}(2)$ which is non-contractible.}. From decomposition \eqref{eq: W(theta) in a loop}, one can see that the Wilson line \eqref{eq: W(theta) as path map} maps the circle $C$ on the Cauchy slice to a path winding around a non-contractible loop in $\tx{PSL}(2, \mathbb{R})$ $n$ times \footnote{Notice that $\begin{pmatrix}
        1 &0
        \\
        0 &1
    \end{pmatrix}$ and $\begin{pmatrix}
        -1 &0
        \\
        0 & -1
    \end{pmatrix}$ are identified in $\tx{PSL}(2, \mathbb{R})$.}. Thus, the nonzero integer $n$ in Eq.~\eqref{eq: trivial T} is identified as the winding number which is an element in the fundamental group $\pi_1(\tx{PSL}(2, \mathbb{R})) \cong \mathbb{Z}$ of the gauge group.

    \begin{figure}
        \centering
        \includegraphics[width=0.5\linewidth]{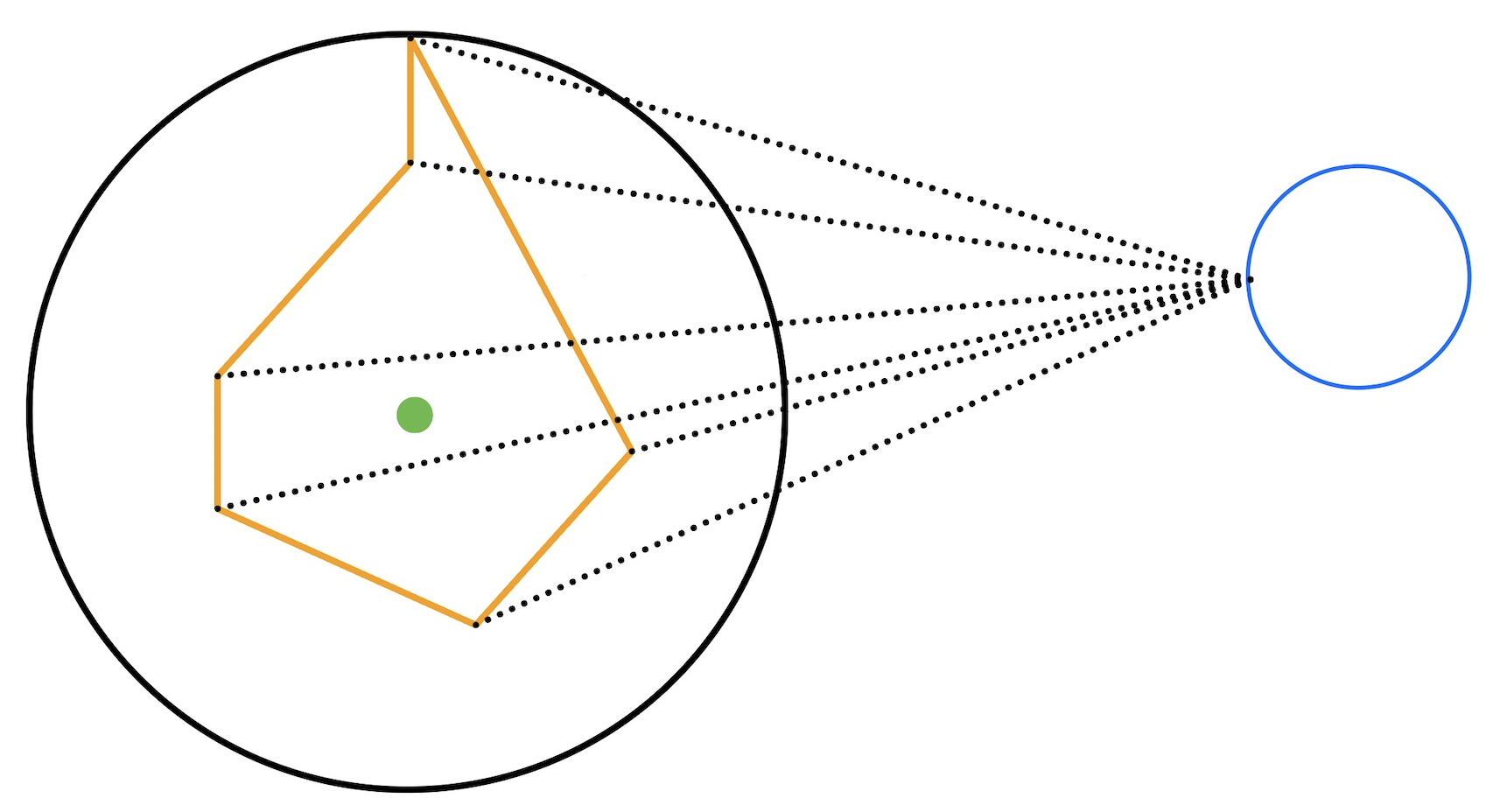}
        \caption{The outer black circle represents the asymptotic boundary. The green dot represents the puncture. The orange polygon represents the boundary monodromy. The blue circle represents a non-contractible circle in $\PSL$. The Wilson line maps a path on the punctured disc to a path in the gauge group $\PSL$. Each segment of the orange polygon is mapped to a path winding around the blue circle once. For the clarity of the illustration, we draw it such that each vertex of the orange polygon is mapped to the same point on the blue circle. The map is represented by six black dashed lines. This picture describes the case where the Euler class $n$ is equal to 6.}
        \label{fig: Euler class}
    \end{figure}

    The Euler class $n$ (i.e. the winding number of the map \eqref{eq: W(theta) as path map}) of flat $\tx{PSL}(2, \mathbb{R})$-connections on a Riemann surface $S$ of \textit{negative} Euler characteristic $\chi$ is bounded as \cite{milnor1958existence, wood1971bundles, goldman1988topological}
    \begin{equation}
    \label{eq: Milnor-Wood}
        |n| \leq -\chi.
    \end{equation}
    The phase space of the $\tx{PSL}(2, \mathbb{R})$ Chern-Simons theory is composed of disconnected components that are labeled by the Euler class $n$. The components satisfying
    \begin{equation}
    \label{eq: Teich component}
        n = \pm \chi
    \end{equation}
    are called the Teichm\"uller components since they are isomorphic to the Teichm\"uller space of hyperbolic metrics on $S$ \cite{goldman1988topological}. The Teichm\"uller component (combined with its anti-chiral counterpart and with the mapping class group modded out) is identified as the phase space of (2+1)-d pure gravity with negative cosmological constant \cite{mess2007lorentz, krasnov2007minimal, Scarinci:2011np, Kim:2015qoa, Maloney:2015ina, Eberhardt:2022wlc}. However, the Milnor-Wood inequality \eqref{eq: Milnor-Wood} does not hold in our case. The Euler class $n$ is unbounded as shown in Eq.~\eqref{eq: trivial T}. This is because the Euler characteristic of a disc is 1 which is \textit{positive}. In this case, disconnected components of Euler class
    \begin{equation}
    \label{eq: n = -+1}
        n = \pm1
    \end{equation}
    can still be interpreted as the Teichm\"uller components according to Eq.~\eqref{eq: Teich component}. Because they (combined with the anti-chiral counterpart) correspond to the smooth global AdS$_3$ spacetime. Using Eq.~\eqref{eq: trivial T}, the (chiral part of) the vacuum energy is given by
    \begin{equation}
    \label{eq: vacuum energy}
        L_0 = -\frac{1}{16\G}. 
    \end{equation}

\subsubsection{Hyperbolic Class}
\label{sec: Hyperbolic class}

    If the chiral stress tensor satisfies
    \begin{equation}
    \label{eq: T > 0}
        T >0,
    \end{equation}
    then the stabilizer group $\tx{Stab}(M)$ of the boundary monodromy \eqref{eq: boundary monodromy at r} is isomorphic to the split Cartan subgroup \eqref{eq: A cong SO(1,1)}, i.e.
\begin{equation}
        \tx{Stab}(M) \cong A.
    \end{equation}
    According to Eq.~\eqref{eq: leaf as coset} and \eqref{eq: A cong SO(1,1)}, the interior phase space \eqref{eq: L (A) leaf} is identified as 
    \begin{equation}
    \label{eq: leaf cong PSL(2, R)/SO(1,1)}
       \mathcal{L}_\bullet \cong \tx{SO}(1,1) \backslash \tx{PSL}(2, \mathbb{R})  .
    \end{equation}
    $\mathcal{L}_\bullet$ is diffeomorphic to the AdS$_2$ space (i.e. a cylinder) since it is a homogeneous space with isometry group $\tx{PSL}(2, \mathbb{R})$ and stabilizer group $\tx{SO}(1,1)$. Eq.~\eqref{eq: tr m = cosh} with $T>0$ implies
    \begin{equation}
        \tr\, m >2,
    \end{equation}
    so we refer to the interior phase spaces \eqref{eq: leaf cong PSL(2, R)/SO(1,1)} as the hyperbolic class. The interior phase spaces of the hyperbolic class correspond to BTZ black holes (assuming the anti-chiral sector is also of the hyperbolic class). Combining Eq.~\eqref{eq: vacuum energy} and \eqref{eq: T > 0} implies that observables on phase spaces in the hyperbolic class are heavy, i.e. the energy is at least of order $T \sim O(\G^{-1})$. In \cite{Schlenker:2022dyo}, it was argued that such heavy observables may behave erratically in the large $N$ limit ($\G \to 0$), which may lead to an apparent ensemble averaging. We will make this proposal more explicit in Section~\ref{sec: Apparent Ensemble Averaging over Entanglers}.

\subsubsection{Parabolic Class}
\label{sec: Parabolic Class}

    If the chiral stress tensor $T = 0$, then the stabilizer group $\tx{Stab}(M)$ of the monodromy \eqref{eq: boundary monodromy at r} is isomorphic to the maximal nilpotent subgroup \eqref{eq: N cong R}, i.e.
    \begin{equation}
        \tx{Stab}(M) \cong N .
    \end{equation}
    According to Eq.~\eqref{eq: leaf as coset}, the interior phase space \eqref{eq: L (A) leaf} is identified as 
    \begin{equation}
    \label{eq: leaf cong PSL(2, R)/N}
        \mathcal{L}_\bullet \cong N \backslash \tx{PSL}(2, \mathbb{R}).
    \end{equation}
    Eq.~\eqref{eq: tr m = cosh} with $T=0$ implies
    \begin{equation}
        |\tr\, m| = 2,
    \end{equation}
    so we refer to the interior phase space \eqref{eq: leaf cong PSL(2, R)/N} as the parabolic class. The parabolic class may be the most subtle class. It (combined with its anti-chiral counterpart) corresponds to a massless spinless BTZ black hole. Given this singular feature at the classical level, it is unclear whether it is still a meaningful notion after quantization. We won't discuss the parabolic class further and only remark that the Wilson loop does not distinguish it from the trivial class.

\subsubsection{Elliptic Class}

    If the chiral stress tensor satisfies
    \begin{equation}
    \label{eq: elliptic T}
        T<0, \qquad T \ne - \frac{n^2}{16 G_{\tx{N}}}, \,\forall n \in \mathbb{N}^+,
    \end{equation}
    then the stabilizer group $\tx{Stab}(M)$ of the boundary monodromy \eqref{eq: boundary monodromy at r} is isomorphic to the maximal compact subgroup \eqref{eq: K cong SO(2)}, i.e.
    \begin{equation}
        \tx{Stab}(M) \cong K.
    \end{equation}
    According to Eq.~\eqref{eq: leaf as coset} and \eqref{eq: K cong SO(2)}, the interior phase space \eqref{eq: L (A) leaf} is identified as 
    \begin{equation}
    \label{eq: leaf cong PSL(2, R)/SO(2)}
       \mathcal{L}_\bullet \cong \tx{SO}(2) \backslash \tx{PSL}(2, \mathbb{R})  .
    \end{equation}
    $\mathcal{L}_\bullet$ is diffeomorphic to the Poincar\'e disc since it is a homogeneous space with isometry group $\tx{PSL}(2, \mathbb{R})$ and stabilizer group $\tx{SO}(2)$. Eq.~\eqref{eq: tr m = cosh} and \eqref{eq: elliptic T} imply
    \begin{equation}
        \tr \, m <2,
    \end{equation}
    so we refer to the interior phase spaces \eqref{eq: leaf cong PSL(2, R)/SO(2)} as the elliptic class. The interior phase spaces (combined with the anti-chiral sector) of the elliptic class correspond to conical defects.

\subsection{Positivity Restrictions from Lorentzian Wormholes}
\label{sec: Positivity Restrictions Lorentzian Wormholes}

In this subsection, we show that the effective black hole interior degrees of freedom \eqref{eq: m def} admit a geometric interpretation from the perspective of observers outside horizons of Lorentzian wormholes. This interpretation results in positivity restrictions on the interior phase spaces, which will be exploited in Section~\ref{sec: Nonlocal observable Algebra with Erratic $N$-Dependence} to illustrate emergent observables with the erratic $N$-dependence. 

There are Lorentzian multi-boundary wormhole solutions~\cite{Aminneborg:1997pz, Aminneborg:1998si, Brill:1998pr, barbot2008causal, Barbot:2005qk} to the vacuum Einstein equation~\eqref{eq: Vacuum Einstein} with the asymptotically AdS$_3$ boundary conditions~\eqref{eq: AdS boundary condition 1} and \eqref{eq: AdS boundary condition 2}. Such a wormhole spacetime $M$ is homeomorphic to $S_{g,n} \times \mathbb{R}$ with $S_{g,n}$ a surface of genus $g$ and with $n$ boundary components. It is composed of $n$ exterior regions (i.e. regions outside horizons) and an interior region (i.e. a region behind all horizons). Each exterior region is isometric to the exterior region of a one-sided black hole~\cite{Aminneborg:1997pz, Brill:1998pr, Barbot:2005qk}. The degrees of freedom are divided into local part and global part. The local part is the boundary gravitons associated to the exterior regions. The global part is characterized by Wilson lines winding around the interior region. More precisely, there exists a domain $D$ in the AdS$_3$ spacetime such that the wormhole spacetime $M$ is a quotient of $D$ by the action of a discrete subgroup $\Gamma$ of the isometry group $\PSL \times \PSL$ of AdS$_3$, i.e.
\begin{equation}
\label{eq: multiboundary wormholes}
    M = D/\Gamma.
\end{equation}
The discrete subgroup $\Gamma$ is the image of the Wilson lines $W$,
\begin{align}
    W: \pi_1 (M) &\to \Gamma \subset \PSL \times \PSL
    \\
    \gamma &\to W[\gamma] = ( m, \overline{m} ).
\end{align}
$m, \overline{m} \in \PSL$. $\pi_1(M)$ is the fundamental group of the wormhole spacetime $M$.~\footnote{The fundamental group of a manifold $M$ is essentially the group of topological loops on $M$ with a common base point $P$. Every loop is oriented and hence has an incoming end and an outgoing end at $P$. Two loops $a$ and $b$ can compose into another loop $ab$ by gluing the outgoing end of $b$ with the incoming end of $a$. The group multiplication is the composition of loops. The inverse of a loop is the same loop with opposite orientation. Contractible loops are equivalent to the identity.} $W[\gamma]$ is the Wilson line evaluated along a loop $\gamma \subset M$ starting and ending at a given base point $P\in M$. A gauge transformation then maps $\Gamma$ to another discrete subgroup $h \Gamma h^{-1}$ with $h \in \PSL \times \PSL$ the gauge transformation at the base point $P$. It does not change the geometry, so the wormhole $M$ is determined by $\Gamma$ up to a $\PSL \times \PSL$-adjoint action. 

We define the phase space $\mathcal{P}_{\tx{wormhole}}^{\tx{interior}}$ of the interior region of the wormhole $M$ as the phase subspace invariant under the asymptotic symmetries on all the asymptotic boundaries. Note that the wormhole spacetime $M$ is homotopy equivalent to the surface $S_{g, n}$. The wormhole interior phase space $\mathcal{P}_{\tx{wormhole}}^{\tx{interior}}$ is hence parametrized by flat $\PSL \times \PSL$-connections on $S_{g, n}$. Only the flat connections with the Euler class saturating the Milnor-Wood inequality~\eqref{eq: Milnor-Wood} correspond to smooth geometries. Similar to the case of a disc in Section~\ref{sec: Trivial class}, the Euler class of flat $\PSL\times \PSL$-connections on general Riemann surfaces can be defined as a pair of winding numbers. So the wormhole interior phase space $\mathcal{P}_{\tx{wormhole}}^{\tx{interior}}$ is identified with a Teichm\"uller component of the moduli space of flat $\PSL \times \PSL$-connections on $S_{g, n}$, which is equivalent to two copies of the Teichm\"uller space $\mathcal{T}(S_{g,n})$ of hyperbolic structures on $S_{g, n}$~\cite{goldman1988topological}. Every wormhole horizon corresponds to two boundaries associated with the two hyperbolic structures on $S_{g,n}$. Without loss of generality, we assume the boundaries are geodesics. See \cite{Scarinci:2011np} for a more compact description of such wormholes. 

In Eq.~\eqref{eq: interior algebra = monodromy algebra}, the puncture monodromy~\eqref{eq: m def} is interpreted as the effective interior degrees of freedom of the one-sided black hole from the perspective of an infalling observer who can access the black hole interior. In the spirit of the black hole complementarity~\cite{Susskind:1993if,Lowe:1995ac}, the puncture monodromy~\eqref{eq: m def} admits another interpretation complementary to the infalling observer perspective. For an exterior observer outside the horizon, the puncture monodromy~\eqref{eq: m def}, as an effective description of the black hole, should be interpreted as the effective degrees of the freedom of the stretched horizon. We call it the horizon edge modes since they are essentially equivalent to the edge modes used to factorize the bulk Hilbert space of (2+1)-d gravity~\cite{Mertens:2022ujr}.

Now we turn to explain the geometric meaning of the horizon edge modes. A (decorated) Teichm\"{u}ller space can be parametrized by the Kashaev coordinates together with a set of homology constraints \cite{Kashaev:1998fc}. Every geodesic boundary component of a hyperbolic surface is characterized by 6 Kashaev coordinates that can be effectively reduced to three independent variables $(x, p , \lambda) \in \mathbb{R} \times \mathbb{R} \times \mathbb{R}^+$ \cite{Teschner:2005bz, Nidaiev:2013bda}. The Kashaev coordinates admit a Poisson structure which gives
\begin{equation}
\label{eq: classical Heisenberg}
    \{ x, p \} = \frac{1}{2 \pi}, \quad \{ \lambda, x \} = 0, \quad \{ \lambda, p \} = 0,
\end{equation}
for each geodesic boundary. It was proven that the Poisson structure of the Kashaev coordinates is equivalent to the Weil-Peterson form on the Teichm\"uller space~\cite{Teschner:2005bz}. It is also known that the Weil-Peterson form is equivalent to the Atiyah-Bott~\cite{Atiyah:1982fa} form on the Teichm\"uller component of the moduli space of $\PSL$ flat connections~\cite{Goldman:1984bow}. The Atiyah-Bott form is the symplectic form determined by the Lagrangian~\eqref{eq: chern-simons lagrangian} of the Chern-Simons theory defined on closed Riemann surfaces. So the Poisson bracket of Wilson loops is equivalent to that of the Kashaev coordinates. Taking the geodesic boundaries into account, there should also be a Poisson map (i.e. a map preserving Poisson brackets) relating the puncture monodromy~\eqref{eq: m def} and the Kashaev coordinates describing the same geodesic boundary. Using Eqs.~ \eqref{eq: k = 1/4G_N}, \eqref{eq: sl2 r-matrix}, \eqref{eq: sl2 K}, \eqref{eq: m_e,f,k parameterization}, \eqref{eq: STS Poisson}, and \eqref{eq: classical Heisenberg}, it is straightforward to check that 
\begin{align}
    p &= \frac{1}{2\pi \sqrt{\G}} \log m_{\mathrm{K}},
    \label{eq: p def}
    \\
    x&= \frac{1}{4 \pi \sqrt{\G}} \log \frac{m_{\mathrm{K}}}{m_{\mathrm{e}}},
    \label{eq: x def}
    \\
    \lambda&= \frac{1}{4\pi \sqrt{\G}} \tx{arccosh} ( \frac{1}{2} \tr\, m),
    \label{eq: lambda  def}
\end{align}
is the desired Poisson map. This map is fixed up to an automorphism preserving the Poisson structure~\eqref{eq: classical Heisenberg}. Since the Wilson loop $\tr \,m$ is a central element, $\lambda$ is also a central element because of Eq.~\eqref{eq: lambda  def}. Every physical phase subspace \eqref{eq: L (A) leaf} hence corresponds to a fixed value of $\lambda$. Then $x$ and $p$ form a pair of canonically conjugate variables for each physical phase subspace. The inverse maps of \eqref{eq: p def}, \eqref{eq: x def}, and \eqref{eq: lambda  def} are given by 
\begin{align}
    m_{\mathrm{K}} &= e^{2 \pi \sqrt{\G} p},
    \label{eq: mK Kashaev}
    \\
    m_{\mathrm{e}} &= e^{2 \pi \sqrt{\G} (p-2x)},
    \label{eq: me Kashaev}
    \\
    m_{\mathrm{f}} &=
    \left(2 \cosh (4\pi \sqrt{\G} \lambda) + 2 \cosh(4\pi \sqrt{\G} p)\right) e^{2\pi \sqrt{\G} (2x-p )}.
    \label{eq: mf Kashaev}
\end{align}
A crucial feature is that the monodromy variables only need to range over
\begin{equation}
\label{eq: positivity m_e, m_f, m_K}
    m_{\mathrm{e}}>0, \quad m_{\mathrm{f}} > 0, \quad m_{\mathrm{K}} >0,
\end{equation}
to describe the Teichm\"uller space. In Section~\ref{sec: Nonlocal observable Algebra with Erratic $N$-Dependence}, we will show that the positivity restrictions~\eqref{eq: positivity m_e, m_f, m_K} gives rise to an emergent extension of the monodromy algebra with an erratic $N$-dependence.

Note that the monodromies associated with $n$ horizons are not enough to describe the interior region of the wormhole. This may be interpreted as the same kind of information loss mentioned in Section~\ref{sec: Introduction}. The wormholes are the smooth remnants of correlations among erratic quantum degrees of freedom which cannot be characterized by smooth observables. In Section~\ref{sec: Emergent Wormholes from the Large-$N$ Filter}, we will provide a quantum mechanical interpretation of such wormholes from the perspective of the large-$N$ filter.

\section{Loss of Erratic Information in the Large $N$ Limit}
\label{ref: Loss of Erratic Information in the Large $N$ Limit}

\subsection{Quantum Nonlocal Observable Algebra}
\label{sec: Quantum Nonlocal Observable Algebra}

In Section~\ref{sec: Positivity Restrictions Lorentzian Wormholes}, the puncture monodromy~\eqref{eq: m def} is identified with the horizon edge modes from the perspective of an exterior observer outside the horizon. According to classical causal structure, it will take the exterior observer an infinitely long time to access the horizon edge modes. However, it is expected that the classical spacetime breaks down on a timescale exponential in black hole entropy due to non-perturbative quantum gravitational effects~\cite{Susskind:2015toa,Susskind:2020wwe, Iliesiu:2021ari}. So the puncture monodromy~\eqref{eq: m def} is at least exponentially hard to access for a boundary observer. Hence we may identify the quantization $\mathcal{A}_\bullet$ of the classical observable algebra $C^\infty (\mathcal{P}_\bullet)$ of the puncture monodromy \eqref{eq: m def} with the algebra $\mathcal{A}_\cft^{\tx{smooth}}$ of smooth observables of high complexity, i.e.
\begin{equation}
\label{eq: Asmooth= quantum monodromy algebra}
    \mathcal{A}_\cft^{\tx{smooth}} = \mathcal{A}_\bullet.
\end{equation}
The smoothness of $\mathcal{A}_\bullet$ inherits from the smoothness of $C^\infty (\mathcal{P}_\bullet)$. 

Combining Eqs.~\eqref{eq: m_e,f,k parameterization} and \eqref{eq: STS Poisson} gives rise to the classical monodromy algebra $C^\infty (\mathcal{P}_\bullet)$,
\begin{align}
    \{ m_{\mathrm{K}}, m_{\mathrm{e}} \} &=  \frac{\pi}{k} m_{\mathrm{e}} m_{\mathrm{K}},
    \label{eq: K, e Poisson}
    \\
    \{ m_{\mathrm{K}}, m_{\mathrm{f}} \} &= -\frac{ \pi}{k} m_{\mathrm{f}} m_{\mathrm{K}},
    \label{eq: K, f Poisson}
    \\
    \{ m_{\mathrm{e}}, m_{\mathrm{f}} \} &= \frac{2\pi}{k} ( m_{\mathrm{K}}^{-2} - m_{\mathrm{K}}^2).
    \label{eq: e, f, Poisson}
\end{align}
Poisson brackets~\eqref{eq: K, e Poisson}, \eqref{eq: K, f Poisson}, and \eqref{eq: e, f, Poisson} is the classical limit ($\hbar \to 0$) of the Drinfel'd-Jimbo quantum group algebra $U_q (\mathfrak{sl}(2,\mathbb{R}))$ \cite{Jimbo:1985zk, drinfeld1986quantum}
\begin{align}
    m_{\mathrm{K}} \,m_{\mathrm{e}} &= q\, m_{\mathrm{e}}\, m_{\mathrm{K}},
    \label{eq: m_K m_e = q m_e m_K}
    \\
    m_{\mathrm{K}}\, m_{\mathrm{f}} &= q^{-1} m_{\mathrm{f}}\, m_{\mathrm{K}},
    \label{eq: m_K m_f = q m_f m_K}
    \\
    [m_{\mathrm{e}}, m_{\mathrm{f}}] &=(q-q^{-1}) (m_{\mathrm{K}}^{-2} - m_{\mathrm{K}}^2),
    \label{eq: [m_e, m_f]}
\end{align}
with $q$ defined by
\begin{equation}
    q \equiv \exp ( \frac{\pi i \hbar}{2 \hbar-k} ).
\end{equation}
More precisely, Eq.~\eqref{eq: m_K m_e = q m_e m_K} implies
\begin{equation}
    [m_{\mathrm{K}}, m_{\mathrm{e}}] = -\frac{\pi i \hbar}{k} m_{\mathrm{e}} m_{\mathrm{K}} + O(\hbar^2)
\end{equation}
with the leading term on the right hand side given by the Poisson bracket~\eqref{eq: K, e Poisson} (multiplied by the standard quantization factor ``$-i\hbar$''). The other two Poisson brackets \eqref{eq: K, f Poisson} and \eqref{eq: e, f, Poisson} are reproduced similarly from Eqs.~\eqref{eq: m_K m_f = q m_f m_K} and \eqref{eq: [m_e, m_f]} in the classical limit ($\hbar \to 0$). Notice that the classical limit refers to the limit ``$\lim_{\hbar \to 0} i\partial_\hbar$'' relating the commutators of quantum operators to the Poisson brackets, i.e.
\begin{equation}
\label{eq: classical limit}
    \{ \mathcal{O}_1, \mathcal{O}_2 \} = \lim_{\hbar \to 0} i\partial_\hbar [ \mathcal{O}_1, \mathcal{O}_2 ].
\end{equation}
While taking the large $N$ limit is taking the limit ``$\lim_{\G \to 0}$'' for observables.

The quantum monodromy algebra $\mathcal{A}_\bullet$ is defined as a quantization $U_q (\psl)$ of the classical monodromy algebra $C^\infty(\mathcal{P}_\bullet)$. In the large $N$ limit ($k \to \infty$), the quantum monodromy algebra $\mathcal{A}_\bullet$ reduces to a trivial commutative algebra. There is also a nontrivial double-scaling large $N$ limit ($k \to \infty$) which keeps the following rescaled variables finite,
\begin{equation}
\label{eq: resacled m def}
    m_{\mathrm{E}} \equiv \frac{m_{\mathrm{e}}}{2 \sin (\frac{\pi\hbar}{k-2\hbar})},
    \quad
    m_{\mathrm{F}} \equiv \frac{m_{\mathrm{f}}}{2 \sin (\frac{\pi\hbar}{k-2\hbar})},
    \quad
    m_{\mathrm{H}} \equiv i\log_q m_{\mathrm{K}}^2.
\end{equation}
In terms of the rescaled variables~\eqref{eq: resacled m def}, the quantum monodromy algebra \eqref{eq: m_K m_e = q m_e m_K}, \eqref{eq: m_K m_f = q m_f m_K}, and \eqref{eq: [m_e, m_f]} is rewritten as
\begin{align}
    q^{-\frac{i}{2} m_{\mathrm{H}} }  m_{\mathrm{E}}\, q^{\frac{i}{2} m_{\mathrm{H}} }&= q\, m_{\mathrm{E}}  ,
    \label{eq: sl_q(2, R) mE}
    \\
    q^{-\frac{i}{2} m_{\mathrm{H}} }  m_{\mathrm{F}}\, q^{\frac{i}{2} m_{\mathrm{H}} } &= q^{-1} m_{\mathrm{F}}  ,
    \label{eq: sl_q(2, R) mF}
    \\
    [m_{\mathrm{E}}, m_{\mathrm{F}}] &=\frac{q^{-im_{\mathrm{H}}}  - q^{i m_{\mathrm{H}}} }{q-q^{-1}} .
    \label{eq: sl_q(2, R) mH}
\end{align}
Eqs.~\eqref{eq: sl_q(2, R) mE} and \eqref{eq: sl_q(2, R) mF} are equivalent to
\begin{align}
    [m_{\mathrm{H}}, m_{\mathrm{E}}] &= 2i\, m_{\mathrm{E}},
    \\
   [m_{\mathrm{H}}, m_{\mathrm{F}}] &= -2i\, m_{\mathrm{F}}.
\end{align}
In the large $N$ limit ($k \to \infty$), Eq.~\eqref{eq: sl_q(2, R) mH} reduces to
\begin{align}
[m_{\mathrm{E}}, m_{\mathrm{F}}] &= -i\,m_{\mathrm{H}}.
\end{align}
$U_q(\psl)$ hence reduces to the matrix algebra $\psl$ defined by matrices~\eqref{eq: sl2 generators} in the large $N$ limit. According to parametrization~\eqref{eq: m_e,f,k parameterization}, $m_{\mathrm{e}}$, $m_{\mathrm{f}}$, $m_{\mathrm{K}}^2$ are real numbers at the classical level. So the corresponding quantum operators are hermitian,
\begin{equation}
\label{eq: star structure}
    m_{\mathrm{e}}^\dagger = m_{\mathrm{e}}, \quad m_{\mathrm{f}}^\dagger = m_{\mathrm{f}}, \quad (m_{\mathrm{K}}^2)^\dagger = m_{\mathrm{K}}^2. 
\end{equation}
Eqs.~\eqref{eq: resacled m def} and \eqref{eq: star structure} imply 
\begin{equation}
    m_{\mathrm{E}}^\dagger = m_{\mathrm{E}}, \quad m_{\mathrm{F}}^\dagger = m_{\mathrm{F}}, \quad m_{\mathrm{H}}^\dagger = m_{\mathrm{H}}. 
\end{equation}

In Eq.~\eqref{eq: interior algebra = monodromy algebra}, the classical observable algebra $C^\infty (\mathcal{P}_\ads^{\tx{interior}})$ of the black hole interior is identified with the classical monodromy algebra $C^\infty (\mathcal{P}_\bullet)$. According to the Schlenker-Witten proposal~\cite{Witten:2021jzq, Schlenker:2022dyo}, some observables involving black hole states do not have a large $N$ limit. However, one may expect that quantizing a classical algebra will give rise to a quantum algebra with a classical limit. So it seems to be unlikely to see erratic $N$-dependence from the quantum monodromy algebra $\mathcal{A}_\bullet$. In Section~\ref{sec: Nonlocal observable Algebra with Erratic $N$-Dependence}, we will show that the quantum monodromy algebra $\mathcal{A}_\bullet$ is automatically extended under the positivity restriction~\eqref{eq: positivity m_e, m_f, m_K} such that an erratic $N$-dependence emerges.

\subsection{Nonlocal Observable Algebra with Erratic $N$-Dependence}
\label{sec: Nonlocal observable Algebra with Erratic $N$-Dependence}

In Section~\ref{sec: Positivity Restrictions Lorentzian Wormholes}, we showed that the puncture monodromy~\eqref{eq: m def} satisfying the positivity restriction~\eqref{eq: positivity m_e, m_f, m_K} is sufficient to describe the horizon. In this subsection, we show that the positivity restriction~\eqref{eq: positivity m_e, m_f, m_K} induces an extension of the quantum monodromy algebra. The extended algebra contains a subalgebra ill-defined in the large $N$ limit. 

In general, new structures may emerge under restrictions. We start with the following observation. Given a subspace $\mathcal{M}$ of a phase space $\mathcal{N}$, it is trivially true that an observable on $\mathcal{N}$ only involving degrees of freedom in $\mathcal{M}$ is also an observable on $\mathcal{M}$, i.e.
\begin{equation}
\label{eq: emergence under restriction}
    C^\infty (\mathcal{\mathcal{N}})|_{\mathcal{M}} \subset C^\infty (\mathcal{M}).
\end{equation}
$C^\infty(\mathcal{M})$ denotes the space of smooth functions on $\mathcal{M}$. $C^\infty (\mathcal{N})|_{\mathcal{M}}$ denotes the restriction of $C^\infty(\mathcal{N})$ to $\mathcal{M}$. But the converse of \eqref{eq: emergence under restriction} is not necessarily true. Hence restricting to the phase subspace $\mathcal{M}$ may automatically give rise to an extension $C^\infty (\mathcal{M})$ of the original observable algebra $C^\infty (\mathcal{N})$. It is indeed the case for the classical monodromy algebra $C^\infty (\mathcal{P}_\bullet)$ with the positivity restrictions~\eqref{eq: positivity m_e, m_f, m_K}. 

Before proceeding, we briefly clarify the relation between classical observables and smooth functions on phase space. In principle, infinity is not measurable, so the classical observable algebra should be composed of smooth functions on phase space. This is also a consistency condition. A Poisson bracket is defined by derivatives, so nested Poisson brackets are well-defined, i.e.
\begin{equation}
\label{eq: nested Poisson}
  \{ \mathcal{O}_n, \{\dots , \{ \mathcal{O}_2, \{\mathcal{O}_1,\mathcal{O}_0 \underbrace{\} \} \dots \}}_n < \infty, \quad \forall n \in \mathbb{N},
\end{equation}
if and only if classical observables $\mathcal{O}_i$'s are smooth functions on the phase space.  

Now we turn to show that the positivity restriction~\eqref{eq: positivity m_e, m_f, m_K} nontrivially realizes the extension~\eqref{eq: emergence under restriction} for the classical monodromy algebra $C^\infty (\mathcal{P}_\bullet)$. Restricting all smooth functions on the real line $\mathbb{R}$ to the positive half real line $\mathbb{R}^+$ does not give rise to all smooth functions on $\mathbb{R}^+$, i.e.
\begin{equation}
    C^{\infty} (\mathbb{R})|_{\mathbb{R}^+} \subsetneq C^\infty (\mathbb{R}^+).
\end{equation}
This is because the function $f(x) =x^\nu$ is smooth everywhere on $\mathbb{R}$ but at $x=0$ for $\nu \notin \mathbb{N}$, i.e. $\exists n \in \mathbb{N}$ such that
\begin{equation}
     \lim_{x \to 0}\partial^n_x (x^\nu) = \infty.
\end{equation}
The quantum monodromy algebra $\mathcal{A}_\bullet$ is generated by $m_{\mathrm{e}}$, $m_{\mathrm{f}}$, and $ m_{\mathrm{K}}$. To have a well-defined classical limit $C^\infty (\mathcal{P}_\bullet)$, the quantum monodromy algebra $\mathcal{A}_\bullet$ without the positivity restriction~\eqref{eq: positivity m_e, m_f, m_K} is composed of linear combinations of elements of the form
\begin{equation}
\label{eq: m n_e m n_f m nu}
    m_{\mathrm{e}}^{n_{\mathrm{e}}} m_{\mathrm{f}}^{n_{\mathrm{f}}} m_{\mathrm{K}}^\nu
\end{equation}
with $n_{\mathrm{e}}, n_{\mathrm{f}}  \in \mathbb{N}$ and $ \nu \in \mathbb{C}$. With the positivity restriction~\eqref{eq: positivity m_e, m_f, m_K}, the exponents $n_{\mathrm{e}}$ and $n_{\mathrm{f}}$ in \eqref{eq: m n_e m n_f m nu} can be arbitrary numbers. In particular, the exponents $n_{\mathrm{e}}$, $n_{\mathrm{f}}$, and $\nu$ can scale as $1/\G$. However, Poisson brackets of such observables are divergent in the large $N$ limit ($\G \to 0$), which is in contradiction with the definition of of Poisson bracket as the classical limit ($\hbar \to 0$). For example, Eq.~\eqref{eq: K, e Poisson} implies
\begin{equation}
\label{eq: e.g. O(k) Poisson}
    \{ m_{\mathrm{K}}^{O(k)}, m_{\mathrm{e}}^{O(k)} \} = O(k)\, m_{\mathrm{e}} m_{\mathrm{K}}. 
\end{equation}
To see the contradiction with the classical limit, note that the Chern-Simons level $k$ and the Plank constant $\hbar$ should always appear together as the combination $\hbar/k$ since $k$ as the coupling constant is an overall factor in the Lagrangian~\eqref{eq: chern-simons lagrangian}. Using Eq.~\eqref{eq: k = 1/4G_N}, the commutator of two operators $\mathcal{O}_1$ and $\mathcal{O}_2$ hence can be written as
\begin{equation}
\label{eq: [ O_1, O_2 ] = f}
    [ \mathcal{O}_1, \mathcal{O}_2 ] = f(\G \hbar),
\end{equation}
with $f$ a function of $\G \hbar$. Recall that the Poisson bracket and the commutator are related by Eq.~\eqref{eq: classical limit}. Eqs.~\eqref{eq: [ O_1, O_2 ] = f} and \eqref{eq: classical limit} imply
\begin{equation}
\label{eq: O ,O Poisson, G f'}
    \{ \mathcal{O}_1,  \mathcal{O}_2\} = i \G f'(0).
\end{equation}
In the case of Eq.~\eqref{eq: e.g. O(k) Poisson}, the Poisson bracket diverges in the large $N$ limit, i.e.
\begin{equation}
    \lim_{\G \to 0} \G f'(0) = \infty.
\end{equation}
$f'(0)$ by definition is independent of $\G$, so Eq.~\eqref{eq: O ,O Poisson, G f'} implies that the Poisson bracket also diverges at finite $\G$, i.e.
\begin{equation}
\label{eq: divergent Poisson}
    \{ \mathcal{O}_1, \mathcal{O}_2 \} = \infty.
\end{equation}
This is in contradiction with the finite Poisson bracket~\eqref{eq: e.g. O(k) Poisson}. Combining Eqs.~\eqref{eq: classical limit} and \eqref{eq: divergent Poisson} imply that the commutator $[\mathcal{O}_1, \mathcal{O}_2 ]$ is not analytic at $\hbar = 0$. Imposing positivity restriction~\eqref{eq: positivity m_e, m_f, m_K} at the classical level hence induces a subalgebra that is not quantizable. Equivalently, the quantization of monodromy observables satisfying positivity restriction~\eqref{eq: positivity m_e, m_f, m_K} contains a subalgebra that does not have a classical limit or a large $N$ limit. This is a signature of erratic large-$N$ behaviors.

The Poisson brackets~\eqref{eq: classical Heisenberg} of the Kashaev coordinates $(x, p , \lambda)$ are quantized as
\begin{equation}
    [x,p]= \frac{\hbar}{ 2\pi i},\quad [\lambda, x] = 0, \quad [\lambda, p] = 0.
\end{equation}
In terms of the quantized Kashaev coordinates, the $U_q(\psl)$ generators can be realized as~\cite{Nidaiev:2013bda}
\begin{align}
    m_{\mathrm{K}} &= e^{\pi b p},
    \label{eq: mk quantum Kashaev}
    \\
    m_{\mathrm{e}} &= e^{\pi b (p-2x)},
    \label{eq: me quantum Kashaev}
    \\
    m_{\mathrm{f}} &= e^{\pi b (x-\frac{p}{2})} \left(2 \cosh (2 \pi b \lambda) + 2 \cosh(2 \pi b p)\right) e^{\pi b (x-\frac{p}{2} )},
    \label{eq: mf quantum Kashaev}
\end{align}
with
\begin{equation}
\label{eq: b def}
    b \equiv (k - 2 \hbar)^{-\frac{1}{2}} .
\end{equation}
The algebra $\mathcal{A}^+_\bullet$ generated by positive operators \eqref{eq: mk quantum Kashaev}, \eqref{eq: me quantum Kashaev}, and \eqref{eq: mf quantum Kashaev} is called the modular double $U_{q\tilde{q}}(\psl)$ of $U_q(\psl)$~\cite{Faddeev:1999fe, Ponsot:1999uf}. With the positivity restrictions~\eqref{eq: positivity m_e, m_f, m_K}, arbitrary powers of $m_{\mathrm{e}}$ and $m_{\mathrm{f}}$ are well-defined. The positive monodromy algebra $\mathcal{A}_\bullet^+$ hence contains the original monodromy algebra $\mathcal{A}_\bullet$ as a proper subalgebra, i.e.
\begin{equation}
\label{eq: identify positive monodromy algebra}
    U_q(\psl) \cong \mathcal{A}_\bullet \subsetneq \mathcal{A}^+_\bullet \cong U_{q\tilde{q}}(\psl).
\end{equation}
It was proven that~\cite{Bytsko:2002br}
\begin{align}
    m_{\widetilde{\mathrm{K}}} \,m_{ \tilde{\mathrm{e}} } &= \tilde{q}\, m_{ \tilde{\mathrm{e}} }\, m_{\widetilde{\mathrm{K}}},
    \label{eq: dual U_q 1}
    \\
    m_{\widetilde{\mathrm{K}}}\, m_{ \tilde{\mathrm{f}} } &= \tilde{q}^{-1} m_{ \tilde{\mathrm{f}} }\, m_{\widetilde{\mathrm{K}}},
    \label{eq: dual U_q 2}
    \\
    [m_{ \tilde{\mathrm{e}} }, m_{ \tilde{\mathrm{f}} }] &=(\tilde{q}-\tilde{q}^{-1}) (m_{\widetilde{\mathrm{K}}}^{-2} - m_{\widetilde{\mathrm{K}}}^2),
    \label{eq: dual U_q 3}
\end{align}
where we defined
\begin{equation}
\label{eq: dual m def}
        m_{ \tilde{\mathrm{e}} } \equiv (m_{\mathrm{e}})^{\frac{1}{b^2 \hbar}},
    \quad
    m_{ \tilde{\mathrm{f}} } \equiv (m_{\mathrm{f}})^{\frac{1}{b^2 \hbar}},
    \quad
    m_{\widetilde{\mathrm{K}}} \equiv (m_{\mathrm{K}})^{\frac{1}{b^2 \hbar}},
\end{equation}
and
\begin{equation}
\label{eq: qdual def}
    \tilde{q} \equiv \exp\left(-  \frac{\pi i k}{\hbar}\right).
\end{equation}
Eqs.~\eqref{eq: dual U_q 1}, \eqref{eq: dual U_q 2}, and \eqref{eq: dual U_q 3} also define a quantum group algebra $U_{\tilde{q}}(\psl)$ but with a different deformation parameter $\tilde{q}$. $U_{\tilde{q}}(\psl)$ is called the modular dual of $U_q (\psl)$. It was also proven that the subalgebra $U_{\tilde{q}}(\psl)$ commutes with the subalgebra $U_q(\psl)$~\cite{Bytsko:2002br}. For generic values of $k$, $U_q(\psl)$ and $U_{\tilde{q}}(\psl)$ are sufficient to generate (a dense subset of) all powers of the monodromy variables, so the positive monodromy algebra is the tensor product between the original monodromy algebra $\mathcal{A}_\bullet \cong U_q (\psl)$ and the modular dual algebra $U_{\tilde{q}}(\psl)$, i.e.
\begin{equation}
\label{eq: positive algebra = monodromy algebra otimes modular dual}
    \mathcal{A}_\bullet^+ \cong \mathcal{A}_\bullet \otimes U_{\tilde{q}}(\psl).
\end{equation}
Eqs.~\eqref{eq: b def} and \eqref{eq: dual m def} imply that the modular dual operators $m_{\tilde{\mathrm{e}}}$, $m_{\tilde{\mathrm{f}}}$, and $m_{\widetilde{\mathrm{K}}}$ are singular in the large $N$ limit ($k \to \infty$). The modular dual algebra $U_{\tilde{q}}(\psl)$ does not have a classical limit ($\hbar \to 0$) either, because the dual quantization parameter~\eqref{eq: qdual def} is not analytic at $\hbar = 0$. See Figure~\ref{fig: erratic q} for an illustration.
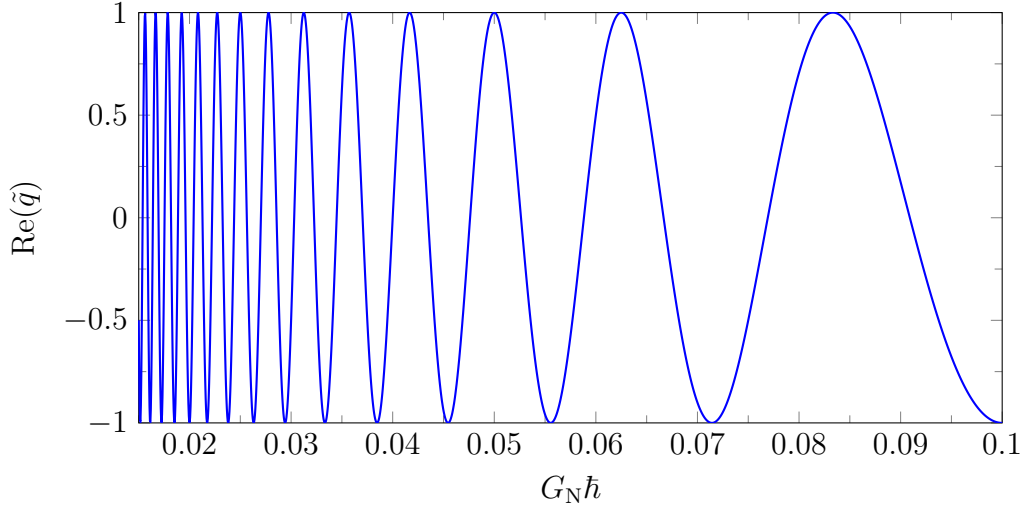
\begin{figure}[htbp]
\centering

\begin{tikzpicture}
\begin{axis}[
    width=13cm,
    height=7cm,
    domain=0.015:0.1,
    samples=2000,
    xlabel={$G_{\mathrm{N}}\hbar$},
    ylabel={$\mathrm{Re}(\tilde{q})$},
    axis lines=box,
    minor tick num=1,
    scaled x ticks=false,
    xticklabel style={/pgf/number format/fixed},
    enlargelimits=false,
]

\addplot[
    blue,
    thick
]
{cos(deg(pi/(2*x)))};

\end{axis}
\end{tikzpicture}

\caption{Erratic behavior of the real part of $\tilde{q}$ in the large $N$ limit ($\G \to 0$).}
\label{fig: erratic q}

\end{figure}
Thus, the identification~\eqref{eq: Asmooth= quantum monodromy algebra} may be extended as
\begin{equation}
\label{eq: Acomplex cong Apositive monodromy}
    \mathcal{A}_\cft^{\tx{complex}} \cong   \mathcal{A}^+_{\bullet},
\end{equation}
where $\mathcal{A}_\cft^{\tx{complex}}$ is the boundary observable algebra dual to the black hole interior introduced in Section~\ref{sec: Introduction}.

\section{Filtering CFTs via Gauging Nonlocal Symmetries}
\label{ref: Filtering CFTs via Gauging Nonlocal Symmetries}

\subsection{Reformulation of Lie Group Symmetries}
\label{sec: Reformulation of Lie Group Symmetries}

In this subsection, we review some useful facts about Lie group symmetries from a phase space perspective. See e.g. \cite{Arnold:1989who, guillemin1990symplectic} for a review. We reformulate the notion of symmetries and charges so as to facilitate the discussion of nonlocal symmetries of black holes in Section~\ref{sec: Nonlocal Symmetries of One-Sided Black Holes}.

Given a theory with a phase space $\mathcal{P}$, suppose observables of the theory transform under an action of a Lie group $G$, i.e.
\begin{align}
    \label{eq: general Lie group symmetry}
    C^{\infty} (\mathcal{P}) \times G &\to C^{\infty} (\mathcal{P}),
    \\
    (\mathcal{O}, g) &\to g \cdot \mathcal{O}.
\end{align}
Usually, the Lie group~$G$ is a symmetry group of the theory if the Lagrangian of the theory is invariant under the $G$-action~\eqref{eq: general Lie group symmetry} up to a total derivative. Given the symmetry transformation~\eqref{eq: general Lie group symmetry}, every element~$\sigma$ in the Lie algebra $\mg$ of the symmetry group $G$ corresponds to an infinitesimal symmetry transformation
\begin{align}
\label{eq: infinitesimal symmetry transformation}
    \mathcal{O} \to \mathcal{O} + \delta_\sigma \mathcal{O}
\end{align}
with $\delta_\sigma \mathcal{O}$ the infinitesimal variation of $\mathcal{O} \in C^\infty (\mathcal{P})$. According to the Noether theorem, every infinitesimal symmetry transformation \eqref{eq: infinitesimal symmetry transformation} corresponds to a local current $j[\sigma]$. The Noether charge $Q[\sigma]$ associated with the symmetry is given by the integral of the local current over a Cauchy slice~$\Sigma$,
\begin{equation}
\label{eq: Q = int j}
    Q[\sigma] = \int_\Sigma j[\sigma].
\end{equation}
The Noether charge $Q[\sigma]$ and the infinitesimal symmetry transformation \eqref{eq: infinitesimal symmetry transformation} are related by the Poisson bracket
\begin{equation}
\label{eq: delta O = Poisson (Q, O)}
    \delta_\sigma \mathcal{O} = \{ Q[\sigma], \mathcal{O} \}, \quad \forall \mathcal{O} \in C^\infty (\mathcal{P}).
\end{equation}

Now we turn to reformulate the above from a phase space perspective. The infinitesimal symmetry transformation \eqref{eq: infinitesimal symmetry transformation} induces a vector field $\mathfrak{X}[\sigma]$ on the phase space~$\mathcal{P}$, i.e. $\forall \mathcal{O} \in C^\infty (\mathcal{P})$,
\begin{equation}
\label{eq: Lie derivative def}
    \mathscr{L}_{\mathfrak{X}[\sigma]} \mathcal{O} \equiv\delta_\sigma \mathcal{O}.
\end{equation}
$\mathscr{L}_{\mathfrak{X}[\sigma]}$ denotes the Lie derivative along $\mathfrak{X}[\sigma]$. Eq.~\eqref{eq: delta O = Poisson (Q, O)} is equivalent to a relation among the Noether charge $Q[\sigma]$, the vector field $\mathfrak{X}[\sigma]$, and the symplectic form $\Omega$ of the phase space $\mathcal{P}$:
\begin{equation}
\label{eq: imath X Omega = delta Q}
    \imath_{\mathfrak{X}[\sigma]} \Omega = \delta Q [\sigma]
\end{equation}
where $\imath$ denotes the interior product. In terms of tensor indices, the interior product is an index contraction
\begin{equation}
    (\imath_{\mathfrak{X}[\sigma]} \Omega )_\alpha  \equiv \mathfrak{X}[\sigma]^\beta \,\Omega_{\beta\alpha}.
\end{equation}
With this notation, the infinitesimal variation in \eqref{eq: infinitesimal symmetry transformation} can be written as
\begin{equation}
    \delta_{\sigma} = \imath_{\mathfrak{X}[\sigma]} \,\delta.
\end{equation}
According to Cartan's magic formula
\begin{equation}
\label{eq: Cartan's magic formula}
    \mathscr{L}_{\mathfrak{X}[\sigma]} = \imath_{\mathfrak{X}[\sigma]} \,\delta + \delta \,\imath_{\mathfrak{X}[\sigma]},
\end{equation}
Eq.~\eqref{eq: imath X Omega = delta Q} implies
\begin{equation}
\label{eq: symplectomorphism}
    \mathscr{L}_{\mathfrak{X}[\sigma]} \Omega = 0.
\end{equation}
So the existence of local currents ensures that the associated symmetry transformation~\eqref{eq: general Lie group symmetry} induces a diffeomorphism preserving the symplectic form of the phase space. Such a diffeomorphism is called a symplectomorphism or a canonical transformation.

We now further reformulate the above so that we do not need assume the existence of the local current to define the charge associated with the symmetry. We consider the exponential of the Noether charge $Q[\sigma]$ as an observable
\begin{equation}
\label{eq: global charge def}
    \mathcal{Q}[\sigma] = e^{ Q[\sigma]}, \quad \sigma \in \mg.
\end{equation}
We call $\mathcal{Q}$ the global charge. The Noether charge $Q[\sigma]$ is a number, so $Q$ can be regarded as a linear function on the Lie algebra $\mg$, i.e.
\begin{align}
    Q: \mg &\to \mathbb{R}
    \\
    \sigma &\to Q[\sigma].
\end{align}
Then the global charge $\mathcal{Q}$ is an exponential function on $\mg$. In terms of the global charge $\mathcal{Q}$, Eq.~\eqref{eq: delta O = Poisson (Q, O)} and Eq.~\eqref{eq: imath X Omega = delta Q} can be rewritten as~\footnote{As a side remark, the way to bootstrap the Wilson line bracket~\eqref{eq: universal WW Poisson} is not simply a trick. The Poisson bracket associated with the asymptotic symmetry of the boundary gravitons is Eq.~\eqref{eq: A phi, W Poisson} in which $\oint_{\bdy}\tr (A \phi) $ is interpreted as the charge $Q$ and the Wilson line $W|_{x,y}$ is interpreted as the observable $\mathcal{O}$ that the asymptotic symmetry acts on. However, one can read the Poisson bracket~\eqref{eq: A phi, W Poisson} in an opposite way such that $W|_{x,y}$ is interpreted as the global charge $\mathcal{Q}$ and $\oint_{\bdy}\tr (A \phi) $ is interpreted as the observable $\mathcal{O}$ that the nonlocal symmetry associated with the global charge $W|_{x,y}$ acts on. This opposite perspective becomes clear if one compares Eq.~\eqref{eq: global charge Poisson} with Eq.~\eqref{eq: ODE}.}
\begin{equation}
\label{eq: global charge Poisson}
    \delta_\sigma \mathcal{O} = ( \mathcal{Q}^{-1}\{ \mathcal{Q}, \mathcal{O} \} ) [\sigma], \quad \forall \mathcal{O} \in C^\infty (\mathcal{P}),
\end{equation}
and
\begin{equation}
    \label{eq: imath X Omega = U^-1 delta U}
    \imath_{\mathfrak{X}[\sigma]} \Omega = (-\mathcal{Q}^{-1}\delta \mathcal{Q} ) [\sigma].
\end{equation}
Remarkably, the equivalence between Eq.~\eqref{eq: global charge Poisson} and Eq.~\eqref{eq: imath X Omega = U^-1 delta U} holds regardless of the existence of local currents. See Appendix~\ref{app: Definitions of Symmetries} for a proof. We can hence conversely define symmetries and the associated charges by Eq.~\eqref{eq: global charge Poisson} or \eqref{eq: imath X Omega = U^-1 delta U}. More precisely, we call the Lie group action~\eqref{eq: general Lie group symmetry} a symmetry if there exists an observable $\mathcal{Q}$ such that Eq.~\eqref{eq: global charge Poisson} or \eqref{eq: imath X Omega = U^-1 delta U} holds for all $\sigma \in \mg$ with $\delta_\sigma$ and $\mathfrak{X}[\sigma]$ induced by the symmetry transformation. $\mathcal{Q}$ is then called the global charge associated with the symmetry.

Note that the notion of symmetry defined by Eq.~\eqref{eq: global charge Poisson} or~\eqref{eq: imath X Omega = U^-1 delta U} generalizes the one based on the Lagrangian invariance. If a global charge cannot be constructed by a local current via Eqs.~\eqref{eq: Q = int j} and~\eqref{eq: global charge def}, then we call the associated symmetry a nonlocal symmetry. Recall that the existence of a local current makes a symmetry transformation \eqref{eq: general Lie group symmetry} a symplectomorphism according to Eq.~\eqref{eq: symplectomorphism}. In general, Eq.~\eqref{eq: symplectomorphism} does not hold for nonlocal symmetries. This provides a useful criterion for determining the nonlocality of a symmetry: if a Lie group symmetry does not generate a canonical transformation, then it is nonlocal. It is usually believed that Eq.~\eqref{eq: symplectomorphism} must hold for a symmetry. Since the unitary operator corresponding to the symmetry at the quantum level preserves commutator, the unitary operator must reduce to a canonical transformation preserving the Poisson bracket. However, if the parameter describing the symmetry transformation rule is also quantized to be an operator, the symmetry group will inherit a Poisson structure in the classical limit such that Eq.~\eqref{eq: symplectomorphism} does not hold.

\subsection{Nonlocal Symmetries of One-Sided Black Holes}
\label{sec: Nonlocal Symmetries of One-Sided Black Holes}

In this subsection, we illustrate that the puncture large gauge transformation \eqref{eq: puncture gauge transformation} is a nonlocal symmetry. According to Section~\ref{sec: Reformulation of Lie Group Symmetries}, we need to show that there exists a corresponding global charge which cannot be constructed from a local current.

By definition~\eqref{eq: m def}, the puncture large gauge transformation \eqref{eq: puncture gauge transformation} acts on the puncture monodromy via adjoint action
\begin{align}
\label{eq: adjoint action on m}
    \mathcal{P}_\bullet \times \tx{PSL}(2, \mathbb{R}) &\to \mathcal{P}_\bullet
    \\
    (m, h) &\to h^{-1} m h.
\end{align}
To derive the global charge by applying Eq.~\eqref{eq: imath X Omega = U^-1 delta U} to the symplectic form~\eqref{eq: Omega_bullet}, we need to consider the infinitesimal adjoint action on $m_\pm$. The finite form of the adjoint action \eqref{eq: adjoint action on m} on $m_\pm$ is 
\begin{equation}
\label{eq: adjoint on m_pm}
    m_\pm \to (h^{-1} m h)_\pm.
\end{equation}
$(h^{-1} m h)_\pm$ is defined by applying the STS decomposition \eqref{eq: Borel decomposition} to $h^{-1} m h$. The puncture large gauge transformation \eqref{eq: adjoint action on m} induces an infinitesimal transformation $\forall \sigma\in \mathfrak{sl}(2, \mathbb{R})$,
\begin{equation}
    m \to m + \delta_\sigma m
\end{equation}
with
\begin{equation}
\label{eq: infinitesimal adjoint action on m}
    \delta_{\sigma} m = \frac{2\pi}{k} [ m, \sigma ].
\end{equation}
The corresponding infinitesimal form of the gauge transformation~\eqref{eq: adjoint on m_pm} is given by
\begin{align}
    \delta_\sigma m_- &= \frac{2 \pi}{k} r_- ( m_+ \sigma m_+^{-1} - m_- \sigma m_-^{-1} ) m_-,
    \label{eq: infinitesimal adjoint action on m_-}
    \\
    \delta_\sigma m_+ &= \frac{2\pi}{k} r_+ ( m_+ \sigma m_+^{-1} - m_- \sigma m_-^{-1}  ) m_+.
    \label{eq: infinitesimal adjoint action on m_+}
\end{align}
$r_\pm (\bullet)$ denotes the contraction defined by
\begin{equation}
\label{eq: r_pm (sigma) def}
    r_\pm (\sigma) \equiv (r_\pm)^{ab} \sigma_b t_a, \quad \forall \sigma \in \mathfrak{sl}(2, \mathbb{R}).
\end{equation}
$r_\pm \in \mathfrak{sl}(2, \mathbb{R}) \otimes \mathfrak{sl}(2, \mathbb{R})$ is defined by Eq.~\eqref{eq: r_pm def}. $\{t_a\}$ is a basis of $\mathfrak{sl}(2, \mathbb{R})$. $\sigma_b \in \mathbb{R}$ is the coefficient in the expansion
\begin{equation}
    \sigma = \sigma_a t^a.
\end{equation}
$\{t^a\}$ is a basis of $\mathfrak{sl}(2, \mathbb{R})$ dual to $\{t_a\}$ in the sense of Eq.~\eqref{eq: dual basis}. Contracted indices are summed over the basis. The consistency of Eqs.~\eqref{eq: Borel decomposition}, \eqref{eq: infinitesimal adjoint action on m}, \eqref{eq: infinitesimal adjoint action on m_-}, and \eqref{eq: infinitesimal adjoint action on m_+} can be verified by a straightforward calculation using Eqs.~\eqref{eq: sl2 r-matrix}, \eqref{eq: sl2 generators}, \eqref{eq: sl2 K}, \eqref{eq: m-}, and \eqref{eq: m+}. See Appendix~\ref{app: Infinitesimal Adjoint Action and the STS Decomposition} for a simpler but less direct derivation.

According to Eq.~\eqref{eq: Lie derivative def}, the puncture large gauge transformation \eqref{eq: puncture gauge transformation} induces a vector field $\mathfrak{X}[\sigma]$ for every $\sigma \in \mathfrak{sl}(2, \mathbb{R})$. It acts on $W_0$ (i.e. the Wilson line connecting the puncture to the base point $x^+ = 0$ on the boundary circle $\bigcirc$) via
\begin{equation}
\label{eq: puncture gauge transformation on W_0}
    \mathscr{L}_{\mathfrak{X}[\sigma]} W_0 = \frac{2\pi}{k}W_0 \sigma.
\end{equation}
Using Eqs.~\eqref{eq: Omega_bullet}, \eqref{eq: Lie derivative def}, \eqref{eq: infinitesimal adjoint action on m_-}, \eqref{eq: infinitesimal adjoint action on m_+}, and \eqref{eq: puncture gauge transformation on W_0}, we have
\begin{align}
\label{eq: imath X[sigma] Omega}
    \imath_{\mathfrak{X}[\sigma]} \Omega_\bullet 
    = 
    -\tr ( m_-^{-1} \delta m_- \sigma ) +\tr ( m_+^{-1} \delta m_+ \sigma ) .
\end{align}
It is not straightforward to verify Eq.~\eqref{eq: imath X[sigma] Omega}. A detailed derivation is provided in Appendix~\ref{app: Global Charges of Nonlocal Symmetries}. 

The puncture large gauge transformation~\eqref{eq: puncture gauge transformation} is a symmetry since Eq.~\eqref{eq: imath X[sigma] Omega} is in fact of the form Eq.~\eqref{eq: imath X Omega = U^-1 delta U}. To make it explicit, note that
\begin{equation}
\label{eq: Q-1 delta Q = (,)}
    \mathcal{Q}^{-1} \delta \mathcal{Q} = ( m_-^{-1} \delta m_-, m_+^{-1} \delta m_+ ) 
\end{equation}
if we define
\begin{equation}
\label{eq: Q def PSL^*}
    \mathcal{Q} \equiv ( m_-, m_+ )
\end{equation}
as a group element in $\tx{PSL}(2, \mathbb{R}) \times \PSL$. According to Eqs.~\eqref{eq: m-} and \eqref{eq: m+}, $\mathcal{Q}$ is valued in a subgroup of $\PSL \times \PSL$ defined by
\begin{equation}
\label{eq: dual PSL def}
    \PSL^* \equiv \{ ( -\left(
\begin{array}{cc}
 m_{\mathrm{K}}^{-1} & 0 \\
 m_{\mathrm{e}} & m_{\mathrm{K}} \\
\end{array}
\right), \left(
\begin{array}{cc}
 m_{\mathrm{K}}  & m_{\mathrm{f}} \\
 0 & m_{\mathrm{K}}^{-1} \\
\end{array}
\right) ) \in \PSL \times \PSL \}.
\end{equation}
The Lie algebra $\psl^*$ of $\PSL^*$ is isomorphic to $\psl$ as a linear space, so $\psl^*$ can be identified as the dual (i.e. the space of linear functions) of $\psl$ by defining 
\begin{align}
    \psl^* \times \psl &\to \mathbb{R}
    \\
    (\chi, \sigma) &\to \chi[ \sigma ] \equiv -\tr ( (\chi_+ - \chi_-) \sigma ),
    \label{eq: chi[sigma]}
\end{align}
$\forall \chi= (\chi_-, \chi_+) \in \psl^*$ and $\sigma \in \psl$. Using Eqs.~\eqref{eq: Q-1 delta Q = (,)} and \eqref{eq: chi[sigma]}, one can rewrite Eq.~\eqref{eq: imath X[sigma] Omega} in the form of Eq.~\eqref{eq: imath X Omega = U^-1 delta U}. Thus, the puncture large gauge transformation \eqref{eq: puncture gauge transformation} is a symmetry with the global charge $\mathcal{Q} \in \PSL^*$ given by Eq.~\eqref{eq: Q def PSL^*}. Using Eqs.~\eqref{eq: Cartan's magic formula}, \eqref{eq: Q def PSL^*}, and \eqref{eq: imath X[sigma] Omega}, we have
\begin{equation}
\label{eq: Lie_sigma Omega_bullet = }
    \mathscr{L}_{\mathfrak{X}[\sigma]} \Omega_\bullet = ( \mathcal{Q}^{-1} \delta \mathcal{Q} \mathcal{Q}^{-1} \delta \mathcal{Q} ) [\sigma].
\end{equation}
$\tx{PSL}(2, \mathbb{R})^*$ is nonabelian, so the right hand side of Eq.~\eqref{eq: Lie_sigma Omega_bullet = } is not zero and the puncture large gauge transformation~\eqref{eq: puncture gauge transformation} hence does not preserve the symplectic form $\Omega_\bullet$. As explained in Section~\ref{sec: Reformulation of Lie Group Symmetries}, the associated global charge~\eqref{eq: Q def PSL^*} cannot be constructed from a local current. Hence the puncture large gauge transformation~\eqref{eq: puncture gauge transformation} is a nonlocal symmetry.

In deriving the global charge~\eqref{eq: Q def PSL^*}, we consider the puncture large gauge transformation~\eqref{eq: puncture gauge transformation} acting on the interior phase space equipped with the symplectic form~\eqref{eq: Omega_bullet}. If we instead consider the full system with the symplectic form~\eqref{eq: Omega odot}, we will derive the same global charge. Recall that the full symplectic form~\eqref{eq: Omega odot} reduces to the interior symplectic form~\eqref{eq: Omega_bullet} by fixing the boundary gauge potential. Since the transformation~\eqref{eq: puncture gauge transformation} acts trivially on the boundary graviton phase space $\mathcal{P}_{\bdy}$ according to Eq.~\eqref{eq: W to A}, fixing the boundary gauge potential or not does not change the calculation result. 

As a remark, Eqs.~\eqref{eq: Banados potential}, \eqref{eq: locality}, \eqref{eq: Q def PSL^*} imply
\begin{equation}
    \{ \mathcal{Q}, L_0 \} = 0.
\end{equation}
where $L_0$ is (the chiral part of) the Hamiltonian of boundary gravitons given by Eq.~\eqref{eq: bdy graviton Hamiltonian}. So the puncture large gauge transformation~\eqref{eq: puncture gauge transformation} is also a Hamiltonian symmetry in the usual sense.

\subsection{Gauging Nonlocal Symmetries}
\label{sec: Gauging Nonlocal Symmetries}

For an exterior observer, the information of observables of high complexity is lost in the large $N$ limit such that the classical algebra of simple observables has a nontrivial commutant corresponding to observables of the one-sided black hole. In this subsection, we show that the classical boundary graviton algebra $C^\infty (\mathcal{P}_{\bdy})$ and the classical monodromy algebra $C^\infty (\mathcal{P}_\bullet)$ are commutant of each other, which justifies the interpretation~\eqref{eq: interior algebra = monodromy algebra}. Boundary gravitons are described by the gauge potential $A$ and $\overline{A}$. We focus on the chiral part $A$ for simplicity. The Wilson line $W$ is reduced to the boundary gravitons via boundary reduction map $\mathcal{F}_{\bdy}$ \eqref{eq: bdy reduction map}. The Wilson line Poisson bracket \eqref{eq: simplest WW Poisson} induces a Poisson bracket on the boundary graviton phase space $\mathcal{P}_{\scalebox{0.5}{$\bigcirc$}}$ via $\mathcal{F}_{\bdy}$. As explained in Section \ref{sec: Intrinsic Incompleteness of One-Sided Boundary Gravitons}, the induced Poisson bracket on $\mathcal{P}_{\bdy}$ is degenerate. So $\mathcal{P}_{\bdy}$ itself cannot be the physical phase space (i.e. a symplectic manifold) of boundary gravitons but decomposes into physical phase subspaces (i.e. symplectic leaves) which are orbits of the large gauge transformations. If the asymptotically AdS$_3$ boundary condition is imposed, then these phase subspaces will be Virasoro coadjoint orbits. In the following, we reproduce the same result from the perspective of gauging the nonlocal symmetry~\eqref{eq: puncture gauge transformation}, which can be interpreted as filtering out the observables of high complexity according to interpretation~\eqref{eq: Asmooth= quantum monodromy algebra}.

The nonlocal symmetry~\eqref{eq: puncture gauge transformation} is a gauge symmetry in the boundary graviton phase space $\mathcal{P}_{\bdy}$ since it acts on the boundary gauge potential~\eqref{eq: A = dWW-1} trivially. Conversely, an observable $\mathcal{O} \in C^\infty (\mathcal{P}_\odot)$ is invariant under the nonlocal symmetry~\eqref{eq: puncture gauge transformation} only if $\mathcal{O}$ is not composed of Wilson lines ending at the puncture. $\mathcal{O}$ as a physical observable must be a functional of Wilson lines. Then $\mathcal{O}$ can only be a functional of the Wilson line on the circle boundary $\bigcirc$, and thus a functional of $A$. So the boundary reduction map \eqref{eq: bdy reduction map} is a map gauging the nonlocal symmetry~\eqref{eq: puncture gauge transformation} with the boundary graviton observable algebra $C^\infty (\mathcal{P}_{\bdy})$ the gauge-invariant subalgebra of $C^\infty (\mathcal{P}_\odot)$.

A symmetry is gauged if its action on physical observables is trivial and the corresponding charge is constant on the phase space. As shown in Section \ref{sec: Nonlocal Symmetries of One-Sided Black Holes}, the charge corresponding to the nonlocal symmetry~\eqref{eq: puncture gauge transformation} is the puncture monodromy. So each physical phase subspace of boundary gravitons can be identified as a space $\mathcal{L}_{\bdy} (m)$ of the gauge potential configurations corresponding to constant puncture monodromy $m$ \eqref{eq: m def}. More precisely,
\begin{equation}
\label{eq: L (m) leaf}
\mathcal{L}_{\bdy} (m)= \mathcal{F}_{\bdy} \left( \mathcal{F}_{\bullet}^{-1} (m)\right)
\end{equation}
with
$\mathcal{F}_{\bdy}$ the boundary reduction map \eqref{eq: bdy reduction map} and $\mathcal{F}_{\bullet}^{-1} (m)$ the preimage of $m$ under the puncture reduction map~\eqref{eq: puncture reduction map}.

Now we describe $\mathcal{L}_{\scalebox{0.5}{$\bigcirc$}}$ in detail. The preimage $\mathcal{F}_{\bullet}^{-1} (m) \subset \mathcal{P}_\odot$ of $m$ is composed of all Wilson lines corresponding to fixed puncture monodromy $m$, i.e.
\begin{equation}
    \mathcal{F}_{\bullet}^{-1} (m) = \{ W \in \mathcal{P}_\odot;\, W^{-1}|_{x} W|_{2 \pi} = m \}.
\end{equation}
In the full phase space $\mathcal{P}_\odot$, any pair of Wilson lines $W$ and $W'$ with the same puncture monodromy are related by a boundary large gauge transformation~\eqref{eq: boundary large gauge transformation} given by
\begin{equation}
    g = W' W^{-1}.
\end{equation}
So every $W \in \mathcal{F}_{\bullet}^{-1} (m) $ can be decomposed as
\begin{equation}
    W = g \check{W}
\end{equation}
where $g$ is a boundary large gauge transformation and $\check{W} \in \mathcal{P}_\odot$ is an arbitrary given representative whose puncture monodromy is $m$. So $\mathcal{F}_{\bullet}^{-1} (m)$ is an orbit of the boundary large gauge transformation acting on $\check{W}$. Using Eq.~\eqref{eq: W to A}, we have
\begin{equation}
    \mathcal{F}_{\bdy} ( g \check{W} ) = dg g^{-1} +  g  \check{A} g^{-1}
\end{equation}
with
\begin{equation}
\label{eq: below}
    \check{A} = d \check{W} \check{W}^{-1}.
\end{equation}
So $\mathcal{L}_{\scalebox{0.5}{$\bigcirc$}} (m) = \mathcal{F}_{\bdy} \left( \mathcal{F}_{\bullet}^{-1} (m) \right)$ is an orbit of the boundary large gauge transformations~\eqref{eq: boundary large gauge transformation} acting on the representative gauge potential $\check{A}$.

Without the asymptotically AdS$_3$ boundary condition \eqref{eq: AdS boundary condition 1}, $\mathcal{L}_{\scalebox{0.5}{$\bigcirc$}} (m) $ would be a $\PSL$ Kac-Moody coadjoint orbit. After imposing Eq.~\eqref{eq: AdS boundary condition 1}, the large gauge transformation \eqref{eq: gauge transformation} reduces to the Virasoro coadjoint action on the stress tensor $T$~\cite{Polyakov:1989dm, Banados:1998gg}

\begin{equation}
\label{eq: delta T}
    \delta_\phi T = vT'  + 2 v' T  - \frac{ v'''}{8 G_{\tx{N}}} 
\end{equation}
where $v$ is a periodic function on the boundary and is related to $\phi$ in Eq.~\eqref{eq: gauge transformation} by 
\begin{equation}
    \phi = \begin{pmatrix}
        \frac{1}{2} v'   &  -\frac{4 G_{\tx{N}} T}{r} v + \frac{1}{2r} v''
        \\
        -r v & -\frac{1}{2} v'
    \end{pmatrix}.
\end{equation}
$r$ is the radial coordinate of the holographic boundary. Eq.~\eqref{eq: delta T} is the infinitesimal version of the Virasoro coadjoint action~\eqref{eq: global Virasoro coadjoint action}. With the asymptotically AdS$_3$ boundary condition \eqref{eq: AdS boundary condition 1}, the physical phase subspaces of boundary gravitons are hence Virasoro coadjoint orbits. We thus have reproduced the intrinsic incompleteness of the boundary gravitons discussed in Section~\ref{sec: Intrinsic Incompleteness of One-Sided Boundary Gravitons}.

We have shown that gauging the nonlocal symmetry~\eqref{eq: puncture gauge transformation} gives rise to the boundary graviton observable algebra $C^\infty (\mathcal{P}_{\bdy})$ as the gauge-invariant subalgebra. Using Eqs.~\eqref{eq: global charge Poisson} and \eqref{eq: Q-1 delta Q = (,)}, it is equivalent to say that $C^\infty (\mathcal{P}_{\bdy})$ is the commutant of the monodromy observable algebra $C^\infty (\mathcal{P}_\bullet)$ since the puncture monodromy is the charge corresponding to the nonlocal symmetry~\eqref{eq: puncture gauge transformation}. To be clear, $C^\infty (\mathcal{P}_{\bdy})$ is the algebra of elements in the full observable algebra $C^\infty (\mathcal{P}_\odot)$ that are Poisson-commuting with all elements in $C^\infty (\mathcal{P}_\bullet)$. In Section~\ref{sec: Interior Phase Spaces from Gauging Asymptotic Symmetries}, we showed that gauging the asymptotic symmetry~\eqref{eq: ASG action} gives rise to the monodromy observable algebra $C^\infty (\mathcal{P}_\bullet)$ as the gauge-invariant subalgebra. Using Eqs.~\eqref{eq: KAc-Moody charge} and \eqref{eq: general Noether charge}, it is equivalent to say that $C^\infty (\mathcal{P}_\bullet)$ is the commutant of $C^\infty (\mathcal{P}_{\bdy})$ since the boundary gauge potential is the charge corresponding to the asymptotic symmetry~\eqref{eq: ASG action}. Thus, $C^\infty (\mathcal{P}_\bullet)$ and $C^\infty (\mathcal{P}_{\bdy})$ are commutant of each other in $C^\infty (\mathcal{P}_\odot)$. This justifies the interpretation~\eqref{eq: interior algebra = monodromy algebra}. As a corollary, $C^\infty (\mathcal{P}_\bullet)$ and $C^\infty (\mathcal{P}_{\bdy})$ share the same center since the center is the overlap between two mutual commutants. Wilson lines in $C^\infty (\mathcal{P}_\bullet)$ (resp. $C^\infty (\mathcal{P}_{\bdy})$) are not connected to the asymptotic boundary (resp. the puncture). The center as the overlap is hence generated by the Wilson loops. Since $\psl$ is a Lie algebra of rank 1, Wilson loops of different winding numbers are functionally dependent. The center is hence generated by the Wilson loop of winding number 1.

The presence of the nontrivial center results in  superselection sectors for the boundary gravitons. Each superselection sector corresponds to a physical phase subspace $\mathcal{L}_{\bdy}(m)$. The information of the BTZ black hole geometry is encoded in the Wilson loops~\eqref{eq: BTZ wilson loop} up to boundary graviton fluctuations. Denote by $\tr_q\, M$ and $| \lambda \rangle$ the quantization of the Wilson loop and its eigenstates respectively. According to Eq.~\eqref{eq: BTZ wilson loop}, the classical Wilson loops are real numbers for black holes, so their quantization should be hermitian operators which implies
\begin{equation}
    \langle \lambda_1 | \lambda_2 \rangle = 0, \quad \lambda_1 \ne \lambda_2.
\end{equation}
As a central element, the Wilson loop commutes with every boundary graviton observable $\mathcal{O}$. Hence $\mathcal{O}$ does not change the eigenvalue of the quantized Wilson loop, i.e.
\begin{equation}
\label{eq: <lambda1 | O | lambda2> = 0}
    \langle \lambda_1 | \mathcal{O} | \lambda_2 \rangle = 0, \quad \forall \mathcal{O} \in \mathcal{A}_{\bdy},
\end{equation}
where $\mathcal{A}_{\bdy}$ denotes the quantization of the boundary graviton observable algebra $C^\infty (\mathcal{P}_{\bdy})$. Thus, based on boundary graviton observables, one cannot distinguish a superposition $|\psi \rangle =  | \lambda_1 \rangle + | \lambda_2 \rangle$ of pure states corresponding to different geometries from a mixed state $\rho = | \lambda_1 \rangle \langle \lambda_1 | + | \lambda_2 \rangle \langle \lambda_2 | $ corresponding to the same geometries, i.e.
\begin{equation}
    \langle \psi |\mathcal{O}  | \psi \rangle = \tr ( \rho \,\mathcal{O} ) , \quad \forall \mathcal{O} \in \mathcal{A}_{\bdy}.
\end{equation}
Since the Wilson loop commutes with both of the puncture monodromy and the boundary gravitons, Eq.~\eqref{eq: <lambda1 | O | lambda2> = 0} is only broken for Wilson lines $W$ connecting the puncture to the boundary, i.e.
\begin{equation}
\label{eq: <lambda1 | W | lambda2> ne 0}
    \langle \lambda_1 | W | \lambda_2 \rangle \ne 0.
\end{equation}
To detect the quantum superposition of different geometries, one hence must be able to access the puncture monodromy~\eqref{eq: m def}, which is expected to be of exponentially high complexity for an exterior observer. Thus, we expect that implementing Wilson line operators changing geometries is at least of exponentially high complexity. Eq.~\eqref{eq: <lambda1 | W | lambda2> ne 0} in fact provides a (2+1)-d gravity realization of the quantum necromancy theorem~\cite{Aaronson:2020ncs} which states that, for a qubit model, the complexity of distinguishing the superposition state $|\psi \rangle$ from the mixed state $\rho$ is essentially the same as the complexity of changing $|\lambda_1 \rangle$ into $|\lambda_2 \rangle$.

As a side remark, the boundary graviton phase space $\mathcal{P}_{\bdy}$ and the monodromy phase space~$\mathcal{P}_\bullet$ form a structure called a symplectic dual pair~\cite{kazhdan1978hamiltonian, guillemin1980moment, weinstein1983local}. More precisely, given a symplectic manifold~$\mathcal{P}_\odot$ and a pair $(\mathcal{F}_{\bdy}, \mathcal{F}_\bullet)$ of maps preserving Poisson brackets, if the the spaces of smooth functions on images $\mathcal{P}_{\bdy} = \mathcal{F}_{\bdy} (\mathcal{P}_\odot)$ and $\mathcal{P}_{\bullet} = \mathcal{F}_{\bullet} (\mathcal{P}_\odot)$ are commutant of each other, then $(\mathcal{P}_{\bdy}, \mathcal{P}_\bullet)$ is called a symplectic dual pair. There is a correspondence among symplectic leaves for a symplectic dual pair \cite{weinstein1983local}: every element $m \in \mathcal{P}_\bullet$ (resp. $A \in \mathcal{P}_{\scalebox{0.5}{$\bigcirc$}}$) corresponds to a symplectic leaf $\mathcal{L}_{\scalebox{0.5}{$\bigcirc$}} (m) \subset \mathcal{P}_{\scalebox{0.5}{$\bigcirc$}}$ (resp. $\mathcal{L}_\bullet (A) \subset \mathcal{P}_\bullet$) via Eq.~\eqref{eq: L (m) leaf} (resp. Eq.~\eqref{eq: L (A) leaf}).

\section{Emergent Wormholes from the Large-$N$ Filter}
\label{sec: Emergent Wormholes from the Large-$N$ Filter}

\subsection{Filtering out Erratic Observables of One CFT}
\label{sec: Filtering One CFT}

In this subsection, we show that filtering out the erratic observables for one CFT is equivalent to gauging the nonlocal symmetries~\eqref{eq: puncture gauge transformation}. The filtered Hilbert space $\mathcal{F}\mathcal{H}_\odot$, as the gauge-invariant subspace, describes a global AdS$_3$ spacetime with boundary gravitons. 

Recall that the algebra of the puncture monodromy~\eqref{eq: m def} contains a subalgebra $U_{\tilde{q}}(\psl)$ of erratic observables $m_{\tilde{\mathrm{e}}}$ $m_{\tilde{\mathrm{f}}}$, and $m_{\widetilde{\mathrm{K}}}$ defined by Eq.~\eqref{eq: dual m def} and satisfying Eqs.~\eqref{eq: dual U_q 1}, \eqref{eq: dual U_q 2}, and \eqref{eq: dual U_q 3}. Using the algebra $U_{\tilde{q}}(\psl)
$, calculations involving these erratic observables in general produce erratic factors in terms of $\tilde{q}$ defined by Eq.~\eqref{eq: qdual def}. The erratic large-$N$ behavior of $\tilde{q}$ is illustrated in Figure~\ref{fig: erratic q}. To filter out such erratic observables, we to restrict the full Hilbert space $\mathcal{H}_\odot$ to a subspace $\mathcal{F}\mathcal{H}_\odot$ on which the erratic observables $m_{\tilde{\mathrm{e}}}$ $m_{\tilde{\mathrm{f}}}$, and $m_{\widetilde{\mathrm{K}}}$ act trivially, i.e.
\begin{equation}
\label{eq: mdual constraints}
    m_{\tilde{\mathrm{e}}} = \mu_{\tilde{\mathrm{e}}} ,\quad m_{\tilde{\mathrm{f}}} = \mu_{\tilde{\mathrm{f}}} ,\quad m_{\widetilde{\mathrm{K}}} = \mu_{\widetilde{\mathrm{K}}} ,
\end{equation}
where $\mu_{\tilde{\mathrm{e}}}$, $\mu_{\tilde{\mathrm{f}}}$, and $\mu_{\widetilde{\mathrm{K}}}$ are three fixed numbers. Notice that $\mu_{\tilde{\mathrm{e}}}$, $\mu_{\tilde{\mathrm{f}}}$, and $\mu_{\widetilde{\mathrm{K}}}$ are not arbitrary since $m_{\tilde{\mathrm{e}}}$ $m_{\tilde{\mathrm{f}}}$, and $m_{\widetilde{\mathrm{K}}}$ are not simultaneously diagonalizable. The compatibility between Eq.~\eqref{eq: dual U_q 3} and Eq.~\eqref{eq: mdual constraints} implies~\footnote{Notice that here we do not require the operator to be positive. It will become clear later that constrain~\eqref{eq: mdual constraints} finally leads to global AdS$_3$ vacuum state which does not have to obey the positivity restriction~\eqref{eq: positivity m_e, m_f, m_K} from the wormhole edge mode interpretation.}
\begin{equation}
\label{eq: mK dual  = i, ii, iii, iiii}
    m_{\widetilde{\mathrm{K}}} = e^{\frac{i\pi}{2} n_p}, \quad n_p \in \mathbb{Z}.
\end{equation}
The compatibility among Eqs.~\eqref{eq: dual U_q 1}, \eqref{eq: mdual constraints}, and \eqref{eq: mK dual  = i, ii, iii, iiii} implies
\begin{equation}
\label{eq: me dual  = 0}
    m_{\tilde{\mathrm{e}}} = 0.
\end{equation}
The compatibility among Eqs.~\eqref{eq: dual U_q 2}, \eqref{eq: mdual constraints}, and \eqref{eq: mK dual  = i, ii, iii, iiii} implies
\begin{equation}
\label{eq: m_f dual  = 0}
    m_{\tilde{\mathrm{f}}} = 0.
\end{equation}

Now we show that filtering out the erratic observables $m_{\tilde{\mathrm{e}}}$ $m_{\tilde{\mathrm{f}}}$, and $m_{\widetilde{\mathrm{K}}}$ is equivalent to gauging the nonlocal symmetry~\eqref{eq: puncture gauge transformation}. Using Eqs.~\eqref{eq: mk quantum Kashaev}, \eqref{eq: me quantum Kashaev}, \eqref{eq: mf quantum Kashaev}, and \eqref{eq: dual m def}, one can prove that~\footnote{Eqs.~\eqref{eq: mk dual quantum Kashaev}, \eqref{eq: me dual quantum Kashaev}, and \eqref{eq: mf dual quantum Kashaev} were proved in terms of the standard realization of the modular double~\cite{Bytsko:2002br}. Eqs.~\eqref{eq: mk quantum Kashaev}, \eqref{eq: me quantum Kashaev}, and \eqref{eq: mf quantum Kashaev} is called the Whittaker model of the modular double and is unitarily equivalent to the standard realization~\cite{Nidaiev:2013bda}.}
\begin{align}
    m_{\widetilde{\mathrm{K}}} &= e^{\pi \tilde{b} p},
    \label{eq: mk dual quantum Kashaev}
    \\
    m_{\tilde{\mathrm{e}}} &= e^{\pi \tilde{b} (p-2x)},
    \label{eq: me dual quantum Kashaev}
    \\
    m_{\tilde{\mathrm{f}}} &= e^{\pi \tilde{b} (x-\frac{p}{2})} \left(2 \cosh (2 \pi \tilde{b} \lambda) + 2 \cosh(2 \pi \tilde{b} p)\right) e^{\pi \tilde{b} (x-\frac{p}{2} )},
    \label{eq: mf dual quantum Kashaev}
\end{align}
where $\tilde{b}$ is defined by
\begin{equation}
\label{eq: b dual def}
    \tilde{b} \equiv \frac{1}{\hbar} (k - 2 \hbar)^{\frac{1}{2}}.
\end{equation}
Combining Eqs.~\eqref{eq: mK dual  = i, ii, iii, iiii} and \eqref{eq: mk dual quantum Kashaev} yields
\begin{equation}
\label{eq: p = n_p i hbar b/2}
      p =  \frac{i n_p  }{2\tilde{b}}, \quad n_p \in \mathbb{Z}.
\end{equation}
Eqs.~\eqref{eq: me dual  = 0}, \eqref{eq: me dual quantum Kashaev}, and \eqref{eq: p = n_p i hbar b/2}   imply
\begin{equation}
\label{eq: x= infty}
    x = \infty.
\end{equation}
Eqs.~\eqref{eq: me quantum Kashaev} and \eqref{eq: x= infty} imply
\begin{equation}
\label{eq: m_e = 0}
    m_{\mathrm{e}} = 0.
\end{equation}
The compatibility between Eqs.~\eqref{eq: [m_e, m_f]} and \eqref{eq: m_e = 0} requires
\begin{equation}
\label{eq: m_K = iiiii}
    m_{\mathrm{K}} = e^{\frac{i\pi}{2} n_b }, \quad n_b \in \mathbb{Z}.
\end{equation}
Combining Eqs.~\eqref{eq: mk quantum Kashaev}, \eqref{eq: b def}, \eqref{eq: b dual def}, \eqref{eq: p = n_p i hbar b/2}, and \eqref{eq: m_K = iiiii} gives
\begin{equation}
\label{eq: quantized k}
    k = (\frac{n_p}{n_b + 4n_k} + 2)\hbar, \quad n_k \in \mathbb{Z}.
\end{equation}
The unitarity of the theory requires $k >0$. So the Chern-Simons level $k$ needs to be a rational number to consistently filter out the erratic observables. The compatibility between Eqs.~\eqref{eq: m_K m_f = q m_f m_K} and \eqref{eq: m_K = iiiii} requires
\begin{equation}
\label{eq: m_f = 0}
    m_{\mathrm{f}} = 0.
\end{equation}
According to Eqs.~\eqref{eq: m_e = 0}, \eqref{eq: m_K = iiiii}, and \eqref{eq: m_f = 0}, consistently filtering out the erratic observables fixes $m_{\mathrm{K}}$, $m_{\mathrm{e}}$, and $m_{\mathrm{f}}$. Thus, Eqs.~\eqref{eq: Q def PSL^*} and \eqref{eq: dual PSL def} imply that the corresponding nonlocal symmetry~\eqref{eq: puncture gauge transformation} is gauged. Conversely, by switching $b$ with $\tilde{b}$ in the above discussion, one can show that gauging the nonlocal symmetry~\eqref{eq: puncture gauge transformation} filters out the erratic observables $m_{\tilde{\mathrm{e}}}$ $m_{\tilde{\mathrm{f}}}$, and $m_{\widetilde{\mathrm{K}}}$.

Using parametrization~\eqref{eq: m_e,f,k parameterization},
the constraints~\eqref{eq: m_e = 0}, \eqref{eq: m_K = iiiii}, and \eqref{eq: m_f = 0} can be compactly written as
\begin{equation}
\label{eq: m = 1}
    m  = \mathbf{1}
\end{equation}
with $m$ a matrix-valued operator
\begin{equation}
\label{eq: m matrix operator def}
    m \equiv \begin{pmatrix}
        m_{\mathrm{K}}^2 & m_{\mathrm{K}} m_{\mathrm{f}}
        \\
        -m_{\mathrm{e}} m_{\mathrm{K}} & m_{\mathrm{K}}^{-2} -m_{\mathrm{f}}m_{\mathrm{e}} 
    \end{pmatrix}
\end{equation}
and with $\mathbf{1}$ the identity matrix
\begin{equation}
\label{eq: 1 matrix def}
    \mathbf{1} \equiv \begin{pmatrix}
        1 & 0
        \\
        0 & 1
    \end{pmatrix}.
\end{equation}
Notice that $\mathbf{1}$ is identified with $-\mathbf{1}$ in $\PSL$. Denote by $\mathcal{F}\mathcal{H}_\odot$ the filtered subspace defined by the constraint~\eqref{eq: m = 1}. Since gauging the nonlocal symmetry~\eqref{eq: puncture gauge transformation} reduces the full observable algebra
to the boundary graviton observable algebra as shown in Section~\ref{sec: Gauging Nonlocal Symmetries}, the filtered observable algebra $\tx{End}(\mathcal{F}\mathcal{H}_\odot)$, as the gauge-invariant observable algebra, is isomorphic to the boundary graviton observable algebra $\mathcal{A}_{\bdy}$,
\begin{equation}
\label{eq: Abdy cong End(FH_odot)}
    \mathcal{A}_{\bdy} \cong \tx{End}(\mathcal{F}\mathcal{H}_\odot).
\end{equation}
The constraint~\eqref{eq: m = 1} corresponds to the trivial class of the interior phase spaces discussed in Section~\ref{sec: Trivial class}, which excludes black hole solutions. So filtering out erratic observables for one CFT effectively filters out all black holes states. If the Euler class $n$ associated with the trivial class is set as $\pm1$ as in Eq.~\eqref{eq: n = -+1}, then $\mathcal{F}\mathcal{H}_\odot$ describes a global AdS$_3$ spacetime with boundary gravitons. See Figure~\ref{fig: Filtering one CFT} for an illustration.
\begin{figure}
    \centering
    \includegraphics[width=0.55\linewidth]{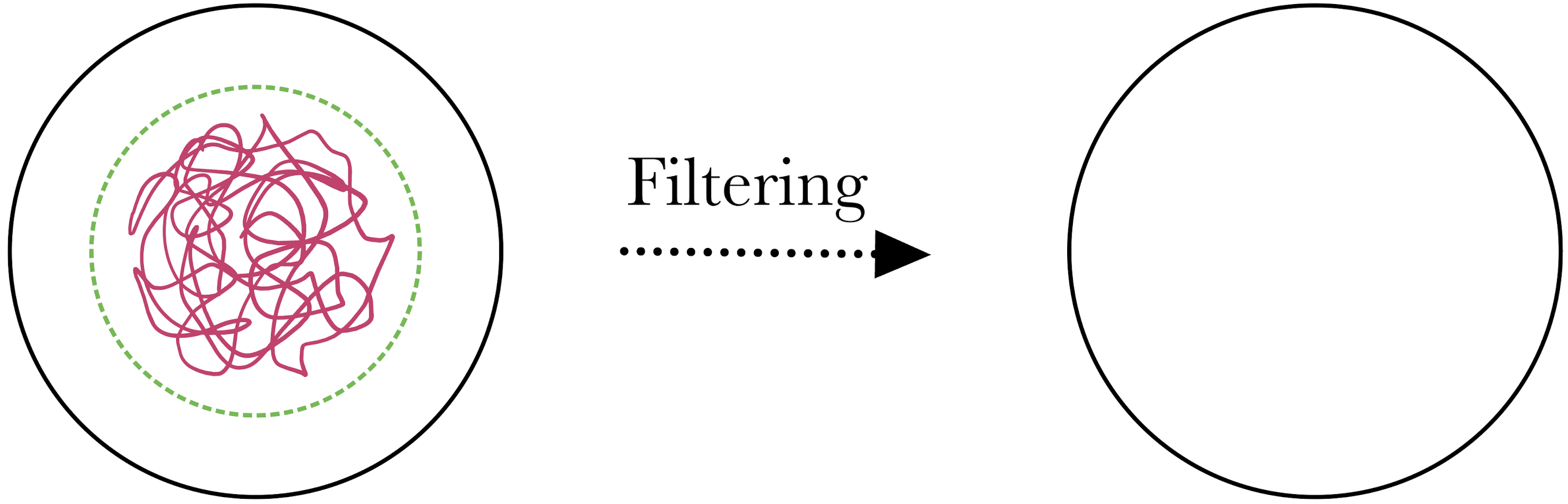}
    \caption{On the left is a schematic illustration of a bulk Cauchy slice dual to one CFT. The outer black circle represents boundary gravitons. The green dashed circle represents the horizon. The red mess inside the horizon represents the black hole degrees of freedom characterized by the erratic observables $m_{\tilde{\mathrm{e}}}$ $m_{\tilde{\mathrm{f}}}$, and $m_{\widetilde{\mathrm{K}}}$. The right picture represents the bulk Cauchy slice dual to the filtered CFT. The black hole states are filtered out. The remaining is the boundary gravitons associated with the global AdS$_3$.}
    \label{fig: Filtering one CFT}
\end{figure}
The tensor product $\mathcal{F}\mathcal{H}_\odot \otimes \mathcal{F}\mathcal{H}_\odot$ hence describes two disconnected global AdS$_3$ spacetimes. In Section~\ref{sec: Apparent Ensemble Averaging over Entanglers}, we will show that the gauge-invariant observable algebra of two CFTs is larger than the boundary graviton algebra because of contribution from the smooth remnant of correlations of erratic observables. 

In terms of the rescaled monodromy variables~\eqref{eq: resacled m def}, the constraint~\eqref{eq: m = 1} is equivalent to
\begin{equation}
\label{eq: quantum constraint in terms of rescaled monodromies}
    m_{\mathrm{E}}  = m_{\mathrm{F}} =m_{\mathrm{H}}  = 0.
\end{equation}
Denote by $\mathcal{U}$ the unitary group generated by the rescaled monodromy variables,
\begin{equation}
\label{eq: cal U def}
    \mathcal{U} \equiv \{ \exp ( i \sum_{\alpha = \mathrm{E}, \mathrm{F},\mathrm{H}} X^\alpha m_{\alpha}  );\, X^\alpha \in \mathbb{R} \}.
\end{equation}
Denote by $\mathcal{H}_\odot$ the full Hilbert space of the CFT. Then the projection $\mathcal{F}: \mathcal{H}_\odot \to \mathcal{F}\mathcal{H}_\odot$ implementing the constraint~\eqref{eq: quantum constraint in terms of rescaled monodromies} can be constructed as
\begin{equation}
\label{eq: one CFT filter}
    \mathcal{F}  \equiv \underset{U \in \mathcal{U}}{\int} dU\,U ,
\end{equation}
with the left-invariant Haar measure $dU$ normalized as~\footnote{One may worry that the Haar measure does not exist due to the non-compactness of $U_q(\psl)$. However, a left-invariant Haar measure for the modular double $U_{q\tilde{q}} (\psl) $ is well-defined from a dual perspective~\cite{Bytsko:2002br}. I believe the two dual perspectives are essentially equivalent and leave a more rigorous treatment for future works.}
\begin{equation}
\label{eq: Haar measure normalization}
    \underset{U \in \mathcal{U}}{\int} dU = 1.
\end{equation}
The left-invariance of the Haar measure $dU$ implies that the filtered Hilbert space $\mathcal{F}\mathcal{H}_\odot $ is $\mathcal{U}$-invariant, i.e.
\begin{equation}
    U \mathcal{F}  = \mathcal{F} ,\quad \forall U \in \mathcal{U},
\end{equation}
which is equivalent to the constraint~\eqref{eq: quantum constraint in terms of rescaled monodromies}. The erratic observables are also filtered out by the projection $\mathcal{F}$~\eqref{eq: one CFT filter} because of the equivalence between filtering out the erratic observables and gauging the nonlocal symmetry~\eqref{eq: puncture gauge transformation}. The projection $\mathcal{F}$ hence can also be constructed in terms of the erratic observables,
\begin{equation}
    \mathcal{F}  \equiv \underset{U \in \widetilde{ \mathcal{U}}}{\int} dU\,U ,
\end{equation}
where $\widetilde{\mathcal{U}}$ is the unitary group
\begin{equation}
    \widetilde{\mathcal{U} } \equiv \{ \exp ( i \sum_{\alpha = \widetilde{\mathrm{E}}, \widetilde{\mathrm{F}}, \widetilde{\mathrm{H}}} X^\alpha m_{\alpha}  );\, X^\alpha \in \mathbb{R} \}
\end{equation}
generated by rescaled erratic observable generators
\begin{equation}
    m_{\widetilde{\mathrm{E}}} \equiv \frac{m_{\tilde{\mathrm{e}}}}{2\sin{\frac{ \pi k}{\hbar}}}, \quad m_{\widetilde{\mathrm{F}}} \equiv \frac{m_{\tilde{\mathrm{f}}}}{2\sin{\frac{ \pi k}{\hbar}}}, \quad m_{\widetilde{\mathrm{H}}} \equiv i \log_{\tilde{q}} m_{\widetilde{\mathrm{K}}}^2.
\end{equation}
The gauged nonlocal symmetries admit an intuitive interpretation as follows. The puncture monodromy observables associated with the black hole interior/horizon edge modes is extremely difficult to access for a boundary observer due to the complexity obstruction. Consequently, actions of unitary operators in the algebra of monodromy observables cannot be efficiently observed and hence effectively form a gauge group.

\subsection{Apparent Ensemble Averaging over Entanglers}
\label{sec: Apparent Ensemble Averaging over Entanglers}

In this subsection, we make the notion of the large-$N$ filter explicit for two CFTs. In Section~\ref{sec: Filtering One CFT}, we introduced for one CFT the projection $\mathcal{F}$~\eqref{eq: one CFT filter} filtering the erratic observables. The filtered Hilbert space $\mathcal{F}\mathcal{H}_\odot$ and the filtered observable algebra $\tx{End}(\mathcal{F}\mathcal{H}_\odot)$ play the role analogous to the filtered partition function~\eqref{eq: filtered one CFT partition function} of one CFT. One can implement the single-CFT filter on two CFTs individually, which yields the filtered Hilbert space $\mathcal{F} \mathcal{H}_\odot \otimes \mathcal{F} \mathcal{H}_\odot$. As we explained in Section~\ref{sec: Filtering One CFT}, $\mathcal{F} \mathcal{H}_\odot \otimes \mathcal{F} \mathcal{H}_\odot$ only describes two disconnected spacetimes. $\mathcal{F} \mathcal{H}_\odot \otimes \mathcal{F} \mathcal{H}_\odot$ plays the role of the first term on the right hand side of Eq.~\eqref{eq: F( Z Z )}. To see the emergent wormhole term, we can instead gauge the ``global'' part $\hat{\mathcal{U}}$ of $\mathcal{U} \otimes \mathcal{U}$ which is isomorphic to a single group $\mathcal{U}$ defined by Eq.~\eqref{eq: cal U def}. For example, we need to consider two replicas of a system to calculate the second R\'enyi entropy. In such a case, each unitary operator has identical actions on the two replicas, which is what we mean by ``global''. In general, there may be smooth remnant of the correlations between erratic behaviors if only the net erratic behaviors are filtered out. One may expect that the charges corresponding to the $\hat{\mathcal{U}}$ are simply the sum of the charges of the two systems: $m_{\mathrm{\alpha}} \otimes 1 + 1\otimes m_{\mathrm{\alpha}}$, $\alpha = \mathrm{E}, \mathrm{F}, \mathrm{H}$. But it is straightforward to check that these naive composite charges do not satisfy the same $U_q(\psl)$ algebra. It needs a more careful treatment. 

Given two classical systems with symplectic forms $\Omega_1$ and $\Omega_2$, suppose they share the same ordinary Lie group symmetry corresponding to a phase space vector field $\mathfrak{X}$. Using Eq.~\eqref{eq: imath X Omega = delta Q}, we have
\begin{equation}
\label{eq: two charges}
    \imath_{\mathfrak{X}} \Omega_1 = \delta Q_1, \quad  \imath_{\mathfrak{X}} \Omega_2 = \delta Q_2.
\end{equation}
$Q_1$ and $Q_2$ are the associated symmetry charges. The symplectic form of the composite system is
\begin{equation}
    \Omega = \Omega_1 + \Omega_2.
\end{equation}
Eq.~\eqref{eq: two charges} then implies
\begin{equation}
    \imath_{\mathfrak{X}} \Omega = \delta (Q_1 + Q_2).
\end{equation}
So the composite charge is indeed the sum of the charges of the two subsystems. By Definition~\eqref{eq: global charge def}, the composite global charge is the product
\begin{equation}
\label{eq: fusion product}
    \mathcal{Q}  = \mathcal{Q}_1 \mathcal{Q}_2
\end{equation}
with
\begin{equation}
    \mathcal{Q} = e^{Q_1 + Q_2},
\quad
    \mathcal{Q}_1 = e^{Q_1}, \quad \mathcal{Q}_2 = e^{Q_2}.
\end{equation}
In fact, the composition rule~\eqref{eq: fusion product} holds for more general situations~\cite{alekseev1998lie}. Combining Eqs.~\eqref{eq: Q def PSL^*}, \eqref{eq: dual PSL def}, and \eqref{eq: fusion product} gives rise to the composite charges. They correspond to operators
\begin{align}
    \hat{m}_{\mathrm{e}} &= m_{\mathrm{e}} \otimes m_{\mathrm{K}}^{-1} + m_{\mathrm{K}}\otimes m_{\mathrm{e}},
        \label{eq: me coproduct}
    \\
    \hat{m}_{\mathrm{f}}  &= m_{\mathrm{f}} \otimes m_{\mathrm{K}}^{-1}  + m_{\mathrm{K}} \otimes m_{\mathrm{f}} ,
    \label{eq: mf coproduct}
    \\
    \hat{m}_{\mathrm{K}}  &= m_{\mathrm{K}} \otimes m_{\mathrm{K}}.
    \label{eq: mK coproduct}
\end{align}
Using Eqs.\eqref{eq: m_K m_e = q m_e m_K}, \eqref{eq: m_K m_f = q m_f m_K}, and \eqref{eq: [m_e, m_f]}, it is straightforward to check that the composite charges satisfy the same $U_q(\psl)$ algebra. Using Eqs.\eqref{eq: dual U_q 1}, \eqref{eq: dual U_q 2}, and \eqref{eq: dual U_q 3}, it is also straightforward to verify that the composites
\begin{align}
    \hat{m}_{\tilde{\mathrm{e}}}   &= m_{\tilde{\mathrm{e}}} \otimes m_{\widetilde{\mathrm{K}}}^{-1} + m_{\widetilde{\mathrm{K}}}\otimes m_{\tilde{\mathrm{e}}},
    \\
    \hat{m}_{\tilde{\mathrm{f}}}   &= m_{\tilde{\mathrm{f}}} \otimes m_{\widetilde{\mathrm{K}}}^{-1}  + m_{\widetilde{\mathrm{K}}} \otimes m_{\tilde{\mathrm{f}}} ,
    \\
    \hat{m}_{\widetilde{\mathrm{K}}}    &= m_{\widetilde{\mathrm{K}}} \otimes m_{\widetilde{\mathrm{K}}}
\end{align}
of the modular dual charges $m_{\tilde{\mathrm{e}}}$, $m_{\tilde{\mathrm{f}}}$, and $m_{\widetilde{\mathrm{K}}}$ defined by Eq.~\eqref{eq: dual m def} satisfy the same $U_{\tilde{q}} (\psl)$ algebra. $\hat{m}_{\tilde{\mathrm{e}}}$, $\hat{m}_{\tilde{\mathrm{f}}}$, and $\hat{m}_{\widetilde{\mathrm{K}}}$ can be viewed as the global part of the erratic observables of the two CFTs, because they are related to the global charges by
\begin{equation}
\label{eq: dual coproduct}
    \left(\hat{m}_{\mathrm{e}} \right)^{\tilde{b}^2 \hbar}  = \hat{m}_{\tilde{\mathrm{e}}} ,
\quad
    \left(\hat{m}_{\mathrm{f}} \right)^{\tilde{b}^2 \hbar}  = \hat{m}_{\tilde{\mathrm{f}}} ,
\quad
    \left(\hat{m}_{\mathrm{K}} \right)^{\tilde{b}^2 \hbar}   = \hat{m}_{\widetilde{\mathrm{K}}} ,
\end{equation}
See~\cite{Bytsko:2002br} for a derivation of Eq.~\eqref{eq: dual coproduct}.

By the same reasoning that leads to the constraint~\eqref{eq: m = 1}, we impose the constraint
\begin{equation}
\label{eq: fixing composite charges}
    \hat{m}_{\mathrm{e}}= \hat{m}_{\mathrm{f}} = 0, \quad \hat{m}_{\mathrm{K}}  = 1,
\end{equation}
on $ \mathcal{H}_\odot \otimes \mathcal{H}_\odot$ to define the filtered subspace $\mathcal{F}(\mathcal{H}_\odot \otimes \mathcal{H}_\odot)$. In terms of the rescaled monodromy variables~\eqref{eq: resacled m def}, the constraint~\eqref{eq: fixing composite charges} is equivalent to
\begin{equation}
\label{eq: fixing composite charges 2}
    \hat{m}_{\mathrm{E}} = \hat{m}_{\mathrm{F}}  = \hat{m}_{\mathrm{H}} = 0
\end{equation}
with composite rescaled charges given by
\begin{align}
    \hat{m}_{\mathrm{E}} &= m_{\mathrm{E}} \otimes q^{\frac{i}{2} m_{\mathrm{H}} } + q^{-\frac{i}{2} m_{\mathrm{H}}}\otimes m_{\mathrm{E}},
    \label{eq: composite rescaled charge E}
    \\
    \hat{m}_{\mathrm{F}}  &= m_{\mathrm{F}} \otimes q^{\frac{i}{2} m_{\mathrm{H}} } + q^{-\frac{i}{2} m_{\mathrm{H}}}\otimes m_{\mathrm{F}},
    \label{eq: composite rescaled charge F}
    \\
    \hat{m}_{\mathrm{H}}  &= m_{\mathrm{H}} \otimes 1 + 1 \otimes m_{\mathrm{H}}.
    \label{eq: composite rescaled charge H}
\end{align}
Similar to Definition~\eqref{eq: cal U def}, the gauge group $\hat{\mathcal{U}}$ of global operations of high complexity is given by
\begin{equation}
\label{eq: Uhat def}
    \hat{\mathcal{U}} \equiv \{ \exp ( i \sum_{\alpha = \mathrm{E}, \mathrm{F},\mathrm{H}} X^\alpha  \hat{m}_{\alpha}  );\, X^\alpha \in \mathbb{R} \}.
\end{equation}
The projection $\mathcal{F}$ onto the $\hat{\mathcal{U}}$-invariant subspace of the Hilbert space $\mathcal{H}_\odot \otimes \mathcal{H}_\odot$ of the full system can be constructed as
\begin{equation}
\label{eq: Uhat projection}
    \mathcal{F}  = \underset{U \in \hat{\mathcal{U}}}{\int} dU\,U . 
\end{equation}

Viewing the two CFTs as two wires in a quantum circuit, the charge $\hat{m}_{\mathrm{H}}$ \eqref{eq: composite rescaled charge H} is a 1-local quantum gate, while $\hat{m}_{\mathrm{E}}$ \eqref{eq: composite rescaled charge E} and $\hat{m}_{\mathrm{F}}$ \eqref{eq: composite rescaled charge F} are 2-local quantum gates. So typical elements in $\hat{\mathcal{U}}$ are quantum circuits generating entanglement between the monodromy/black hole degrees of freedom of the two CFTs. Such entanglement structures are indistinguishable to an exterior observer due to the complexity barrier or the erratic behavior. The projection $\mathcal{F}$~\eqref{eq: Uhat projection} as the large-$N$ filter is hence implemented by the ensemble averaging over these indistinguishable circuits in $\hat{\mathcal{U}}$. The individually filtered Hilbert space $\mathcal{F} \mathcal{H}_\odot \otimes \mathcal{F}\mathcal{H}_\odot $ is spanned by pure states satisfying the constrains
\begin{equation}
\label{eq: m O 1 = 1 O 1}
    m \otimes \mathbf{1} = \mathbf{1} \otimes \mathbf{1}
\end{equation}
and
\begin{equation}
\label{eq: 1 O m = 1 O 1}
     \mathbf{1} \otimes m = \mathbf{1} \otimes \mathbf{1}
\end{equation}
with $m$ and $\mathbf{1}$ defined by Eq.~\eqref{eq: m matrix operator def} and Eq.~\eqref{eq: 1 matrix def} respectively. Combining constraints~\eqref{eq: m O 1 = 1 O 1} and \eqref{eq: 1 O m = 1 O 1} yields the constraint~\eqref{eq: fixing composite charges}. The constraint~\eqref{eq: fixing composite charges} admits more solutions, so the individually filtered Hilbert space is a proper subspace of the filtered Hilbert space $\mathcal{F}(\mathcal{H}_\odot \otimes \mathcal{H}_\odot)$, i.e.
\begin{equation}
\label{eq: F(H otimes H) = FH otimes FH + H^wormhole}
    \mathcal{F} (\mathcal{H}_\odot \otimes \mathcal{H}_\odot) = ( \mathcal{F} \mathcal{H}_\odot \otimes \mathcal{F}\mathcal{H}_\odot )\oplus \mathcal{H}^{\tx{wormhole}},
\end{equation}
where $\mathcal{H}^{\tx{wormhole}}$ is a nontrivial subspace. Unlike the case of filtering one CFT where Eq.~\eqref{eq: Abdy cong End(FH_odot)} holds, the filtered observable algebra $\tx{End}(\mathcal{F}(\mathcal{H}_\odot \otimes \mathcal{H}_\odot))$ is hence larger than the observable algebra $\mathcal{A}_{\bdy} \otimes \mathcal{A}_{\bdy} \cong \tx{End} ( \mathcal{F}\mathcal{H}_\odot \otimes \mathcal{F}\mathcal{H}_\odot)$ of boundary gravitons on the two asymptotic boundaries, i.e.
\begin{equation}
\label{eq: filtered algebra constains boundary graviton algebras}
    \mathcal{A}_{\bdy} \otimes \mathcal{A}_{\bdy} \subsetneq \tx{End}(\mathcal{F}(\mathcal{H}_\odot \otimes \mathcal{H}_\odot)).
\end{equation}
What makes $\tx{End}(\mathcal{F}(\mathcal{H}_\odot \otimes \mathcal{H}_\odot))$ larger is a set of observables encoding an emergent wormhole bulk. 

The full unfiltered observable algebra $\mathcal{A}_\odot \otimes \mathcal{A}_\odot$ is the tensor product of the quantization $\mathcal{A}_\odot$ of the full classical observable algebra $C^\infty (\mathcal{P}_\odot)$ defined by Poisson bracket~\eqref{eq: simplest WW Poisson}. The theory is defined on two punctured discs. Denote by $W_l$ (resp. $W_r$) the Wilson line in the first disc (resp. second disc) connecting the puncture to the boundary. The operator $W_l W_r^{-1}$ is an element in $\tx{End}(\mathcal{F} (\mathcal{H}_\odot \otimes \mathcal{H}_\odot))$. In other words, $W_l W_r^{-1}$ commutes with the composite charges~\eqref{eq: me coproduct}, \eqref{eq: mf coproduct}, and \eqref{eq: mK coproduct}. Using~Eqs.\eqref{eq: simplest WW Poisson}, \eqref{eq: m def}, and \eqref{eq: m_e,f,k parameterization}, this is straightforward to verify at the classical level. It also holds at the quantum level since we expect that the quantization preserves the commutant property, otherwise there will be quantum corrections to a quantity that vanishes at the classical level. On the contrary, $W_l W^{-1}_r$ is not an element in $\mathcal{A}_{\bdy} \otimes \mathcal{A}_{\bdy}$. Because the nonlocal symmetry~\eqref{eq: puncture gauge transformation} acting on $W_l$ and $W_r$ individually does not keep $W_l W_r^{-1}$ invariant. So $W_l W_r^{-1}$ does not commute with the corresponding charges which are the puncture monodromy variables according to Eqs.~\eqref{eq: Q def PSL^*} and \eqref{eq: dual PSL def} and hence are elements in the commutant of the boundary graviton algebra. Using Eq.~\eqref{eq: simplest WW Poisson}, it is straightforward to check that $W_l W^{-1}_r$ satisfies the two-sided Wilson line algebra~\eqref{eq: Two-Sided Wilson Line Poisson Bracket}. The fact that it does not belong to the individually filtered observable algebra $\tx{End} ( \mathcal{F}\mathcal{H}_\odot \otimes \mathcal{F}\mathcal{H}_\odot) \cong \mathcal{A}_{\bdy} \otimes \mathcal{A}_{\bdy}$ is compatible with the fact that $ \mathcal{F}\mathcal{H}_\odot \otimes \mathcal{F}\mathcal{H}_\odot$ describes two disjoint spacetimes as explained in Section~\ref{sec: Filtering One CFT}. The filtered observable algebra $\tx{End}(\mathcal{F}(\mathcal{H}_\odot \otimes \mathcal{H}_\odot))$ contains a subalgebra of $W_lW^{-1}_r$ that is isomorphic to an algebra of Wilson lines supported on a two-sided wormhole spacetime. This indicates the emergence of a wormhole not encoded in the individually filtered observable algebra~$\mathcal{A}_{\bdy} \otimes \mathcal{A}_{\bdy}$. At the level of Hilbert space, the wormhole is encoded into corresponding extra part $\mathcal{H}^{\tx{wormhole}}$ in Eq.~\eqref{eq: F(H otimes H) = FH otimes FH + H^wormhole}. Thus, Eq.~\eqref{eq: F(H otimes H) = FH otimes FH + H^wormhole} is a Hilbert space version of Eq.~\eqref{eq: F( Z Z )} with the projection $\mathcal{F}$ playing the role of the large-$N$ filter.

Now we describe the wormhole subspace $\mathcal{H}^{\tx{wormhole}}$ in detail. As shown in Section~\ref{sec: Gauging Nonlocal Symmetries}, the Wilson loop $\tr\, m$ of winding number 1 generates the center of the boundary graviton algebra. The Wilson loop is related to the Kashaev coordinate $\lambda$ by Eq.~\eqref{eq: lambda  def}. So the  irreducible representations $V_\lambda$ of the boundary graviton algebra are labeled by $\lambda$. Using Eqs.~\eqref{eq: trace normalization} and Eq.~\eqref{eq: m_e,f,k parameterization}, it is straightforward to verify that the classical Wilson loop is equal to the Casimir element in the classical puncture monodromy algebra $C^\infty (\mathcal{P}_\bullet)$, i.e.
\begin{equation}
\label{eq: classical Casimir}
    \tr\, m =  - m_{\mathrm{K}}^2 - m_{\mathrm{K}}^{-2} + m_{\mathrm{f}} m_{\mathrm{e}}.
\end{equation}
If the quantization preserves the commutant property, then the quantized Wilson loop $\tr_q \, m$ will still be the Casimir element at the quantum level~\footnote{The erratic modular dual observables form another central element without a large $N$ limit. But its eigenvalue still depends on $\lambda$.}. The quantized Casimir element is given by~\cite{Nidaiev:2013bda}
\begin{equation}
\label{eq: quantum Casimir}
   \tr_q\,m =  -q\, m_{\mathrm{K}}^2 - q^{-1} m_{\mathrm{K}}^{-2} + m_{\mathrm{f}} m_{\mathrm{e}}.
\end{equation}
With the realization \eqref{eq: mk quantum Kashaev}, \eqref{eq: me quantum Kashaev}, and \eqref{eq: mf quantum Kashaev}, we have the eigenvalue equation~\cite{Kashaev:2000ku, Nidaiev:2013bda}
\begin{equation}
\label{eq: eigen Eq of tr_q m}
    (\tr_q\, m)\, V_\lambda = 2\cosh (2 \pi b \lambda)\, V_\lambda.
\end{equation}
According to Section~\ref{sec: Hyperbolic class}, $V_\lambda$ with $\lambda>0$ describes black holes states. Denote by $|\lambda; a \rangle$ the basis of the irreducible representation $V_\lambda$. For example, $a$ labels the Virasoro descendants if the asymptotically AdS$_3$ boundary condition~\eqref{eq: AdS boundary condition 1} is imposed. But we leave the boundary condition unspecified since it does not affect our discussion. The basis $|\lambda; a \rangle$ by definition satisfies Eq.~\eqref{eq: eigen Eq of tr_q m}. According to Eq.~\eqref{eq: filtered algebra constains boundary graviton algebras}, the filtered subspace $\mathcal{F}(\mathcal{H}_\odot \otimes \mathcal{H}_\odot)$ is a representation of $\mathcal{A}_{\bdy} \otimes \mathcal{A}_{\bdy}$ spanned by
\begin{equation}
\label{eq: reps of Abdy otimes Abdy}
    | \lambda_1; a_1 \rangle \otimes | \lambda_2; a_2 \rangle.
\end{equation}
$\lambda_1$ and $\lambda_2$ are related by the constraint~\eqref{eq: fixing composite charges}. Combining Eqs.~\eqref{eq: me coproduct}, \eqref{eq: mf coproduct}, \eqref{eq: mK coproduct}, and \eqref{eq: fixing composite charges} yields
\begin{equation}
\label{eq: 1 otimes m = m otimes 1}
    1 \otimes m_{\mathrm{e}} = - m_{\mathrm{e}} \otimes 1, \quad 1 \otimes m_{\mathrm{f}} = - m_{\mathrm{f}} \otimes 1, \quad 1 \otimes m_{\mathrm{K}} = m_{\mathrm{K}}^{-1} \otimes 1.
\end{equation}
Combining Eq.~\eqref{eq: quantum Casimir} and Eq.~\eqref{eq: 1 otimes m = m otimes 1}, we have
\begin{equation}
   1 \otimes \tr_q\, m = \tr_q \, m \otimes 1.
\end{equation}
Eq.~\eqref{eq: eigen Eq of tr_q m} then implies
\begin{equation}
    \lambda_1 = \lambda_2.
\end{equation}
Thus, the filtered subspace $\mathcal{F}(\mathcal{H}_\odot \otimes \mathcal{H}_\odot)$ decomposes as
\begin{equation}
\label{eq: F(H otimes H) = int}
    \mathcal{F}(\mathcal{H}_\odot \otimes \mathcal{H}_\odot) = \int V_\lambda \otimes V_\lambda.
\end{equation}
According to Eq.~\eqref{eq: eigen Eq of tr_q m}, irrep $V_\lambda$ with $\lambda >0$ corresponds to geometry with Wilson loop larger than 1. As discussed in Section~\ref{sec: Hyperbolic class}, they correspond to black hole solutions. Suppose $\mathcal{H}^{\tx{wormhole}}$ only contains black hole states. Then comparing Eqs.~\eqref{eq: F(H otimes H) = FH otimes FH + H^wormhole} and Eq.~\eqref{eq: F(H otimes H) = int} gives rise to
\begin{equation}
\label{eq: H^wormhole decomposition}
    \mathcal{H}^{\tx{wormhole}} = \underset{0 < \lambda}{\int} V_\lambda \otimes V_\lambda.
\end{equation}
So $\mathcal{H}^{\tx{wormhole}}$ describes two-sided black holes~\footnote{$\mathcal{H}^{\tx{wormhole}}$ should also admit a mapping class group action which may be encoded in the algebra generated by the Wilson loop $\tr_q\, m$ and the Wilson line $W_l W_r^{-1}$. We assume this action is modded out and leave a clearer treatment for future works.}. See Figure~\ref{fig: filtering two CFTs} for an illustration. Eq.~\eqref{eq: H^wormhole decomposition} only shows that in general which states can survive after filtering. In a model of 3d gravity with a concrete Hamiltonian, the wormhole states appear in the theory may only form a subspace of $\mathcal{H}^{\tx{wormhole}}$. We do not specify a concrete model in this paper.
\begin{figure}
    \centering
    \includegraphics[width=0.45\linewidth]{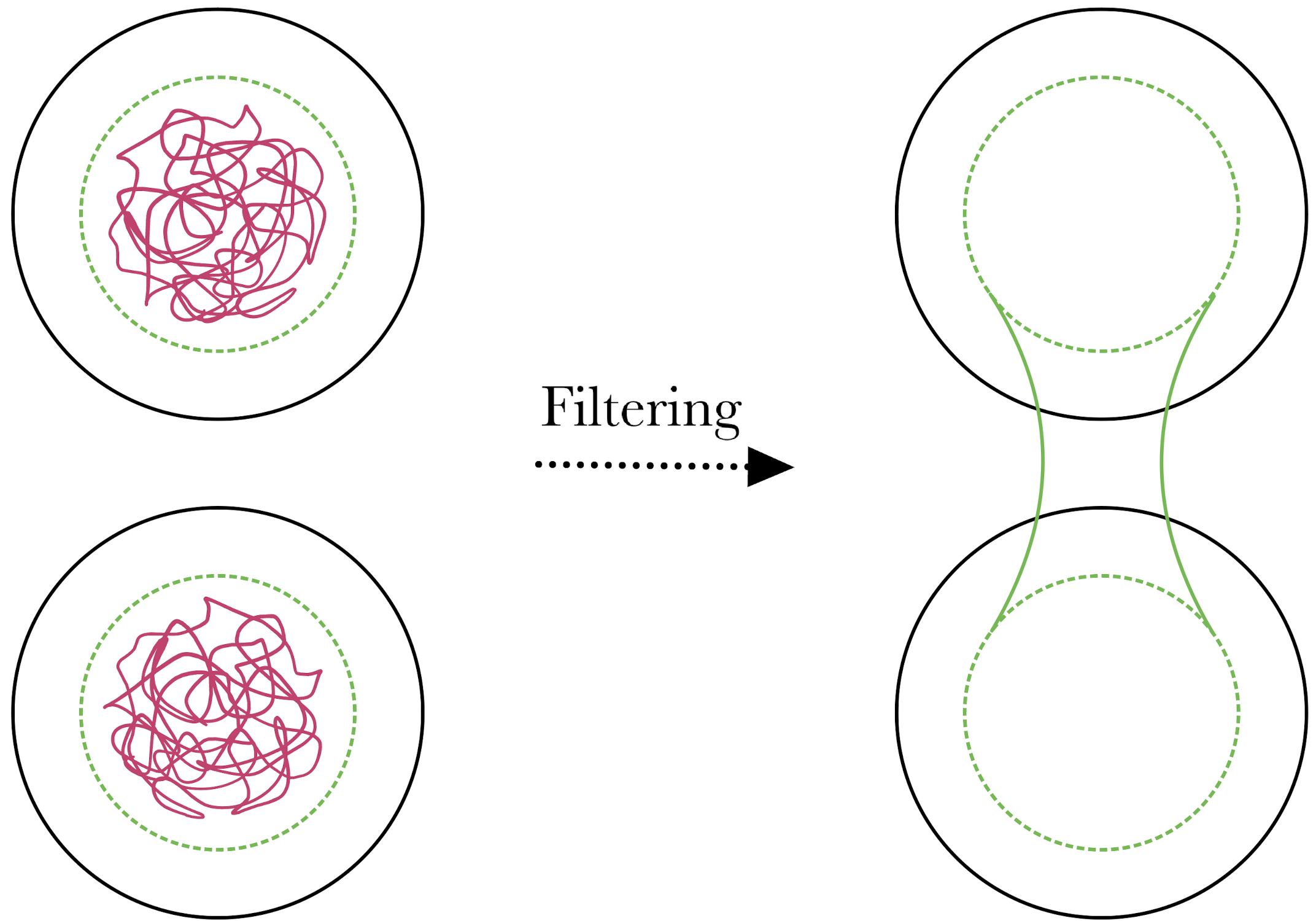}
    \caption{On the left is a schematic illustration of a bulk Cauchy slice dual to two CFTs. Each outer black circle represents boundary gravitons on each asymptotic boundary. Each green dashed circle represents the horizon for each black hole spacetime. The red mess inside the horizon represents the black hole degrees of freedom characterized by the erratic observables $m_{\tilde{\mathrm{e}}}$ $m_{\tilde{\mathrm{f}}}$, and $m_{\widetilde{\mathrm{K}}}$ of each CFT. More precisely, it represents states on which the erratic observables act nontrivially. The right picture is a schematic illustration of the bulk Cauchy slice dual to the filtered CFTs. The green solid lines represent the emergent wormhole geometry dual to the smooth remnant of the correlations between erratic observables of the two CFT.}
    \label{fig: filtering two CFTs}
\end{figure}

Notice that the filtered subspace $\mathcal{F}\mathcal{H}_\odot$ of one CFT does not describe (chiral sector of) the massless spinless BTZ black hole corresponding to the parabolic class phase spaces~\eqref{eq: leaf cong PSL(2, R)/N} discussed in Section~\ref{sec: Parabolic Class}. As explained in Section~\ref{sec: Filtering One CFT}, Eq.~\eqref{eq: m = 1} implies that $\mathcal{F}\mathcal{H}_\odot$ (with constraint~\eqref{eq: n = -+1}) is the Hilbert space of chiral boundary gravitons associated with the global AdS$_3$ spacetime. The potential confusion may be caused by the fact that, as remarked at the end of Section~\ref{sec: Parabolic Class}, the classical Wilson loop $\tr\, m$~\eqref{eq: tr m = cosh} and hence the Kashaev coordinate $\lambda$~\eqref{eq: lambda  def} does not distinguish the trivial class discussed in Section~\ref{sec: Trivial class} from the parabolic class, i.e. $\lambda = 0$ for both classes. One needs the full monodromy data, as in Eq.~\eqref{eq: m = 1}, to distinguish them.

Using Eq.~\eqref{eq: F(H otimes H) = FH otimes FH + H^wormhole}, we can make the wormhole emergence from the large-$N$ filter $\mathcal{F}$ more explicit. Consider two independent CFTs with commuting Hamiltonians $H_1 \otimes 1$ and $1 \otimes H_2$ individually. Denote the Hamiltonians by $H_1$ and $H_2$ for simplicity. Suppose $H_1 $ and $H_2$  are elements in the unfiltered observable algebra $\mathcal{A}_\odot \otimes \mathcal{A}_\odot$, which means they may behave erratically in the large $N$ limit ($\G \to 0$). Consider the partition function of the two CFTs
\begin{equation}
    Z \equiv \tx{Tr}_{\mathcal{H}_\odot \otimes \mathcal{H}_\odot} ( e^{-\beta_1 H_1} e^{ - \beta_2 H_2} ).
\end{equation}
$\tx{Tr}_{\mathcal{H}}$ denotes the trace over the Hilbert space $\mathcal{H}$. $Z$ is factorized into the product of the partition functions of the two CFTs, i.e.
\begin{equation}
    Z = Z_1 Z_2
\end{equation}
with
\begin{equation}
    Z_j \equiv \tx{Tr}_{\mathcal{H}_\odot} ( e^{-\beta_j  H_j} ), \quad j = 1,2.
\end{equation}
Recall that actions of unitaries in $\hat{\mathcal{U}}$~\eqref{eq: Uhat def} are expected to be extremely difficult to observe. The partition function
\begin{equation}
\label{eq: Z_U def}
    Z [U] \equiv \tx{Tr}_{ \mathcal{H}_\odot \otimes \mathcal{H}_\odot} (  e^{-\beta_1  H_1} e^{ - \beta_2 H_2  )} U)
\end{equation}
modified by a circuit $U \in \hat{\mathcal{U}}$ and the original partition function $Z$ are hence almost operationally indistinguishable to a boundary observer due to the complexity barrier, i.e.
\begin{equation}
\label{eq: Z indistinguishability}
    Z[U] \approx Z.
\end{equation}
From the perspective of Lorentzian evolution, $U \in \hat{\mathcal{U}}$ can be viewed as an evolution operator only involving the puncture monodromy which is an effective description of degrees of freedom behind the horizon. Hence the evolution operator $\exp{(it_1 H_1 + i t_2 H_2)}$ is almost indistinguishable from the evolution operator $\exp{(it_1 H_1 + i t_2 H_2)} U$ for a boundary observer. 

The approximate indistinguishability~\eqref{eq: Z indistinguishability} naturally leads to a partition function averaged over the ensemble of circuits in $\hat{\mathcal{U}}$~\eqref{eq: Uhat def},
\begin{equation}
\label{eq: ensemble average of ZZ}
    \mathbb{F}\{Z_1 Z_2\} \equiv   \underset{U \in \hat{\mathcal{U}}}{\int} Z[U]\, dU.
\end{equation}
Combining Eqs.~\eqref{eq: Uhat projection}, \eqref{eq: Z_U def}, and \eqref{eq: ensemble average of ZZ} yields
\begin{equation}
\label{eq: filtering ZZ}
   \mathbb{F}\{Z_1 Z_2\} = \tx{Tr}_{\mathcal{F}(\mathcal{H}_\odot \otimes \mathcal{H}_\odot)} ( e^{-\beta_1 H_1} e^{ - \beta_2 H_2} ).
\end{equation}
So the apparent ensemble averaging in Eq.~\eqref{eq: ensemble average of ZZ} is equivalent to the large-$N$ filtering in Eq.~\eqref{eq: filtering ZZ} \textit{without} the need of an ensemble. We similarly define the ensemble-averaged partition function $\mathbb{F}\{ Z_j \}$ for a single CFT
\begin{equation}
    \mathbb{F}\{Z_j \} \equiv \underset{U \in \mathcal{U}}{\int} Z_j [U] \, dU, 
\end{equation}
with the modified single CFT partition function defined by
\begin{equation}
    Z_j [U] \equiv \tx{Tr}_{ \mathcal{H}_\odot} (  e^{-\beta_j  H_j }U ),\quad j = 1,2.
\end{equation}
Eq.~\eqref{eq: one CFT filter} implies that $\mathbb{F}\{ Z_j \}$ is equal to the partition function of a single filtered CFT, 
\begin{equation}
\label{eq: single CFT filtered partition function}
    \mathbb{F}\{Z_j\} = \tx{Tr}_{\mathcal{F}\mathcal{H}_\odot} ( e^{-\beta_j H_j} ).
\end{equation}
Combining Eqs.~\eqref{eq: F(H otimes H) = FH otimes FH + H^wormhole},  \eqref{eq: filtering ZZ}, and \eqref{eq: single CFT filtered partition function} gives rise to
\begin{equation}
\label{eq: realizing large N filter}
    \mathbb{F}\{Z_1 Z_2\} = \mathbb{F}\{Z_1 \} \mathbb{F}\{ Z_2\} + \langle Z_1 Z_2\rangle_{\tx{wormhole}},
\end{equation}
with the wormhole term defined by
\begin{equation}
\label{eq: wormhole partition function def}
    \langle Z_1 Z_2\rangle_{\tx{wormhole}} \equiv \tx{Tr}_{\mathcal{H}^{\tx{wormhole}}} ( e^{-\beta_1 H_1} e^{ - \beta_2 H_2}).
\end{equation}
Eq.~\eqref{eq: single CFT filtered partition function} implies that the first term on the right hand side of Eq.~\eqref{eq: realizing large N filter} corresponds to two disconnected spacetimes.

According to the equivalence between gauging the nonlocal symmetry and filtering out the erratic observables, the constraint~\eqref{eq: fixing composite charges} is equivalent to the constraint in terms of the erratic observables,
\begin{equation}
    \hat{m}_{\tilde{\mathrm{e}}}= \hat{m}_{\tilde{\mathrm{f}}} = 0, \quad \hat{m}_{\widetilde{\mathrm{K}} } = 1.
\end{equation}
So the projection $\mathcal{F}$~\eqref{eq: Uhat projection} filters out the global erratic observables of the two CFTs. The group $\hat{\mathcal{U}}$ in the ensemble average~\eqref{eq: ensemble average of ZZ} hence can be equivalently replaced by the unitary group $\widetilde{\hat{\mathcal{U}}}$ generated by the corresponding erratic observables
\begin{equation}
    \widetilde{\hat{\mathcal{U}}} \equiv \{ \exp ( i \sum_{\alpha = \widetilde{\mathrm{E}}, \widetilde{\mathrm{F}},\widetilde{\mathrm{H}}} X^\alpha  \hat{m}_{\alpha}  );\, X^\alpha \in \mathbb{R} \},
\end{equation}
with
\begin{equation}
    \hat{m}_{\widetilde{\mathrm{E}}} \equiv \frac{\hat{m}_{\tilde{\mathrm{e}}}}{2\sin{\frac{ \pi k}{\hbar}}}, \quad \hat{m}_{\widetilde{\mathrm{F}}} \equiv \frac{\hat{m}_{\tilde{\mathrm{f}}}}{2\sin{\frac{ \pi k}{\hbar}}}, \quad \hat{m}_{\widetilde{\mathrm{H}}} \equiv i\log_{\tilde{q}} \hat{m}_{\widetilde{\mathrm{K}}}^2.
\end{equation}
Consequently, the ensemble averaging~\eqref{eq: ensemble average of ZZ} also washes out the net erratic observables in the product partition function. But the smooth remnant of correlations of the erratic observables between the two CFTs survives under the ensemble average~\eqref{eq: ensemble average of ZZ}. The wormhole contribution $\langle Z_1 Z_2 \rangle_{\tx{wormhole}}$ hence emerges and makes the filtered partition function $\mathbb{F}\{ Z_1 Z_2 \}$ non-factorized. Thus, Eq.~\eqref{eq: realizing large N filter} is a (2+1)-d gravity realization of Eq.~\eqref{eq: F( Z Z )}.

\section{Conclusions}

We showed that the observable algebra of one-sided boundary gravitons has a nontrivial center. In the Chern-Simons formulation of (2+1)-d gravity, this center is generated by Wilson loops. Based on asymptotic symmetries of boundary gravitons, we bootstrapped a general Poisson bracket~\eqref{eq: universal WW Poisson} to describe potential completions of the boundary graviton observable algebra. The simplest form of the the general Poisson bracket reproduces the observable algebra~\eqref{eq: simplest WW Poisson} of the Chern-Simons theory defined on a punctured disc. It is isomorphic to the observable algebra of the chiral WZNW model. This isomorphism identifies the Wilson lines in the Chern-Simons theory with the WZNW fields satisfying the same algebra. The observable algebra of the puncture monodromy and the boundary graviton observable algebra are commutants of each other. The puncture monodromy is hence interpreted as an effective description of one-sided black holes. 

Gauging the asymptotic symmetries yields an effective observable algebra~\eqref{eq: STS Poisson} of the one-sided black holes. The naive quantization of this effective black hole algebra is the quantum group algebra $U_q (\psl)$. The puncture monodromy with positivity restrictions~\eqref{eq: positivity m_e, m_f, m_K} is sufficient to describe an exterior region of the Lorentzian multi-boundary wormholes up to boundary graviton fluctuations. With the positivity restrictions, more elements emerge in the puncture monodromy observable algebra, which extends it to be the modular double $U_{q\tilde{q}} (\psl)$ of the $U_q (\psl)$. The extended monodromy algebra contains erratic observables that do not survive the large $N$ limit. 

The puncture monodromy encodes the charge~\eqref{eq: Q def PSL^*} corresponding to the large gauge transformations~\eqref{eq: puncture gauge transformation} at the puncture. They are nonlocal symmetries since they are not canonical transformations such that there is no associated local currents. At the quantum level, gauging the nonlocal symmetries is equivalent to filtering out the erratic observables. Consistently gauging the nonlocal symmetries requires a quantization condition~\eqref{eq: quantized k} on the coupling constant. For one CFT, the filtered subspace $\mathcal{F}\mathcal{H}_\odot$ is the Hilbert space of boundary gravitons associated with the global AdS$_3$ vacuum. For two CFTs, gauging the global part of the nonlocal symmetries preserves the smooth correlations between the erratic observables of the two CFTs such that a subspace $\mathcal{H}^{\tx{wormhole}}$~\eqref{eq: H^wormhole decomposition} describing wormholes survives. The filtered subspace $\mathcal{F}(\mathcal{H}_\odot \otimes \mathcal{H}_\odot)$ of the two CFTs decomposes into a direct sum~\eqref{eq: F(H otimes H) = FH otimes FH + H^wormhole} of the wormhole subspace $\mathcal{H}^{\tx{wormhole}}$ and the tensor product of two copies of the vacuum subspaces $\mathcal{F}\mathcal{H}_\odot$. Gauging the global nonlocal symmetries also filters the partition function of the two CFTs, which leads to an apparent ensemble averaged partition function $\mathbb{F}\{ Z_1Z_2 \}$~\eqref{eq: ensemble average of ZZ}. It decomposes into a sum~\eqref{eq: realizing large N filter} of a product term corresponding to two disconnected spacetimes and a non-factorized term~\eqref{eq: wormhole partition function def} corresponding to the wormhole contributions from the wormhole subspace $\mathcal{H}^{\tx{wormhole}}$.

\section*{Acknowledgments}
I thank Andreas Balaey, Ling-Yan Hung, Juan Maldacena, Thomas G. Mertens, Joan Sim\'on, Gabriel Wong, Lorenzo Russo, and Meng-Yang Zhang for helpful discussions. I thank Fabio Ori for suggestions on the paper title. I acknowledge financial support from the European Research Council (grant BHHQG-101040024). Funded by the European Union. Views and opinions expressed are however those of the author(s) only and do not necessarily reflect those of the European Union or the European Research Council. Neither the European Union nor the granting authority can be held responsible for them.

\appendix

\section{Two-Sided Wilson Line Poisson Bracket}
\label{app: Two-Sided Wilson Line Poisson Bracket}

Consider a Chern-Simons theory defined on a finite cylinder with two circle boundaries. Denote the coordinates on the two boundaries by $x$ and $y$ respectively. The Poisson bracket of Wilson lines connecting the two two boundaries is given by~\cite{Mertens:2025ydx}
\begin{equation}
\label{eq: Two-Sided Wilson Line Poisson Bracket}
    \{ \underset{1}{W}|_{x_1, y_1}, \underset{2}{W}|_{x_2, y_2} \} = \frac{ \pi}{k} \underset{1}{W}|_{x_1, y_1} \underset{2}{W}|_{x_2, y_2}   \underset{1}{W}|_{y_1, y_2} \underset{12}{K} \underset{1}{W}|_{y_2, y_1} \left( \tx{sgn}(x_{12}) - \tx{sgn}(y_{12})  \right),
\end{equation}
with $|x_{12}|, |y_{12}| < 2 \pi$. We used notations~\eqref{eq: W def}, \eqref{eq: very convenient notation}, \eqref{eq: x_12}, and \eqref{eq: sgn def}. Setting
\begin{equation}
    \underset{12}{r} = -  \underset{1}{W}|_{y_1, y_2} \,\underset{12}{K} \,\underset{1}{W}|_{y_2, y_1} \,\tx{sgn}(y_{12})  
\end{equation}
in Eq.~\eqref{eq: universal WW Poisson} yields Eq.~\eqref{eq: Two-Sided Wilson Line Poisson Bracket}.

\section{Symmetries and Global Charges}
\label{app: Nonabelian Momentum Map}

\subsection{Definitions of Symmetries}
\label{app: Definitions of Symmetries}

In this subsection, we prove the equivalence between Eq.~\eqref{eq: global charge Poisson} and Eq.~\eqref{eq: imath X Omega = U^-1 delta U} for general Lie group symmetries. We summarize the equivalence as Proposition \ref{prop: nonabelian momentum map}. 

\begin{proposition}
\label{prop: nonabelian momentum map}
       Suppose a Lie group $G$ acts on a symplectic manifold $\mathcal{P}$ equipped with a symplectic form $\Omega$. Given a function $\mathcal{Q}: M \times \mg \to \mathbb{C}$, denote by $\mathcal{Q}[\sigma]$ its value at $\sigma \in \mg$. Then
       \begin{equation}
       \label{eq: moment map}
\mathscr{L}_{\mathfrak{X}[\sigma]} \mathcal{O} =  \mathcal{Q}^{-1} \{\mathcal{Q}, \mathcal{O} \} [\sigma], \quad  \forall \sigma \in \mg,\, \mathcal{O} \in C^\infty (\mathcal{P}),
    \end{equation}
    if and only if
    \begin{equation}
    \label{eq: Q-1 delta Q, u = i Xu Omega}
        \imath_{
        \mathfrak{X}[\sigma]} \Omega = - \mathcal{Q}^{-1} \delta \mathcal{Q} [\sigma]   , \quad \forall \sigma \in \mg. 
    \end{equation}
\end{proposition}
\begin{proof}

Denote by $\mathfrak{X}_f$ the Hamiltonian vector field generated by a function $f \in C^\infty (M)$ associated with the symplectic form $\Omega$.

Given Eq.~\eqref{eq: moment map},
    \begin{align}
        \imath_{\mathfrak{X}_f}  \mathcal{Q}^{-1} \delta \mathcal{Q}[\sigma]
        &=
         \mathcal{Q}^{-1} \{ f, \mathcal{Q} \} [\sigma]
        \\
        &=
        \mathscr{L}_{\mathfrak{X}[\sigma]} f
        \\
        &=
        \imath_{\mathfrak{X}[\sigma] } \delta f
        \\
        &=
        \imath_{\mathfrak{X}[\sigma]} \imath_{\mathfrak{X}_f} \Omega
        \\
        &=
        -\imath_{\mathfrak{X}_f} \imath_{\mathfrak{X
        } [\sigma]} \Omega.
    \end{align}
The first equality holds by the definition of Poisson bracket. Eq.~\eqref{eq: moment map} implies the second equality. We used Cartan's magic formula~\eqref{eq: Cartan's magic formula} in the third equality. The forth equality holds by the definition of the Hamiltonian vector field. The anti-symmetry of the symplectic form implies the last equality. Then
\begin{equation}
    \imath_{\mathfrak{X}_f}  ( \mathcal{Q}^{-1} \delta \mathcal{Q}[\sigma] + \imath_{\mathfrak{X
        } [\sigma]} \Omega )= 0 .
\end{equation}
$f \in C^\infty (M)$ is arbitrary, so Eq.~\eqref{eq: Q-1 delta Q, u = i Xu Omega} holds.

    Given Eq.~\eqref{eq: Q-1 delta Q, u = i Xu Omega},
    \begin{align}
         \mathcal{Q}^{-1} \{ f, \mathcal{Q} \} [\sigma]
        &=
        \imath_{\mathfrak{X}_f}  \mathcal{Q}^{-1} \delta \mathcal{Q} [\sigma]
        \\
        &=
        -\imath_{\mathfrak{X}_f} \imath_{\mathfrak{X}[\sigma]} \Omega
        \\
        &=
        \imath_{\mathfrak{X}[\sigma]} \imath_{\mathfrak{X}_f} \Omega
        \\
        &=
        \imath_{\mathfrak{X}[\sigma] } \delta f
        \\
        &=
        \mathscr{L}_{\mathfrak{X}[\sigma]} f.
    \end{align}
The first equality holds by the definition of Poisson bracket. Eq.~\eqref{eq: Q-1 delta Q, u = i Xu Omega} implies the second equality. The anti-symmetry of the symplectic form implies the third equality. The forth equality holds by the definition of the Hamiltonian vector field. Cartan's magic formula~\eqref{eq: Cartan's magic formula} implies the last equality. $f \in C^\infty (M)$ is arbitrary, so Eq.~\eqref{eq: moment map} holds.
\end{proof}

In the proof of Proposition~\ref{prop: nonabelian momentum map}, we do not assume the existence of a local current, i.e. a codimension-1 form used to construct the charge. Since the left hand side of Eqs.~\eqref{eq: moment map} and \eqref{eq: Q-1 delta Q, u = i Xu Omega} are linear in $\sigma \in \mg$, the global charge $\mathcal{Q}$ is an exponential function on $\mg$. Equivalently, the global charge $\mathcal{Q}$ is $G^*$-valued with $G^*$ the Lie group corresponding to the dual $\mg^*$ of $\mg$. Such a global charge is called a nonabelian moment map or Lie-group valued moment map. The associated symmetry is called a Poisson-Lie symmetry in the literature~\cite{lu1991momentum, alekseev1998lie}. See also e.g. \cite{Babelon:2003qtg} for a review. See e.g.~\cite{Girelli:2025pyx} for recent progress.

The right-$\tx{PSL}(2, \mathbb{R})$ action \eqref{eq: puncture gauge transformation} is not an ordinary symmetry since it does not preserve the symplectic form \eqref{eq: Omega odot}. But it preserves the Poisson structure if $\tx{PSL}(2, \mathbb{R})$ itself is equipped with the Sklyanin Poisson structure \cite{sklyanin1979complete, sklyanin1982some, sklyanin1983some, faddeev1982integrable} which is the only type of Poisson structures that a semisimple Lie group admits \cite{Semenov-Tian-Shansky:1985mgd
,drinfeld1986quantum, cahen1990lie, vg1990hamiltonian}. The Sklyanin bracket is determined by the same $r$-matrix in Poisson bracket \eqref{eq: simplest WW Poisson}, so it is unique in the case of $\tx{PSL}(2, \mathbb{R})$. Ordinary symmetries can be generalized to Poisson-Lie symmetries which are maps preserving Poisson structures of the product of the phase space and the symmetry group. The right-$\tx{PSL}(2, \mathbb{R})$ action \eqref{eq: puncture gauge transformation} is then a Poisson-Lie symmetry \cite{Gawedzki:1990jc,Alekseev:1990vr,Alekseev:1991tx,Falcet:1991xt, Alekseev:1991wq}.

\subsection{Adjoint Action and the STS Decomposition}
\label{app: Infinitesimal Adjoint Action and the STS Decomposition}

In this subsection, we derive Eq.~\eqref{eq: infinitesimal adjoint action on m_-} and Eq.~\eqref{eq: infinitesimal adjoint action on m_+}.

Using Eqs.~\eqref{eq: K def} and Eq.~\eqref{eq: r_pm (sigma) def}, an element $\sigma \in \mg$ can be decomposed as
\begin{equation}
\label{eq: sigma = sigma_+ - sigma_-}
    \sigma = r_+(\sigma) - r_- (\sigma).
\end{equation}
Using Eq.~\eqref{eq: infinitesimal adjoint action on m}, we have
\begin{equation}
\label{eq: delta_(r_pm) m}
    \delta_{r_\pm(\sigma)} m = - \frac{2 \pi}{k} r_\pm (\sigma) m + \frac{2 \pi}{k} m r_\pm (\sigma).
\end{equation}
Using Eq.~\eqref{eq: Borel decomposition} and Eq.~\eqref{eq: delta_(r_pm) m}, it is straightforward to verify
\begin{align}
    \delta_{r_-(\sigma)} m_- &= \frac{2 \pi}{k} ( m_- r_-(\sigma) - r_-(m_+ \sigma m_+^{-1}) m_- ),
    \label{eq: delta_- m_-}
    \\
    \delta_{r_-(\sigma)} m_+ &= \frac{2 \pi}{k} ( m_+ r_+ ( \sigma ) - r_+ ( m_+ \sigma m_+^{-1} ) m_+ ),
    \label{eq: delta_- m_+}
    \\
    \delta_{r_+(\sigma)} m_- &= \frac{2 \pi}{k} ( m_- r_- ( \sigma ) - r_- ( m_- \sigma m_-^{-1} ) m_- ),
    \label{eq: delta_+ m_-}
    \\
    \delta_{r_+(\sigma)} m_+ &= \frac{2 \pi}{k} ( m_+ r_+(\sigma) - r_+(m_- \sigma m_-^{-1}) m_+ ).
    \label{eq: delta_+ m_+}
\end{align}
Combining Eqs.~\eqref{eq: sigma = sigma_+ - sigma_-}, \eqref{eq: delta_- m_-}, and \eqref{eq: delta_+ m_-} yields Eq.~\eqref{eq: infinitesimal adjoint action on m_-}. Combining Eqs.~\eqref{eq: sigma = sigma_+ - sigma_-}, \eqref{eq: delta_- m_+}, and \eqref{eq: delta_+ m_+} yields Eq.~\eqref{eq: infinitesimal adjoint action on m_+}.

\subsection{Global Charges of Nonlocal Symmetries}
\label{app: Global Charges of Nonlocal Symmetries}

In this subsection, we derive Eq.~\eqref{eq: imath X[sigma] Omega}.

Combining Eqs.~\eqref{eq: delta M = 0} and \eqref{eq: m = W_0^-1 M W_0} yields
\begin{equation}
    \delta ( W_0 m W_0^{-1} ) = 0,
\end{equation}
which is equivalent to
\begin{equation}
\label{eq: delta m = - W_0 delta W_0 m + ...}
    \delta m = - W_0^{-1} \delta W_0 m  + m W_0^{-1} \delta W_0 .
\end{equation}
Eq.~\eqref{eq: delta m = - W_0 delta W_0 m + ...} implies
\begin{equation}
\label{eq: Ad_m - Ad_m^-1}
    (\Ad_m - \Ad_{m}^{-1} ) (W_0^{-1} \delta W_0) = m^{-1} \delta m + \delta m m^{-1}.
\end{equation}
Via the puncture large gauge transformation~\eqref{eq: puncture gauge transformation on W_0}, $r_\pm(\sigma)$ acts on $W_0$ as
\begin{equation}
\label{eq: sigma_pm acting on W_0}
    \mathscr{L}_{\mathfrak{X}[r_\pm(\sigma)]} W_0 = \frac{2\pi}{k}W_0\, r_\pm (\sigma).
\end{equation}
Using Eq.~\eqref{eq: m = W_0^-1 M W_0}, the first term of the symplectic form $\Omega_\bullet$~\eqref{eq: Omega_bullet} can be written as
\begin{equation}
    \frac{k}{4 \pi}  \tr (  W_0^{-1} \delta W_0  m  W_0^{-1} \delta W_0 m^{-1} ).
\end{equation}
Using Eqs.~\eqref{eq: Ad_m - Ad_m^-1} and \eqref{eq: sigma_pm acting on W_0}, we have
\begin{equation}
\label{eq: super difficult}
    \frac{k}{4 \pi}\imath_{\mathfrak{X}[r_- (\sigma)]}   \tr (  W_0^{-1} \delta W_0  m  W_0^{-1} \delta W_0 m^{-1} ) =  \frac{1}{2} \tr ( r_- (\sigma) ( m^{-1} \delta m + \delta m m^{-1})  ).
\end{equation}
Using Eqs.~\eqref{eq: delta_- m_-} and \eqref{eq: delta_- m_+}, we have
\begin{align}
    &\quad\imath_{\mathfrak{X}[r_-(\sigma)]} \tr ( \delta m_- m_-^{-1} \delta m_+ m^{-1}_+ )
    \\
    &=
    \tr ( \frac{2 \pi}{k} ( m_- r_-(\sigma) - r_-(m_+ \sigma m_+^{-1}) m_- ) m_-^{-1} \delta m_+ m^{-1}_+ ) 
    \\
    &- \tr ( \delta m_- m_-^{-1} \frac{2 \pi}{k} ( m_+ r_\pm ( \sigma ) - r_\pm ( m_+ \sigma m_+^{-1} ) m_+ ) m^{-1}_+ )
    \\
    &=
    \tr ( \frac{2 \pi}{k} ( m_- r_-(\sigma) m_-^{-1} - r_-(m_+ \sigma m_+^{-1})  )  \delta m_+ m^{-1}_+ ) 
    \\
    &- \frac{2 \pi}{k} \tr ( \delta m_- m_-^{-1}  ( m_+ r_\pm ( \sigma )  m^{-1}_+ - r_\pm ( m_+ \sigma m_+^{-1} )  )  )
    \\
    &=
    \tr ( \frac{2 \pi}{k} ( m_- r_-(\sigma) m_-^{-1} - r_-(m_+ \sigma m_+^{-1})  )  \delta m_+ m^{-1}_+ ) 
    \\
    &- \frac{2 \pi}{k} \tr ( \delta m_- m_-^{-1}   m_+ r_- ( \sigma )  m^{-1}_+ - \delta m_+ m_+^{-1} r_+ ( m_+ \sigma m_+^{-1} )  )  )
    \\
    &=    \tr ( \frac{2 \pi}{k} ( m_- r_-(\sigma) m_-^{-1} + m_+ \sigma m_+^{-1}  )  \delta m_+ m^{-1}_+ - \delta m_- m_-^{-1}   m_+ r_- ( \sigma )  m^{-1}_+) 
    \\
    &=
    \frac{2 \pi}{k} \tr (  m_- r_-(\sigma) m_-^{-1} \delta m_+ m^{-1}_+ +  \sigma m_+^{-1} \delta m_+     - \delta m_- m_-^{-1}   m_+ r_- ( \sigma )  m^{-1}_+)
    \label{eq: tedious}
\end{align}
Using Eqs.~\eqref{eq: Borel decomposition}, \eqref{eq: Omega_bullet}, \eqref{eq: super difficult}, and \eqref{eq: tedious}, we have
\begin{equation}
\label{eq: r_- Omega_bullet}
    \imath_{\mathfrak{X}[r_-(\sigma)]} \Omega_\bullet = - \tr ( m^{-1}_+ \delta m_+  \sigma).
\end{equation}
Similarly, we have
\begin{equation}
\label{eq: r_+ Omega_bullet}
    \imath_{\mathfrak{X}[r_+(\sigma)]} \Omega_\bullet = - \tr (  m^{-1}_- \delta m_- \sigma ).
\end{equation}
Combining Eqs.~\eqref{eq: sigma = sigma_+ - sigma_-}, \eqref{eq: r_- Omega_bullet}, and \eqref{eq: r_+ Omega_bullet} proves Eq.~\eqref{eq: imath X[sigma] Omega}.

\bibliographystyle{unsrt}
\bibliography{ref.bib}

\end{CJK*}
\end{document}